\documentclass[final,12pt,notitlepage]{article}
\usepackage{amsfonts}
\usepackage{amsmath}
\usepackage{amssymb}
\usepackage{mathtools}
\usepackage{amsthm}
\usepackage[svgnames]{xcolor}
\usepackage[sc,labelsep=period,font=small]{caption}
\usepackage[margin=0.97in]{geometry}
\usepackage[bottom,multiple,splitrule]{footmisc}
\usepackage{graphicx}
\usepackage[authoryear,round,longnamesfirst]{natbib}
\usepackage[inline]{enumitem}
\usepackage[labelfont=rm]{subfig}
\usepackage{hyperref}
\usepackage[onehalfspacing]{setspace}
\usepackage{titlesec}
\usepackage{sectsty}
\usepackage{dsfont}
\usepackage{enumitem}
\usepackage{xargs}
\usepackage[normalem]{ulem}
\usepackage{comment}
\usepackage[capitalize,nameinlink]{cleveref}
\usepackage{bm}
\usepackage{tkz-euclide}
\usetikzlibrary{angles} 
\DeclareMathOperator{\prefto}{\succcurlyeq}
\DeclareMathOperator{\sprefto}{\succ}

\DeclareMathOperator*{\argmax}{arg\,max}
\newcommand{\CCP}{\mathcal C}
\newcommand{\policyspace}{X}

\newcommandx{\alitodo}[2][1=]{\todo[linecolor=blue,backgroundcolor=blue!25,bordercolor=blue,#1]{#2}}
\renewcommand{\equiv}{:=}
\newcommand{\asfav}{X_A^*}
\renewcommand{\Re}{\mathbb{R}}
\newcommand\smallO{
  \mathchoice
    {{\scriptstyle\mathcal{O}}}
    {{\scriptstyle\mathcal{O}}}
    {{\scriptscriptstyle\mathcal{O}}}
    {\scalebox{.7}{$\scriptscriptstyle\mathcal{O}$}}
  }

\hypersetup{
pdftitle={Who Controls the Agenda Controls the Polity},
pdfauthor={S. Nageeb Ali, B. Douglas Bernheim, Alexander W. Bloedel, and Silvia Console Battilana},
pdfkeywords={Agenda-setting, concentration of power}
}
\hypersetup{
colorlinks,breaklinks,naturalnames,
citecolor=DarkBlue,urlcolor=DarkBlue,linkcolor=DarkBlue
}

\bibpunct[, ]{(}{)}{;}{a}{}{,}

\usepackage{tikz}
\usetikzlibrary{patterns, decorations.pathreplacing, arrows.meta}
\usetikzlibrary{arrows}

\newtheorem{innercustomgeneric}{\customgenericname}
\providecommand{\customgenericname}{}
\newcommand{\newcustomtheorem}[2]{%
  \newenvironment{#1}[1]
  {%
   \renewcommand\customgenericname{#2}%
   \renewcommand\theinnercustomgeneric{##1}%
   \innercustomgeneric
  }
  {\endinnercustomgeneric}
}

\newcustomtheorem{customthm}{Theorem}

\theoremstyle{plain}
\newtheorem{theorem}{Theorem}

\newtheorem{claim}{Claim}
\newtheorem{fact}{Fact}

\newtheorem{corollary}{Corollary}
\newtheorem{lemma}{Lemma}

\newtheorem{definition}{Definition}

\newtheorem{example}{Example}
\theoremstyle{definition}
\theoremstyle{remark}

\crefrangeformat{theorem}{Theorems #3#1#4--#5#2#6}
\crefname{claim}{claim}{claims}
\crefname{fact}{fact}{facts}
\sectionfont{\color{DarkRed}}
\subsectionfont{\color{DarkRed}}
\subsubsectionfont{\color{DarkRed}}

\newcommand\posscite[1]{\citeauthor{#1}'s (\citeyear{#1})}

\usepackage{xcolor}
\usepackage{tikz}
\usetikzlibrary{patterns,hobby,snakes}
\usetikzlibrary{shapes.misc, positioning}

\begin{document}

\begin{titlepage}
\title{Who Controls the Agenda Controls the Polity\thanks{We thank Renee Bowen, Navin Kartik, David Levine, Elliot Lipnowski, Jay Lu, Steve Matthews, Ludvig Sinander, Alex Wolitzky, and various conference and seminar audiences for useful comments. We would also like to thank George Orwell, who recognized the importance of agenda control (albeit in a somewhat different setting), for inspiring our title. We gratefully acknowledge financial support from NSF Grants SES-0137129 (Bernheim) and SES-1530639 (Ali \& Bernheim). Bloedel gratefully acknowledges the hospitality and financial support of Caltech and Stanford, where parts of this research were completed.}\vspace{.2in}
}  

\author{
{ S. Nageeb Ali}\thanks{Department of Economics, Pennsylvania State University. Email: \href{mailto:nageeb@psu.edu}{nageeb@psu.edu}.}
\and 
{B. Douglas Bernheim}\thanks{Department of Economics, Stanford University. Email: \href{mailto:bernheim@stanford.edu}{bernheim@stanford.edu}.}
\and 
{Alexander W. Bloedel}\thanks{Department of Economics, UCLA. Email: \href{mailto:abloedel@econ.ucla.edu}{abloedel@econ.ucla.edu}.}
\and  
{Silvia Console Battilana}\thanks{Auctionomics. Email: \href{mailto:silviacb@auctionomics.com}{silviacb@auctionomics.com}.}
\vspace{.2in}}
\date{December 1, 2022}
\maketitle

\begin{abstract}

This paper models legislative decision-making with an agenda setter who can propose policies sequentially, tailoring each proposal to the status quo that prevails after prior votes. Voters are sophisticated and the agenda setter cannot commit to her future proposals. Nevertheless, the agenda setter obtains her favorite outcome in every equilibrium regardless of the initial default policy. Central to our results is a new condition on preferences, \emph{manipulability}, that holds in rich policy spaces, including spatial settings and distribution problems. Our results overturn the conventional wisdom that voter sophistication alone constrains an agenda setter's power.

\end{abstract}

\thispagestyle{empty} 
\vspace{0.4in}

\medskip

\end{titlepage}

\setcounter{page}{1}

\section{Introduction}\label{Section-Introduction}

A central goal of democratic institutions is to balance competing interests by distributing political power evenly among society's members. A recurring concern is that power is often highly concentrated in the hands of a select few leaders, who steer policy in their favor by manipulating institutional procedures. Our objective is to understand how seemingly pro-democratic institutions can yield this concentration of power.

Legislative institutions merit particular attention in this regard because of the central role they typically play in determining policy. Although legislative procedures vary, they often involve (i) a single agenda setter (e.g., committee chair or party leader) who acts as a gatekeeper for proposals, and (ii) a group of legislators who vote to approve or reject proposals. While the agenda setter controls which policies come up for a vote, she cannot unilaterally dictate policy because passage of her proposals requires majority support. Even so, the degree to which voting constrains the agenda setter is unclear. 

An important literature in political economy seeks to understand the power that flows from agenda control. In a seminal contribution, \cite{mckelvey1976intransitivities} observes that, in rich policy settings, the agenda setter can exploit cycles in the majority relation to obtain \emph{any} desired outcome by appropriately sequencing proposals, and is therefore unconstrained by the majority's will. However, this striking conclusion assumes legislators are non-strategic and vote myopically, without accounting for subsequent modifications of the policy. Sophisticated legislators, who have learned from experience, anticipate the paths of proposals and votes, and accept or reject proposals based on the final outcomes to which those paths lead. Subsequent research concludes that the requirement of majority approval constrains the agenda setter's power when voters are sophisticated. Specifically, \cite{shepsle1984uncovered} show that the agenda setter can only achieve policies that are not \emph{covered} by the initial default option.\footnote{A policy $x$ \emph{covers} $y$ if a majority of voters strictly prefer $x$ to $y$, and every policy that is majority-preferred to $x$ is also majority preferred to $y$. \label{footnote-covering}} In rich policy settings, this constraint often is stringent \citep{mckelvey1986covering}. 

 The prior literature generally assumes the agenda setter must commit to a sequence of proposals, which is set in stone regardless of which proposals pass. And yet, in many settings, nothing prevents the agenda setter from bringing different proposals to the floor depending on how previous votes turned out. The impact of the fixed-agenda assumption is complex. On the one hand, it attenuates the power of agenda setters because it precludes them from tailoring their proposals to the prevailing circumstances. On the other hand, it magnifies their power because it endows them with the ability to commit to proposals they might not want to make when the time arrives. 
 Relative to this benchmark, it is unclear how the benefits that real-world agenda setters accrue from being able to flexibly tailor proposals stack up against the costs they incur from being unable to commit.
 This is the question we address.

We study a novel model of \emph{real-time agenda control}. A single agenda setter and a group of voters choose a policy using the following procedure: for each of finitely many rounds, (i) the agenda setter can propose an alternative to the contemporaneous default policy (where the first-round default policy is given exogenously), and (ii) a vote is taken to determine whether the current proposal or the contemporaneous default policy will serve as the default policy in the next round. The policy that prevails in the final vote (for the terminal round) is implemented and determines payoffs. We investigate the subgame perfect equilibria of this game, assuming neither voters nor the agenda setter can commit to their future decisions. 

Real-time agenda control has stark implications for collective choice problems that satisfy the following property: a collective choice problem is \emph{manipulable} if, for every policy $x$ other than the agenda setter's favorite, there is an alternative policy $y$ that both she and a majority of voters strictly prefer to $x$. We show that canonical formulations of spatial and distributive politics satisfy manipulability. More precisely, the standard spatial model almost-surely satisfies manipulability if the policy space has three or more dimensions. We also define a broad class of ``distribution problems'' that satisfy manipulability. This class includes divide-the-dollar problems, as well as problems that mix non-zero-sum policies with transfers. Indeed, augmenting any collective choice problem with pork or transfers makes that problem manipulable. Our favored interpretation of manipulability, and its prevalence, is that the discordance of majority will in multidimensional problems inevitably creates opportunities for the agenda setter. 

We establish the potency of real-time agenda control when the environment is manipulable. Our main finding, stated informally, is as follows:
\begin{quote}
\textbf{Main Result.} \emph{If there are sufficiently many rounds, the agenda setter obtains her favorite policy in every equilibrium regardless of the initial default policy if and only if the collective choice problem is manipulable.}
\end{quote}
Thus, for a wide range of collective choice problems, the agenda setter effectively dictates policy, despite voters' sophistication and her lack of commitment. Manipulability is necessary and sufficient for this conclusion; absent manipulability, we show that equilibrium outcomes sometimes remain bounded away from the agenda setter's favorite. 

\Cref{Theorem-GenericFiniteAlternatives,Theorem-epsilongrid,Theorem-Capricious} formalize our main conclusion under a range of technical conditions, including for finite and continuous policy spaces, and clarify how many rounds are ``sufficient.'' The core argument involves a simple observation concerning ``one-step'' improvements: if the default option going into the terminal round is not the agenda setter's favorite, she will propose her favorite policy among those that both she and a majority of voters prefer to the default. Manipulability guarantees that such an improvement exists. Applying this logic iteratively implies that, if she can make proposals for $t$ rounds, she can obtain the outcome generated by the $t$-fold iteration of this ``favorite improvement'' operator. At each stage, voters pass proposals that lead to this outcome because a majority prefer it to the outcome that would emerge otherwise. When $t$ is sufficiently large, this iterative process yields an policy arbitrarily close to (if not exactly the same as) the agenda setter's favorite.

Though simple, this logic is extremely general. It applies for voting rules other than simple majority, so long as the analog of manipulability holds. The same conclusion also holds for other widely studied legislative procedures, such as the \emph{closed-rule} or \emph{successive procedure} as well as \emph{open-rule} bargaining. More broadly, we obtain a \emph{protocol-equivalence} result for the class of ``generalized amendment'' protocols: fixing a preference profile and voting rule, all of these protocols (and others) yield the same equilibrium outcome under real-time agenda control.

Our findings thereby illuminate the forces that contribute to the concentration of political power, and explain why voter sophistication may not be an effective safeguard against agenda control. Recent empirical findings highlight similar themes: \cite{berry2016cardinals,berry2018congressional} observe that chairs of congressional committees have disproportionate influence on policymaking, and \cite{fouirnaies2018agenda} finds that special interest groups make greater campaign contributions to legislators endowed with procedural authority. 

Our analysis has the additional implication that agenda setters benefit from bundling policy choices with transfers and pork. Augmenting any collective choice problem with transfers renders it manipulable not only for simple majority rule, but also for any voting rule that does not provide any individual with veto power. Our main results then imply that the bundling strategy allows the agenda setter to obtain her favorite policy for any ``veto-proof'' voting rule.\footnote{Previous studies have highlighted the detrimental effects of pork on legislative and democratic politics \citep[e.g.][]{lizzeri2001provision,battaglini2008dynamic,maskin2019pandering}. Our analysis shows that the mere ability to use pork or transfers yields dictatorial power; in equilibrium, the agenda setter does not actually transfer benefits to any party.} Moreover, we show that in these settings, the dictatorship result holds even when the process involves a relatively small number of rounds: for simple majority rule, three rounds suffice, and for general ``veto-proof'' voting rules, the number of rounds need not exceed the number of voters.
Analogously, we find that the agenda setter may benefit from linking policy issues in order to make bargaining more multidimensional. Specifically, our analysis of spatial politics shows that the collective choice problem is generically manipulable if the policy space has three or more dimensions, but typically fails to be so otherwise. Thus, if the current legislative debate concerns a one- or two-dimensional policy decision, the agenda setter benefits from bundling that decision with other policy issues---even ``settled'' ones for which the default option already coincides with her favorite policy---because the overall problem thereby becomes manipulable.

To isolate the role of sequential rationality constraints in real-time agenda control, we compare our model to the commitment benchmark in which the agenda setter can commit to any strategy in the dynamic game. Therein, we find that the agenda setter can obtain her favorite policy among those that are reachable from the initial default option through a finite chain of majority improvements, mirroring \posscite{miller1977graph} classic characterization of outcomes achievable through general binary  voting trees. Relative to this benchmark, our main results show that manipulability not only enables the agenda setter to attain her commitment payoff without having to commit, but also guarantees that the commitment outcome coincides with her favorite policy. Absent manipulability, sequential rationality typically precludes her from achieving the commitment benchmark. In such cases, even a commitment to a fixed agenda, as in \cite{shepsle1984uncovered}, may leave her better off.

As in the prior literature, we assume the agenda process involves a finite number of rounds. This assumption is appropriate when the purpose of negotiation is to solve a time-indexed collective choice problem, i.e., to select the policy that prevails at a given point in time. Problems of this form are ubiquitous. For example, when a legislature negotiates over the budget for a given fiscal year, it cannot continue those negotiations into the subsequent fiscal year.
In such cases, there is both a \emph{deadline} for meaningful deliberations (e.g., 11:59pm on December 31), and an inherent constraint on the speed at which the legislature can consider a proposal. In combination, these considerations imply that the number of rounds is necessarily bounded. 
We assume, for the sake of tractability, that this number is known in advance, but our results plainly extend to settings in which an initially unknown termination point becomes evident during the course of negotiations. 

While we think it is reasonable to assume the existence of a deadline---insofar as negotiations over time-indexed actions are widespread---one could alternatively consider processes that allow negotiations to continue indefinitely. \cite{diermeier2011legislative} and \cite{anesi2014bargaining} adopt that approach, analyzing infinite-horizon counterparts of our baseline model. They show that the limitless potential for reconsideration severely constrains agenda power. Read in the context of that work, our findings establish the critical importance of deadlines. 
We show, in effect, that a simple commitment to termination at a fixed point in time allows an agenda setter to achieve her favorite outcome (in manipulable environments).
More broadly, we find that an agenda setter's preference over the length of negotiations is non-monotonic: she prefers a moderate number of proposal rounds both to a single round and to an an open-ended process with no limit on duration. Thus, our analysis implies that even if there is no natural deadline, a strategic agenda setter benefits from inventing excuses to establish one. 

 We are not the first to show that certain collective choice processes can produce dictatorial outcomes. \cite{kalandrakis2004three} finds that bargaining with an endogenous status quo and changing proposers yields such results in a divide-the-dollar setting. \cite{bernheim2006power} analyze a model of pork-barrel politics with changing proposers, and show that the final proposer is effectively the dictator. The endogenous evolution of the default option is an essential feature of those frameworks. A few papers conclude that dictatorial power prevails under the opposite assumption of closed-rule negotiations, where accepting an offer results in its immediate implementation. \cite{ali2019predictability} find that a modest form of predictability about future bargaining power results in the first proposer obtaining the entire surplus in a closed-rule divide-the-dollar setting; \cite{duggan2018extreme} consider settings in which a single agent makes all proposals, and show that she has approximate dictatorial power.

These prior studies obtain results for specific policy spaces and legislative procedures, building primarily on the legislative bargaining literature. Our work differs in several important respects. It is instead rooted in the classical literature on agenda setting, to which we contribute by investigating the implications of real-time agenda control without commitment. Instead of focusing on a particular policy space, we identify manipulability---a property that any given space may or may not satisfy---as the necessary and sufficient condition for dictatorial power, and we also show that canonical models of spatial and distributive politics have this property. Moreover, we demonstrate that the strategic logic behind our main result is robust, in that it applies to a wide range of legislative procedures, allowing for either evolving or fixed default options.

\Cref{Section-Examples} illustrates the core logic of our results through a simple example. \Cref{Section-Model} describes the general model, and \Cref{Section-MainResults} contains our main results. \Cref{Section-ManipulableCCPs} explains why manipulability holds in spatial and distributive politics. \Cref{Section-Robustness} describes the commitment benchmark, investigates the implications of real-time agenda control for other legislative procedures, and elaborates on the role of deadlines.  \Cref{Section-Conclusion} concludes. All omitted proofs are in Appendices.

\section{An Example}\label{Section-Examples}

A legislature, comprised of an agenda setter and $n$ voters (where $n$ is odd), chooses a policy from $\{w,x,y,z\}$. The agenda setter has a strict preference relation that coincides with the listed order. Each voter has a complete and transitive preference relation that is also strict, but the profile of voter preferences results in a strict majority relation $\succ_M$ that cycles.\footnote{This majority relation can arise whenever $n\geq 3$. For the $3$-voter case, suppose $z\succ_1 w\succ_1 x \succ_1 y$, $x\succ_2 y\succ_2 z \succ_2 w$, and $y\succ_3 z\succ_3 w \succ_3 x$.\label{Footnote-VoterPref}} We depict the agenda setter's preferences and the majority relation in \Cref{Figure-IntroFiniteExample}. 

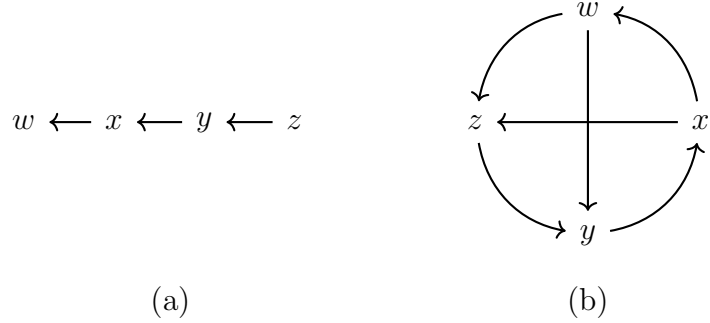
\begin{figure}[t]
	\centering  
	\begin{tikzpicture}[domain=0:3, scale=3,thick,shorten > = 1pt, shorten <= 1pt]
		\path (0.0,1) node (w1){$w$};
		\path (0.4,1) node (x1){$x$};
		\path (0.8,1) node (y1){$y$};
		\path (1.2,1) node (z1){$z$};
		
		\draw[<-] (w1)->(x1);
		\draw[<-] (x1)--(y1);
		\draw[<-] (y1)--(z1);
		\path (2.5,1.5) node (w2){$w$};
		\path (3,1) node (x2){$x$};
		\path (2.5,0.5) node (y2){$y$};
		\path (2,1) node (z2){$z$};
		\draw[<-] (z2) to[out=80,in=190] (w2);
		\draw[<-] (w2) to[out=-10,in=100] (x2);
		\draw[<-] (x2) to[out=260,in=10] (y2);
		\draw[<-] (y2) to[out=170,in=280] (z2);
		\draw[->] (w2) to (y2);
		\draw[->] (x2) to (z2);
		\draw (0.65,0.2)node{(a)};
		\draw (2.5,0.2)node{(b)};
	\end{tikzpicture}  
	\caption{Panel (a) shows the agenda setter's preferences. Panel (b) shows the majority relation. In each case, an arrow from policy $p$ to $p'$ denotes $p'\succ p$. }\label{Figure-IntroFiniteExample}  
\end{figure}

The legislature selects a policy through what is known in the literature as the \emph{amendment procedure}, which operates as follows:
\begin{enumerate}[noitemsep,label={\normalfont (\roman*)}]
    \item There is an initial default option;
    \item In each of finitely many rounds, the agenda setter proposes a new policy (the \emph{proposal}), 
    which is put to a vote against the prevailing default;
    \item In each non-terminal round, the policy that obtains a majority of the votes becomes the default for the subsequent round;
    \item The policy that obtains a majority of the votes in the final round is implemented.
\end{enumerate}

The amendment procedure features prominently in practice and is the focus of considerable prior work.
In much of the literature, the agenda setter lays out the sequence of proposals in advance, prior to any voting. We call this procedure a \emph{fixed agenda} protocol because it does not permit the agenda setter to vary her proposals based on the concurrent default or prior votes. Our analysis contrasts this protocol with \emph{real-time agenda control}, which allows the agenda setter to tailor proposals to the circumstances that arise, but does not endow her with any commitment power. We use this example to illustrate the distinct implications of fixed agenda protocols and real-time agenda control.

Suppose the initial default is $z$, the agenda setter's least favorite policy. \cite{mckelvey1976intransitivities} points out that if voters are myopic, the agenda setter can obtain her favorite policy $w$ by exploiting the cycles in the majority relation: she uses a fixed agenda where $y$ is the first proposal, $x$ is the second proposal, and $w$ is the third and final proposal. Because voters are myopic, in each instance they anticipate no further revisions, so each proposal passes, and the process selects $w$. \cite{shepsle1984uncovered} show that this conclusion does not hold if voters are sophisticated. Instead, the agenda setter can obtain only those policies that are \emph{not covered} by the initial default option (as defined in \Cref{footnote-covering}). In our example, the agenda setter's favorite uncovered policy is $x$, which she can obtain with the following fixed agenda: propose $x$ in the first round and $y$ is the second.\footnote{Although a majority of voters prefer $z$ to $x$, sophisticated voters anticipate that rejection of $x$ in the first round would lead to a final outcome of $y$, as $y$ is majority preferred to $z$.}
Voter sophistication would therefore appear to limit the power of agenda control.

Our central insight is that giving the agenda setter the flexibility to make proposals in real-time, so that she can tailor each proposal to the prevailing default option, unleashes the full power of agenda control and allows her to obtain her favorite policy even if voters are sophisticated.  While it is intuitive that the agenda setter benefits from greater flexibility, note that we simultaneously remove her ability to commit, which could in principle limit her power by introducing sequential rationality constraints. 

To illustrate how the agenda setter can exploit real-time agenda control, we construct an equilibrium for a $3$-round game that selects policy $w$. Consider the following strategy for the agenda setter: if the default option in any round is policy $p$, she proposes her \emph{favorite improvement} to $p$---in other words, her favorite policy among those the majority prefers to $p$. %
We use $\phi(p)$ to denote this policy. On the equilibrium path (starting from an initial default of $z$), this strategy prescribes proposing $y$ first, then $x$, and then $w$. Notice that this sequence coincides with the optimal agenda for myopic voters. But in this instance, voters approve each policy not out of myopia, but rather because they (correctly) anticipate future play. We can verify this claim through backward induction:
\begin{description}
    \item[{$\mathbf{t=3}$}:] If the default option is $p$, the agenda setter proposes $\phi(p)$, which, by construction, results in $\phi(p)$. 
    \item[{$\mathbf{t=2}$}:] If the default option is $p$, then the proposer proposes $\phi(p)$. Anticipating the behavior at $t=3$, voters understand that approving this policy today ultimately results in $\phi^2(p)$---the two-fold iteration of the $\phi$ operator---whereas rejecting this policy results in $\phi(p)$. Since a majority of voters prefer $\phi^2(p)$ to $\phi(p)$, the proposal passes.
    \item[{$\mathbf{t=1}$}:] Analogously, in the first period, the agenda setter proposes $y=\phi(z)$. Voters anticipate that approving this proposal ultimately results in $w=\phi^2(y)$, whereas rejecting it ultimately results in $x=\phi^2(z)$. Since a majority favor $w$ over $x$, the proposal passes. 
\end{description}
Thus, a majority of voters always finds it sequentially rational at each stage to approve the proposal this strategy prescribes.

Because the agenda setter cannot make commitments, her behavior must also be sequentially rational. Indeed, in the final round, she proposes her favorite option among those that will pass. Given the equilibrium for the final round, her second-round proposal always achieves her favorite outcome among the feasible alternatives. Likewise, given the equilibrium for the last two rounds, she cannot improve on her prescribed first-round proposal. Therefore, no deviation can make her strictly better off. 

Thus, there is a subgame perfect equilibrium in which the agenda setter obtains $w$. Our main result (\cref{Theorem-GenericFiniteAlternatives}) reaches a stronger conclusion: even though there are multiple equilibria, $w$ is the \emph{unique} subgame perfect equilibrium outcome regardless of the initial default so long as there are three or more rounds.\footnote{We impose the standard refinement that voters vote as if they are pivotal. } Real-time agenda control therefore guarantees that the group will select the agenda setter's favorite policy. 

In this example, the agenda setter's preferences and the majority relation jointly satisfy a condition we call \emph{manipulability}: for every policy $p$ other than the agenda setter's favorite, there is a policy $p'$ that both she and a majority of voters strictly prefer to $p$. Our main results show that the agenda setter exercises dictatorial power if and only if this condition is satisfied: when it fails, then for some initial default options, the agenda setter cannot obtain her favorite policy in any equilibrium. 

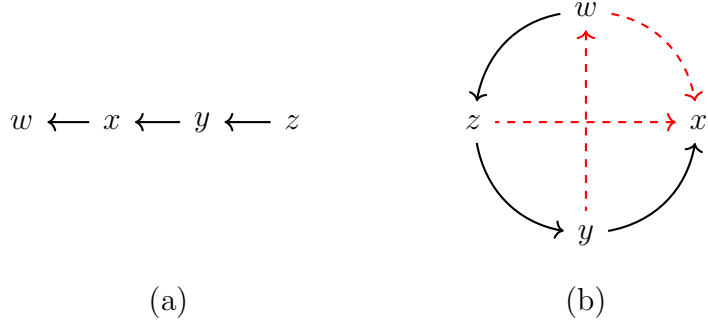
\begin{figure}[t]
	\centering  
	\begin{tikzpicture}[domain=0:3, scale=3,thick,shorten > = 1pt, shorten <= 1pt]
		\path (0.0,1) node (w1){$w$};
		\path (0.4,1) node (x1){$x$};
		\path (0.8,1) node (y1){$y$};
		\path (1.2,1) node (z1){$z$};
		
		\draw[<-] (w1)->(x1);
		\draw[<-] (x1)--(y1);
		\draw[<-] (y1)--(z1);
		\path (2.5,1.5) node (w2){$w$};
		\path (3,1) node (x2){$x$};
		\path (2.5,0.5) node (y2){$y$};
		\path (2,1) node (z2){$z$};
		\draw[<-] (z2) to[out=80,in=190] (w2);
		\draw[->,dashed,red] (w2) to[out=-10,in=100] (x2);
		\draw[<-] (x2) to[out=260,in=10] (y2);
		\draw[<-] (y2) to[out=170,in=280] (z2);
		\draw[<-,dashed,red] (w2) to (y2);
		\draw[<-,dashed,red] (x2) to (z2);
		\draw (0.65,0.2)node{(a)};
		\draw (2.5,0.2)node{(b)};
	\end{tikzpicture}  
	\caption{Panel (a) shows the agenda setter's preferences. Panel (b) shows the majority relation, with red dashed arrows denoting differences from that in \Cref{Figure-IntroFiniteExample}. }\label{Figure-FiniteExample-2} 
\end{figure}
To understand why manipulability is necessary, consider the majority relation in \cref{Figure-FiniteExample-2}. The solid black arrows are the same as before, but the red dashed arrows are different. Policy $x$ is now \emph{unimprovable}: there is no other policy that the agenda setter and a majority of voters all prefer to $x$. As $x$ is not the agenda setter's favorite option, manipulability fails. Our characterization result (\Cref{lemma:structural}) implies that if the initial default option is $z$, the agenda setter necessarily obtains $x$ in every equilibrium regardless of the horizon. Intuitively, voters anticipate that if $z$ remains the default option in the terminal round, sequential rationality will compel the agenda setter to propose $x$, because $x$ is her favorite policy among the options that will pass. But then $x$ must also be the outcome of a two stage game starting with a default of $z$: rejecting the first proposal leads to $x$, and $x$ is unimprovable, so the majority will not support any proposal leading to an option the agenda setter would prefer. The same argument applies, recursively, to games of any length.

In this example, the agenda setter is stymied by her inability to commit: were she able to lock in a fixed agenda, as in \cite{shepsle1984uncovered}, she could achieve $w$ by proposing $w$ in the first round and $y$ in the second. A majority of voters would then approve $w$ in the first round because rejection would yield $y$. In our setting, the agenda setter cannot achieve this outcome because proposing $y$ in the second round is not sequentially rational: if $z$ remains the default option in the second round, she would instead propose $x$, and anticipating that behavior, voters would be unwilling to approve $w$ in the first round. 

These examples illustrate the role of manipulability in empowering the agenda setter to obtain her favorite policy without the need for commitment. Although manipulability may appear restrictive, we show in \Cref{Section-ManipulableCCPs} that it is satisfied in standard models of spatial and distributive politics.

\section{Model}\label{Section-Model}

Our model consists of two components: (i) a (static) collective choice problem comprising the set of feasible policies and agents' preferences over them, and (ii) a dynamic procedure for selecting a policy. We describe each in turn.

 \paragraph{Collective Choice Problem.} A group $N \equiv \left\{1, \dots, n\right\}$ of \emph{voters} (where $n$ is odd) and a single non-voting \emph{agenda setter} ($A$) choose a policy from \emph{policy space} $X$. This space is compact and metrizable; in most of our examples, it is either finite or a subset of a finite-dimensional Euclidean space. For each $i\in \{1,\ldots,n,A\}$, $\prefto_i$ denotes player $i$'s preference relation over policies. Each relation is continuous and has a continuous utility representation $u_i: X \to \mathbb{R}$. If a majority of voters weakly (resp. strictly) prefers $x$ to $y$, we say that $x \prefto_M y$ (resp., $x \succ_M y$).  We use $X_A^*\equiv \argmax_{x\in X} u_A(x)$
 to denote the set of the agenda setter's favorite policies. Together, the policy space and preference profile constitute a \emph{Collective Choice Problem}, $\CCP:=(X,\{\prefto_i\}_{i=1,\ldots,n,A})$.

\paragraph{Legislative Procedure.} Our baseline analysis focuses on what the literature calls the \emph{amendment procedure}. Voting takes place in rounds $t \in \left\{1, \dots, T \right\}$, where $T$ is finite. Activity prior to round $t$ determines a \emph{default} policy $x^{t-1}$. The initial default, $x^0$, is exogenous. In each round $t$, the agenda setter proposes a policy (the \emph{proposal}) denoted $a^t \in X$, which can coincide with the existing default policy. The proposal $a^t$ is then put to a vote against the default $x^{t-1}$. If a majority of voters vote in favor of the proposal (i.e., it ``passes''), then it becomes the new default for the subsequent round: $x^{t} = a^t$. If the proposal does not pass, the default remains unchanged: $x^{t} = x^{t-1}$. The policy $x^{T}$ that prevails (after voting) in round $T$ determines payoffs.\footnote{Formally, in the pertinent literature \citep[e.g.][]{shepsle1984uncovered}, this bargaining framework is known as a ``forward agenda'' amendment procedure. The literature also considers ``backward agenda'' procedures wherein, after all amendments have been incorporated, the amended bill is put to a final up-or-down vote against the original default. Our analysis also applies to settings in which the legislature can consider a sequence of bills, each amendable through a backward agenda procedure, prior to the deadline (i.e., the date at which the policy that the bills concern is to take effect). Under this interpretation, each ``round'' of the procedure we study in this paper pertains to a distinct bill that, if passed, will be implemented unless it is subsequently supplanted by passage of another bill.}

\paragraph{Solution Concept.} All players can condition their actions, both proposals and votes, on the history of prior actions. This assumption captures the idea that people take actions---and, in particular, set the agenda---in real time. A \emph{history} $h^t$ as of the beginning of round $t$ records the initial default policy $x^0$, the sequence of proposals $(a^1,\ldots,a^{t-1})$, and the sequence of voting profiles $(v^1,\ldots,v^{t-1})$ in all prior rounds. It therefore identifies the default $x^{t-1}$ prevailing at the beginning of round $t$. $\mathcal{H}^t$  denotes the space of all round-$t$ histories. A strategy for the agenda setter is a mapping $\sigma_A : \cup_{t=1}^{T} \mathcal{H}^t  \to \Delta(X)$ specifying, for each history $h^t$, a distribution $\sigma_A(h^t) \in \Delta(X)$ over proposals $a^t$. 
A strategy for voter $i$ is a mapping $\sigma_i : \cup_{t=1}^{T} \mathcal{H}^t \times X \to \Delta\left(\left\{y, n \right\}\right)$ specifying, for each history $h^t$, a distribution over \emph{y}es or \emph{n}o votes for each potential proposal $a^t$. 

We study subgame perfect equilibria of this game. 
We also assume ``as-if pivotal'' voting: if passage (resp. rejection) of the current proposal ultimately leads to continuation outcome $x$ (resp. $y$), then anyone who has a \emph{strict} preference for $x$ votes for the option that leads to $x$, and similarly for $y$.\footnote{There is no restriction on the behavior of voters who are indifferent between $x$ and $y$. This definition applies only at histories where continuation outcomes do not depend on the composition of the current vote, conditional on which proposal prevails. As will become apparent, our analysis either assumes strict preferences (\cref{Theorem-GenericFiniteAlternatives,Theorem-epsilongrid}) or allows for indifference while imposing a mild refinement (\cref{Theorem-Capricious}), in each case thereby ensuring that such equilibria exist.} 
This standard assumption rules out unreasonable equilibria in which nonpivotal voters, who are technically indifferent because they cannot affect the outcome, vote contrary to their preferences.\footnote{When voters have strict preferences, subgame perfection with ``as-if pivotal'' (simultaneous) voting is outcome-equivalent to both (i) iterated deletion of weakly dominated strategies under simultaneous voting and (ii) mere subgame perfection when voting in each round occurs via ``roll call'' in a fixed sequential order (see Chapter 4 of \citealt{austen2005positive} and references therein).} Henceforth, we use the term \emph{equilibrium} to denote this solution concept. 

\section{The Power of Real-Time Agenda Control}\label{Section-MainResults}

We now turn to our main results concerning the agenda setter's power. \Cref{subsection:weakalignment} defines what it means for a collective choice problem to be manipulable. Sections \ref{subsection:finite} and \Cref{subsection:infinite} consider finite and general policy spaces, respectively. 

\subsection{Improvability and Manipulability}\label{subsection:weakalignment}

We begin by describing policies that the agenda setter can improve upon with a single proposal round. 
\begin{definition}\label{defi:Improv}
Policy $x$ is {\textcolor{DarkBlue}{Improvable}} if there exists a policy $y$ such that $y\succ_A x$ and $y\succ_M x$; 
if no such policy exists, then policy $x$ is {\textcolor{DarkBlue}{Unimprovable}}. 
\end{definition}
One can view the set of unimprovable policies as the \emph{core} of a suitably defined cooperative game in which all decisive coalitions contain both the agenda setter and at least a simple majority of voters.

The set of unimprovable policies must include all of the agenda setter's favorites, $\asfav$. For an important class of collective choice problems, everything else is improvable.
\begin{definition}\label{defi:Manip}\hypertarget{Definition:Manipulable}
A collective choice problem $\CCP$ is {\hyperlink{Definition:Manipulable}{Manipulable}} if every $x \not\in X^*_A$ is improvable.
\end{definition}
Manipulability is connected to intransitivity of the majority relation. If that relation is transitive, 
a Condorcet winner exists, and the choice problem is manipulable if and only if that policy is the agenda setter's favorite. Intransitivities make it easier for the agenda setter to find improvements that the majority will accept.\footnote{Manipulability is, however, distinct from the notion of global intransitivity in the majority relation (or ``chaos'') studied by the classical literature \citep[e.g.][]{mckelvey1976intransitivities}. Global intransitivity stipulates that for any two policies $x$ and $y$, there is a majority chain $\{a^k\}_{k=0}^{k=K}$ such that $x=a^0$, $y=a^K$, and $a^k\succ_M a^{k-1}$ for all $k\in \{1,\ldots,K\}$. However, manipulability plainly does not require global intransitivity and, unlike manipulability, global intransivity does not depend on the agenda setter's preferences. \label{Footnote-GlobalIntransitivity}} 

In the rest of this section, we identify the implications of real-time agenda control for manipulable and non-manipulable collective choice problems. In \Cref{Section-ManipulableCCPs}, we demonstrate that collective choice problems belonging to some familiar and important classes are manipulable.

\subsection{Dictatorial Power with Finite Alternatives}\label{subsection:structural}\label{subsection:finite}

To convey the logic of agenda-setting power most transparently, we start with finite policy spaces and make the ``generic'' assumption that all players have strict preferences.
\begin{definition}\label{ass:GFA}\hypertarget{GFA}
A collective choice problem $\CCP$ satisfies {\hyperlink{GFA}{Generic Finite Alternatives}} if $X$ is finite, and each $\prefto_i$ and $\prefto_A$ are antisymmetric.
\end{definition}
In such settings, the agenda setter generally obtains her favorite policy:
\begin{theorem}\label{Theorem-GenericFiniteAlternatives}
Suppose the collective choice problem $\CCP$ satisfies \hyperlink{GFA}{Generic Finite Alternatives}. For any game with at least $|X|-1$ rounds, the agenda setter obtains her favorite policy in every equilibrium regardless of the initial default if and only if $\CCP$ is \hyperlink{Definition:Manipulable}{Manipulable}. 
\end{theorem}
\Cref{Theorem-GenericFiniteAlternatives} articulates the power of real-time agenda control: with a manipulable policy space, the agenda setter always obtains her favorite policy. On its own, manipulability merely ensures that the agenda setter can find some improvement pallatable to a majority. Indeed, when agendas are fixed in advance as in \citet{shepsle1984uncovered} (so that proposals are not conditional on prior votes), the agenda setter can do no better than her favorite policy among those uncovered by the initial default option, even if the policy space is manipulable. It is therefore the combination of manipulability and real-time agenda control that yields dictatorial power.

The argument for \Cref{Theorem-GenericFiniteAlternatives} is elementary. Denote the set of policies that are majority preferred to $x$ by $M(x) := \left\{y \in X : y\succ_M x \text{ or } y=x \right\}$. We define the \emph{agenda setter's favorite improvement mapping} $\phi : X \to X$ by
\begin{align}\label{Equation-FavoriteImprovement}
    \left\{\phi(x)\right\} := \argmax_{y\in M(x)} u_A(y).
\end{align}
Given the agenda setter's strict preferences, $\phi(\cdot)$ is well-defined. We denote the fixed points of this mapping by $E := \left\{ x \in X : \ x = \phi (x) \right\}$. Note that a policy $x$ is unimprovable if and only if $x \in E$. We write the $t$-fold iteration of $\phi$ for any initial default option $x^0$ as $\phi^t(x^0)$. By definition of $\phi$, for every default $x^0$, (i) $\phi^{t+1}(x^0)\prefto_A \phi^t(x^0)$, and (ii) if $T\geq |\policyspace|-1$, $\phi^T(x^0)$ is an element of $E$ (i.e., unimprovable).

We prove \Cref{Theorem-GenericFiniteAlternatives} by showing that equilibrium outcomes are characterized by iterations of the $\phi$ mapping, regardless of whether manipulability holds. Define the \emph{equilibrium outcome correspondence} for a $T$-round game as $f_T : X \rightrightarrows X$, where $f_T(x^0)$ is the set of policies chosen with positive probability in any equilibrium given an initial default of $x^0$.

\begin{lemma}\label{lemma:structural}
Suppose the collective choice problem $\CCP$ satisfies \hyperlink{GFA}{Generic Finite Alternatives}. For any game with $T$ rounds and initial default policy $x^0$, the equilibrium outcome correspondence satisfies $f_T(x^0) = \left\{ \phi^T (x^0 ) \right\}$. Moreover: 
\begin{enumerate}[label={\normalfont (\alph*)}]
    \item \label{bullet-simpleequilibrium}There exists a pure-strategy equilibrium in which (i) the agenda setter always proposes $\phi(x)$ when the current default is $x$ and (ii) each voter $i$ votes to approve proposal $y$ in round $t$ if and only if $\phi^{T-t}(y) \prefto_i \phi^{T-t}(x^{t-1})$. 
    \item For an initial default $x^0$, $f_T (x^0) = \left\{x^0 \right\}$ if and only if $x^0 \in E$.  
    \item If $T \geq |X|-1$, then $\bigcup_{x^0 \in X} f_T (x^0) = E$.
\end{enumerate}
\end{lemma}

\Cref{lemma:structural} states that the equilibrium outcome correspondence with $T$ rounds coincides with the $T$-fold iteration of the agenda setter's favorite improvement mapping, implying that all equilibria are outcome-equivalent. It also asserts the existence of a simple equilibrium in which the agenda setter follows a ``greedy'' strategy, always acting as if the current round is the last one.\footnote{In this equilibrium, voters break ties in favor of the agenda setter's proposals. Under \hyperlink{GFA}{Generic Finite Alternatives}, voters are indifferent between accepting and rejecting a proposal if and only if both choices lead to the same continuation outcome.} Finally, it records some useful implications: (b) the fixed points of the equilibrium outcome correspondence are the unimprovable policies, and (c) given sufficiently many rounds, every equilibrium outcome is unimprovable. 

Perhaps surprisingly, the agenda setter's strategy in the simple equilibrium above would implement the same outcome if voters were myopic as in \cite{mckelvey1976intransitivities}---that is, if they ignored the possibility of further amendments. But in our setting, voters approve each proposal \emph{precisely} because they anticipate future revisions and prefer the continuation path associated with the proposal. More specifically, the group of voters who approve each proposal along the equilibrium-path are those who favor $\phi^T(x)$ to $\phi^{T-1}(x)$. Because the continuation outcomes for acceptance and rejection of the current proposal do not vary along the equilibrium path, \emph{the same coalition of voters} supports each on-path proposal.

\Cref{Theorem-GenericFiniteAlternatives} is an immediate corollary of \Cref{lemma:structural}(c): the set of unimprovable policies $E$ coincides with $\asfav$ if and only if $\CCP$ is \hyperlink{Definition:Manipulable}{Manipulable}. We therefore sketch the proof of \cref{lemma:structural} here (the full proof is in the Appendix): 
\begin{enumerate}[label={\normalfont (\roman*)}]
    \item With a single round, an equilibrium policy is an element of $\phi(X)$, where $\phi(X)$ is the image of $X$ under $\phi$: if the default option $x^0$ is improvable, then in equilibrium, the agenda setter proposes her favorite improvement $\phi(x^0)$, which passes. 
    \item With two rounds, an equilibrium policy is an element of $\phi^2(X)$. If the initial default option $x^0$ prevails at the end of the first round, then by (i), the resulting policy is $\phi(x^0)$. If the latter policy is improvable, then there exist policies $y$ such that $\phi(y) \succ_M \phi(x^0)$ (for example, $y=\phi(x^0)$). Crucially, in equilibrium, the agenda setter is guaranteed passage of any such proposal in the first round because voters anticipate, by step (i) above, that accepting $y$ would lead to final outcome $\phi(y)$ while rejecting it would lead to $\phi(x^0)$. By definition, the agenda setter's favorite improvement over $\phi(x^0)$ is $\phi^2(x^0)$, so any proposal $y$ for which $\phi(y) = \phi^2(x^0)$ is optimal for her. As described in \Cref{lemma:structural}(a), one such first-round proposal is $y = \phi(x^0)$.
  
    \item By induction, with $T$ rounds, an equilibrium policy is an element of $\phi^T(X)$. As noted before, $\phi^T(X)$ must coincide with $E$ if $T\geq |\policyspace|-1$. 
\end{enumerate}

While the default evolves gradually in the simple equilibrium of \Cref{lemma:structural}(a), there are other equilibria with sudden transitions. Specifically, if $\phi^T(x^0)$ is unimprovable, there are equilibria where the agenda setter proposes it in the first round and it passes.\footnote{However, if $\phi^T(x^0)$ is improvable, then subgame perfection requires gradualism: were voters to accept $\phi^T(x^0)$ in the first round, they would expect the agenda setter to further amend the policy to obtain additional gains for herself, contrary to the majority's interests.} Thus, if the policy space is manipulable and $T\geq |\policyspace|-1$, the group may adopt the agenda setter's favorite policy immediately even though a majority does not prefer it to the initial default.

Note that our explanation for \Cref{lemma:structural} did not invoke any properties of majority rule. Consequently, with appropriate adjustments to the notions of \emph{favorite improvement} and \emph{manipulability}, these results generalize to arbitrary voting rules. In \Cref{Section-OtherProcedures}, we obtain similar results for legislative procedures that feature adjournment clauses that terminate deliberation, such as the successive procedure / closed-rule bargaining and open-rule bargaining. Therefore, the simple recursive logic applies to a broad range of legislative institutions.

Two caveats are in order. First, although manipulability is generic in rich multidimensional collective choice problems (see \Cref{Section-ManipulableCCPs}), the same statement does not hold under \hyperlink{GFA}{Generic Finite Alternatives}: for any finite policy space $\policyspace$, the set of utility profiles $(u_1, \dots, u_n, u_A)$ for which manipulability holds has strictly positive, but not full, Lebesgue measure in $\mathbb{R}^{|\policyspace| \times (n+1)}$. Second, as the cardinality of $X$ increases, the above results require the number of rounds to increase without bound. We address both issues below.

\subsection{Near-Dictatorial Power with Continuous Policy Spaces}\label{subsection:main-cts}\label{subsection:infinite}

Next, we extend our analysis to settings with continuous policy spaces using two distinct approaches. For the first, we take the view 
that the typical real-world collective choice problem offers an extremely large but finite number of alternatives, and that the assumption of continuity is usually a convenient analytic approximation (e.g., for budgets, pennies are indivisible). Instead of studying the continuous case that approximates the settings of interest, we study the discrete settings that the continuous case approximates (i.e., those with large numbers of alternatives). For the second approach, we study continuous policy spaces directly but impose a mild equilibrium refinement. Both approaches yield the same conclusion: regardless of the initial default option, with sufficiently many rounds, the agenda setter's payoff is arbitrarily close to its maximum. 

\paragraph{Discretized Settings.} 
Consider a collective choice problem $\CCP:=(X,\{\prefto_i\}_{i=1,\ldots,n,A})$ that need not satisfy \hyperlink{GFA}{Generic Finite Alternatives}. Let $d(x,y)$ denote a metric on $X$; for a subset $A \subseteq X$, $d(x,A) := \inf_{y \in Y} d(x,y)$ denotes the distance of $x$ from $A$. A \textbf{generic $\epsilon$-grid} is a finite subset $X_\epsilon \subseteq X$ for which {\normalfont (a)} $\max_{x \in X} d(x, X_\epsilon) < \epsilon$, and {\normalfont (b)} the preferences of voters and the agenda setter are antisymmetric within $\policyspace_\epsilon$. We study ``ambient'' collective choice problems that admit generic $\epsilon$-grids for \emph{every} $\epsilon>0$. As we establish in the Appendix (\Cref{Lemma-FA-TI-equiv}), such problems are characterized by the condition that all players have ``thin'' indifference curves. Formally, let $I_i(x) := \left\{y \in X : y \sim_i x \right\}$ denote player $i$'s indifference curve through policy $x$.

\begin{definition}\label{defi:TI}\hypertarget{TI}
A collective choice problem $\CCP$ satisfies \hyperlink{TI}{Thin Individual Indifference} if $I_i(x) \backslash\{x\}$ has empty interior for every player $i$ and $x \in X$.
\end{definition}

\hyperlink{TI}{Thin Individual Indifference} holds in most applications with continuous policies, including divide-the-dollar problems and any setting with strictly convex preferences. The assumption also features in \cite{mckelvey1979general} and \cite{shepsle1984uncovered}, who further assume that the policy space has no isolated points. Our definition generalizes their condition, and for finite policy spaces is equivalent to \hyperlink{GFA}{Generic Finite Alternatives}.\footnote{\hyperlink{GFA}{Generic Finite Alternatives} implies that $x \in X$ is isolated and $I_i(x) \backslash\{x\} = \emptyset$.}

Loosely, we show that, under this assumption, in games with large numbers of rounds and options, the agenda setter ``nearly'' exercises dictatorial power in all equilibria if and only if the ambient collective choice problem is manipulable. Formally, defining $u^*_A \equiv \max_{x \in X}u_A(x)$ for any (continuous) utility representation $u_A$ of $\prefto_A$, we have:\footnote{The statement of \Cref{Theorem-epsilongrid} uses a cardinal measure of near-dictatorial power, but we can restate it in ordinal terms: if and only if $\CCP$ is manipulable, the final policy itself must be close to the agenda setter's favorite policies, $X_A^*$. That is, for every $\delta>0$, there exist $\epsilon_{\delta}>0$ and $T_{\delta}\in\mathbb{N}$ such that, for any sufficiently fine grid and long horizon, and for any initial default, all equilibrium outcomes $y$ are close to the agenda setter's favorite policies, in the sense that  $d(y, X^*_A) < \delta$.}

\begin{theorem}\label{Theorem-epsilongrid}
Suppose the collective choice problem $\CCP$ satisfies \hyperlink{TI}{Thin Individual Indifference}. The following holds if and only if $\CCP$ is \hyperlink{Definition:Manipulable}{Manipulable}:
\begin{quote}
For every $\delta>0$, there exist ${\epsilon}_\delta>0$ and $T_\delta \in \mathbb{N}$ such that, if the policy space is restricted to any generic $\epsilon$-grid $\policyspace_\epsilon$ where $\epsilon<\epsilon_\delta$, and the game has at least $T_\delta$ rounds, then given any initial default in $\policyspace_\epsilon$, the agenda setter's payoff is no lower than $u^*_A-\delta$ in any equilibrium.
\end{quote}
\end{theorem}

We note three additional features of this result. First, it does not require the discretized collective choice problems, $\CCP_\epsilon :=(X_\epsilon,\{\prefto_i\}_{i=1,\ldots,n,A})$, to be manipulable. Second, the agenda setter achieves a payoff within $\delta>0$ of her maximum among all policies in the ambient policy space $X$, not merely those in the grid $X_\epsilon$. Third, the minimal horizon length $T_\delta$ and maximal discretization size $\epsilon_\delta$ depend on the payoff approximation $\delta$, but are uniform across both the choice of the particular grid $X_\epsilon$ and the initial default within that grid. These features distinguish \Cref{Theorem-epsilongrid} from \Cref{Theorem-GenericFiniteAlternatives}: even if $\CCP_\epsilon$ were manipulable, \Cref{Theorem-GenericFiniteAlternatives} would only establish that the agenda setter achieves her favorite option if the number of rounds is at least $|\policyspace_\epsilon|-1$, which explodes as $\epsilon\to 0$. In contrast, \Cref{Theorem-epsilongrid} shows that, with $T_\delta$ rounds, the agenda setter obtains a payoff within $\delta$ of her maximum for \emph{all} sufficiently fine grids. 

The following is a sketch of the proof. First, we show that, if the ambient collective choice problem $\CCP$ is manipulable, then policies that are \emph{unimprovable within the grid $X_\epsilon$} lie within a neighborhood of $X^*_A$ that shrinks to $X^*_A$ as $\epsilon \to 0$. Thus, even if the agenda setter cannot obtain her favorite policy in $X_\epsilon$ (let alone in $X$), she can make sequences of successful proposals that bring the policy \emph{arbitrarily close} to her favorite. The second step shows that, as long as the grid is sufficiently fine, she can achieve these gains within a fixed number of rounds that does not depend on the particular grid. The essence of the argument is that, for any $\delta>0$, there exists a minimal payoff improvement $\eta_\delta>0$ such that, whenever the agenda setter's payoff differs from that of her favorite policy by more than $\delta$, she can find another policy that improves both her payoff and the payoffs of a majority of voters by at least $\eta_\delta$. Using this observation, it is easy to determine the number of rounds that necessarily bring her payoff within $\delta$ of her maximum.\footnote{The desired conclusion follows when the number of rounds exceeds $\left[u^*_A-\min_{x\in \policyspace} u_A(x)\right]/\eta_\delta$, which allows the agenda setter to migrate the policy from her least favorite to her favorite.}

\paragraph{An Equilibrium Refinement for Continuous Settings.}

When considering settings with continuous policy spaces, we cannot assume away indifference. This inconvenient fact raises two issues. First, how do voters break ties when indifferent between two continuation paths? Second, how do we define ``as-if pivotal'' voting when future tie-breaking for other players, and hence continuation paths, may differ depending on the composition of the majority in the current round? A standard approach in the literature is to resolve both issues by restricting attention to pure strategy Markov perfect equilibria in which voters always break indifference in favor of the proposal \citep[e.g.][]{baron1989bargaining,diermeier2011legislative}.
While convenient, this tie-breaking convention potentially stacks the deck in the agenda setter's favor. We therefore consider a weaker refinement: we allow voters to break ties \emph{against} the agenda setter's proposals, as long as they always resolve the same tie (between ultimate outcomes) the same way. Formally:

\begin{definition}\label{defi:NC}\hypertarget{NonCapricious}
An equilibrium is \hyperlink{NonCapricious}{Non-Capricious} if it has the following properties:
\begin{enumerate}[label={\normalfont (\alph*)}]
    \item The induced mapping from histories to continuation outcomes is deterministic  and Markovian (it conditions on the history only through the prevailing default and number of remaining rounds).

\item For each voter $i$ and pair of distinct policies $x$ and $y $ such that $x \sim_i y$, at every history-proposal pair for which $x$ is the continuation outcome if the proposal is accepted and $y$ is the continuation outcome if the proposal is rejected,  voter $i$ either {\normalfont (i)} always votes for the proposal or {\normalfont (ii)} always votes against the proposal.    
    
\end{enumerate}

\end{definition}

Part (a) slightly weakens the standard definition of Markov perfect equilibrium by allowing strategies, but not the continuation outcomes they induce, to depend on payoff-irrelevant features of the history. Part (b) is more important because it disciplines tie-breaking across histories. Suppose voter $i$ is indifferent between policies $x$ and $y$, and that at history $h$ (resp. $h'$), accepting a proposal $a$ (resp. $a'$) leads to policy $x$, while rejecting it leads to $y$. Then if $i$ votes for (resp. against) proposal $a$ at history $h$, she must also vote for (resp. against) proposal $a'$ at history $h'$. In other words, the manner in which a voter breaks ties only depends on the resulting continuation outcomes. 
The logic of this restriction is that the particular history is ``water under the bridge,'' and consequently should not affect the voter's deliberations, even in cases of indifference.\footnote{In settings with \hyperlink{GFA}{Generic Finite Alternatives}, Non-Capriciousness is always satisfied because (i) all equilibria are outcome-equivalent to the specific pure strategy Markov perfect equilibrium from \Cref{lemma:structural}(a), and (ii) voters are never indifferent between \emph{distinct} continuation outcomes (which means the tie-breaking restriction in \Cref{defi:NC}(b) has no bite). Thus, \Cref{Theorem-Capricious} below is a proper generalization of \Cref{Theorem-GenericFiniteAlternatives}.}

We prove that a Non-Capricious equilibrium exists, and that in all such equilibria, the agenda setter has near-dictatorial power whenever the collective choice problem is manipulable. 

\begin{theorem}\label{Theorem-Capricious}
For any collective choice problem $\CCP$, the following hold:
\begin{enumerate}[label={\normalfont (\alph*)}]
    \item There exists a \hyperlink{NonCapricious}{Non-Capricious} equilibrium. 
    \item The following holds if and only if $\CCP$ is \hyperlink{Definition:Manipulable}{Manipulable}: For every $\delta>0$, there exists some $T_\delta \in \mathbb{N}$ such that if the game has $T\geq T_\delta$ rounds, then given any initial default, the agenda setter's equilibrium payoff is no lower than $u^*_A-\delta$ in any \hyperlink{NonCapricious}{Non-Capricious} equilibrium.\footnote{This result can be equivalently stated in ordinal terms: if and only if $\CCP$ is manipulable, the final policy in any non-capricious equilibrium must itself be close to the agenda setter's favorite policies, $X^*_A$.}
\end{enumerate}
\end{theorem}
The general logic of our earlier results continues to govern the proof: once there are sufficiently many rounds, every Non-Capricious equilibrium outcome must be \emph{nearly} unimprovable. Manipulability of the collective choice problem and continuity of the agenda setter's preferences then imply that the agenda setter's payoff is \emph{nearly} maximized. The complete proof, which involves considerable technical detail, appears in our Supplementary Appendix. We illustrate its structure through a full analysis of the standard divide-the-dollar problem in the Main Appendix. That analysis highlights two additional features. First, even with a small number of rounds (in this case, three), the agenda setter may obtain her favorite policy.\footnote{\Cref{Theorem-ThreeRounds} in \Cref{Subsection-Distributive} extends this result to a broad class of distribution problems.} Second, the main conclusion of \Cref{Theorem-Capricious} requires Non-Capricious tie-breaking: for the divide-the-dollar game, there is a Markovian equilibria with \emph{capricious} tie-breaking in which the agenda setter's power is more limited.

\section{Manipulable Collective Choice Problems}\label{Section-ManipulableCCPs}
In this section, we demonstrate that the property driving our main results, manipulability,  is (generically) satisfied in canonical models of spatial and distributive politics. 

\subsection{Spatial Politics}\label{subsection:spatial}
 
In the canonical spatial model, a policy consists of $d$ continuous components. Formally, the policy space is $\policyspace=\mathbb{R}^d$, each player $i$ has an ideal point $x^*_i$, and $u_i(x) = - \frac{1}{2}\| x - x_i^*\|^2$, i.e., players evaluate a policy based on its Euclidean distance from their ideal points.\footnote{Although we have assumed in \Cref{Section-Model} that the policy space is compact, it is convenient here to treat it as  unbounded to simplify the statement of \cref{Theorem-Spatial-Manip} below. However, the proof of that result (also sketched below) establishes the improvability of all policies aside from $x^*_A$ in the \emph{interior} of a compact and convex policy space $X\subsetneq \mathbb{R}^d$. Policies on the boundary of such $X$ are also improvable provided that all ideal points are interior, which is plausible when boundary policies represent extreme alternatives.}  
Under the assumption, the profile of ideal points, $(x^*_i)_{i=1,\dots,n,A} \in \mathbb{R}^{d(n+1)}$, completely characterizes the preference profile. 

Our analysis invokes a property we call \emph{Non-Coplanarity}. For a vector $x\in \Re^d$ where $d\geq 3$, let $[x]_{abc}\equiv (x_a,x_b,x_c)\in \Re^3$ be the projection  of $x$ into the subspace spanned by any three of the dimensions $a,b,c$. Non-Coplanarity entails the following property:

\begin{definition}\label{Definition-NoCoplanarity}\hypertarget{DefinitionNoCoplanarity}
For $d\geq 3$, the profile of ideal points $(x^*_i)_{i=1,\dots,n,A}$ satisfies \hyperlink{DefinitionNoCoplanarity}{Non-Coplanarity} if for every $a,b,c$, no four elements of the set $\left\{[x^*_1]_{abc}, \dots, [x^*_n]_{abc}, [x^*_A]_{abc}\right\} \subset \mathbb{R}^{3}$ are coplanar.
\end{definition}

 When there are only three policy dimensions ($d=3$), \Cref{Definition-NoCoplanarity} simply states that no four ideal points lie in the same plane. If there are more than three dimensions, it requires the same to be true for \emph{all} $3$-dimensional projections---that is, when we only consider dimensions $a,b,c$ and ignore the rest.

Our main result shows that spatial collective choice problems are manipulable whenever \hyperlink{DefinitionNoCoplanarity}{Non-Coplanarity} is satisfied and, moreover, that this condition holds generically.

\begin{theorem}\label{Theorem-Spatial-Manip}
Consider a collective choice problem $\CCP$ with  policy space $X = \mathbb{R}^d$, where $d \geq 3$ and players have Euclidean preferences with ideal points $(x^*_i)_{i=1,\dots,n,A}$.
\begin{enumerate}[noitemsep,label={\normalfont (\alph*)}]
    \item \label{bullet-spatialmanipulable} If the profile $(x^*_i)_{i=1,\dots,n,A}$ satisfies \hyperlink{DefinitionNoCoplanarity}{Non-Coplanarity}, then the collective choice problem $\CCP$ is \hyperlink{Definition:Manipulable}{Manipulable}.
    \item \label{bullet-spatialgeneric} The set of profiles for which  \hyperlink{DefinitionNoCoplanarity}{Non-Coplanarity} holds has full Lebesgue measure and is open-dense in $\mathbb{R}^{d (n+1)}$.
\end{enumerate}
\end{theorem}

\Cref{Theorem-Spatial-Manip} demonstrates that, when there are at least three policy dimensions, the spatial model generically satisfies manipulability, i.e., all policies other than the agenda setter's ideal point, $x_A^*$, are improvable. Equivalently, for a cooperative game in which the decisive coalitions are those containing the agenda setter and some majority of voters, \Cref{Theorem-Spatial-Manip} states that, with three or more dimensions, the \emph{core} of the spatial model generically contains the agenda setter's ideal point and nothing else.  
Given the importance of the spatial model and the elementary geometric argument used to prove \Cref{Theorem-Spatial-Manip}, this result may be of independent interest.\footnote{\cite{duggan2018extreme} offer a related result for ``constrained core points,'' which are similar to unimprovable policies. They show that with four or more dimensions, for the class of $\textbf{C}^2$-smooth and strictly concave utility functions, there are no interior constrained core points for a topologically generic class of utility functions.}

Together with our prior results (\Cref{Theorem-GenericFiniteAlternatives,Theorem-epsilongrid,Theorem-Capricious}), \Cref{Theorem-Spatial-Manip} establishes that the agenda setter can exploit real-time agenda control to obtain her ideal point (exactly or approximately) whenever there are three or more policy dimensions. This conclusion holds even if voters' preferences are largely congruent. Suppose voters' ideal points are (relatively) close to each other, and the agenda setter's ideal point lies far outside their convex hull. As long as \hyperlink{DefinitionNoCoplanarity}{Non-Coplanarity} holds, the agenda setter inevitably obtains her ideal point, \emph{even if the initial default option lies within that convex hull}. In contrast, for fixed agenda models, \cite{mckelvey1986covering} shows that an agenda setter can only achieve policies near the initial default, and the distance between the initial and final policies shrinks to $0$ as voters' ideal points converge to a single point.

The existence of three policy dimensions is critical for \Cref{Theorem-Spatial-Manip}. In the unidimensional case, all policies between $x_A^*$ and the ideal point of the median voter are unimprovable (a consequence of the Median Voter Theorem). For the two-dimensional case, manipulability necessarily fails whenever the agenda setter's ideal point is outside the convex hull of voters' ideal points. We illustrate this failure and elaborate further in \Cref{Section-Supplementary2D}.

The fact that our result requires three or more policy dimensions \citep[rather than two, as in the ``chaos theorem'' of][]{mckelvey1976intransitivities} has implications for \emph{policy bundling}. If the legislature faces a decision involving only one or two dimensions, the agenda setter benefits from introducing a third dimension---even if the associated default is already her ideal---because the collective choice problem thereby becomes manipulable, enabling her to achieve her optima in all three dimensions. 

Next, we sketch the geometric argument for \Cref{Theorem-Spatial-Manip}\ref{bullet-spatialmanipulable} in the $3$-dimensional case. The full proof appears in the Supplementary Appendix. 

\paragraph{Proof Sketch for $d=3$.} 

Consider a policy $x$ that is not the agenda setter's favorite, $x^*_A$. We show that $x$ is improvable using a two-step argument. First, we find a policy $y$ near $x$ such that a majority of voters strictly prefer $y$ to $x$, and moving from $x$ to $y$ generates a \emph{second-order} loss for the agenda setter. Second, we perturb $y$ to some $z$ such that the same majority of voters strictly prefers $z$ to $x$, but moving from $y$ to $z$ generates a \emph{first-order} gain for the agenda setter, so that the agenda setter also strictly prefers $z$ to $x$. This argument then establishes that $x$ is improvable. \medskip

\begin{figure}[h]
    \centering
    \begin{minipage}{.5\textwidth}
  \centering  
\begin{tikzpicture}[domain=0:3, scale=3,thick]

\let\radius\undefined
\newlength{\radius}
\setlength{\radius}{0.5pt}

\let\shorten\undefined
\newlength{\shorten}
\setlength{\shorten}{4pt}

\coordinate (x) at (0,0);
\coordinate (xa) at (-0.8,0.8);
\coordinate (xi) at (1.2,-0.2);
\coordinate (xj) at (-1.2, 0.2);
\coordinate (yi) at (0.5,0.5);
\coordinate (yj) at (-0.5,-0.5);
\coordinate (S) at (.3,0.75);

\draw[transparent] (1.1,-0.7)node[right]{$x_k^*$}; 
\draw[transparent] (xi)node[right,yshift=-5]{$x_i^*$};
\draw[transparent] (xj)node[right,yshift=5]{};

\filldraw[red](xa)circle(\radius);
\draw[red] (xa)node[right,yshift=5]{$x_A^*$};

\draw[red,dashed,<-,shorten >=\shorten, shorten <=\shorten] (xa)--(x);

\draw (-1,-0.7)--(0.4,0.7)--(1,0.7)--(-0.4,-0.7)--(-1,-0.7);
\filldraw[blue](x)circle(\radius);

\tkzMarkRightAngle[size=.05](xa,x,yj);
\draw[red] (x)node[right,yshift=-5]{$x$};
\draw (S)node{$S$};
\filldraw[red](x)circle(\radius);

\end{tikzpicture}
  
    \end{minipage}%
    \begin{minipage}{0.5\textwidth}
 \centering
 \begin{tikzpicture}[domain=0:3, scale=3,thick]

\let\radius\undefined
\newlength{\radius}
\setlength{\radius}{0.5pt}

\let\shorten\undefined
\newlength{\shorten}
\setlength{\shorten}{4pt}

\coordinate (x) at (0,0);
\coordinate (xa) at (-0.8,0.8);
\coordinate (xi) at (1.2,-0.2);
\coordinate (xj) at (-1.2, 0.2);
\coordinate (yi) at (0.5,0.5);
\coordinate (yj) at (-0.5,-0.5);
\coordinate (S) at (.3,0.75);

\filldraw[red](xa)circle(\radius);
\draw[red] (xa)node[right,yshift=5]{$x_A^*$};
\draw (S)node{$S$};

\draw[red,dashed,<-,shorten >=\shorten, shorten <=\shorten] (xa)--(x);

\draw (-1,-0.7)--(0.4,0.7)--(1,0.7)--(-0.4,-0.7)--(-1,-0.7);
\filldraw[blue](x)circle(\radius);

\draw[blue](x)node[right]{$y_A^*$};

\filldraw[blue](xi)circle(\radius) (xj)circle(\radius) (yi)circle(\radius) (yj)circle(\radius);
\draw[blue] (xi)node[right,yshift=-5]{$x_i^*$} (xj)node[right,yshift=5]{$x_j^*$} (yi)node[right,yshift=5]{$y_i^*$} (yj)node[right,yshift=-5]{$y_j^*$} (x)node[right]{$y_A^*$};

\draw[blue,dashed,<-,shorten >=\shorten, shorten <=\shorten] (xi)--(yi);
\draw[blue,dashed,<-,shorten >=\shorten, shorten <=\shorten] (xj)--(yj);

\draw[blue] (yj)--(yi);

\draw[blue,dashed,<-,shorten >=\shorten, shorten <=\shorten] (1.1,-0.7)--(0.2,0.2);
\filldraw[blue] (1.1,-0.7)circle(\radius);
\draw[blue] (1.1,-0.7)node[right]{$x_k^*$};

\draw[purple,ultra thick] (0.6,-0.05)--(0.6,-0.35) (0.75,-0.15)--(0.45,-0.25);

\end{tikzpicture}

    \end{minipage}
    \caption{Construction of plane $S$ (left) and non-colinearity of constrained ideal points (right).}
      \label{fig:Euclid-1}
\end{figure}
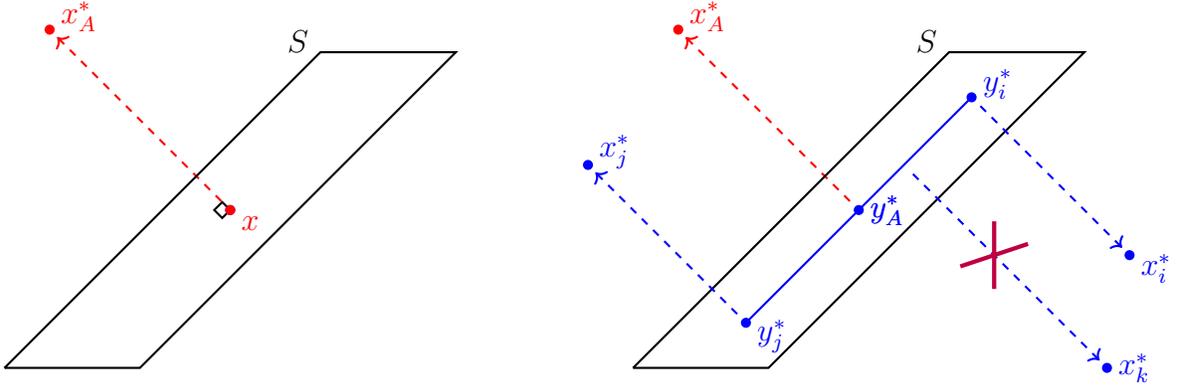

\noindent\textit{Step 1: Constructing $y$.} Let $S \subset \mathbb{R}^3$ denote the plane that is tangent to the agenda setter's indifference curve at the point $x$. As depicted in \cref{fig:Euclid-1} (left panel), $S$ is orthogonal to the gradient $\nabla u_A(x) = x^*_A - x$. Denote the agenda setter's favorite point in $S$---henceforth, her \emph{constrained ideal point}---by $y_A^*$, and observe that, by construction, $y_A^*$ coincides with $x$. Similarly, let $y^*_i \in S$ denote each voter $i$'s constrained ideal point and note that the gradient $\nabla u_i (y^*_i) = x^*_i - y^*_i$ is orthogonal to $S$.

We claim the following: 
\begin{align}\label{Equation-Existenceofy}
    \text{Under \hyperlink{DefinitionNoCoplanarity}{Non-Coplanarity}, }\exists\, y\in S\text{ such that }y\succ_M y_A^*.
\end{align}
To prove \eqref{Equation-Existenceofy}, we make two preliminary observations: (i) there are at most two voters $i \neq j$ such that $y^*_A$, $y^*_i$, and $y^*_j$ all lie on the same line in $S$, and (ii) there can be at most one voter $i$ for whom $y^*_i = y^*_A$. \cref{fig:Euclid-1} (right panel) illustrates the argument for (i): if there were a third voter $k \notin\{i,j\}$ for whom $y^*_k$ were collinear with $\{y^*_A, y^*_i, y^*_j \}$, then, because all players' gradients at their constrained ideal points are orthogonal to the same plane $S$, the four unconstrained ideal points $\{x^*_A, x^*_i, x^*_j, x^*_k \}$ would be coplanar, contradicting the assumption of \hyperlink{DefinitionNoCoplanarity}{Non-Coplanarity}. The argument for (ii) is similar: if there were two voters $i\neq j$ such that $y^*_i = y^*_j = y^*_A$, then the unconstrained ideal points $\{x^*_A, x^*_i, x^*_j \}$ would be collinear, and hence coplanar with the ideal point of any other voter.

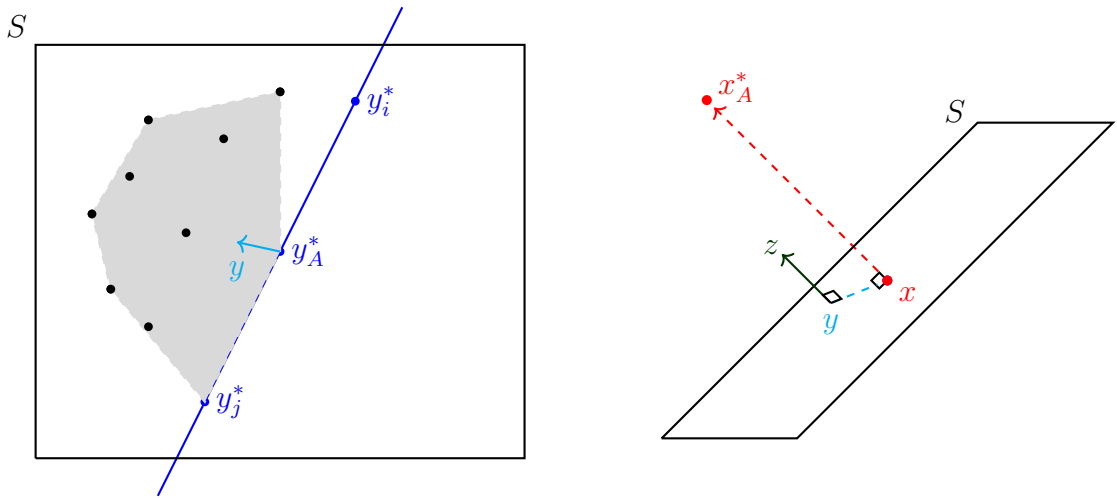
\begin{figure}[h]
    \centering
    \begin{minipage}{.5\textwidth}
\centering
\begin{tikzpicture}[domain=0:3, scale=2.5,thick]

\let\radius\undefined
\newlength{\radius}
\setlength{\radius}{0.5pt}

\coordinate (ya) at (0,0);
\coordinate (yi) at (0.4,0.8);
\coordinate (yj) at (-0.4,-0.8);
\coordinate (S) at (-1.4,1.2);
\draw (S)node{$S$};

\filldraw[blue](ya)circle(\radius) (yi)circle(\radius) (yj)circle(\radius);
\draw[blue] (ya)node[right]{$y_A^*$} (yi)node[right]{$y_i^*$} (yj)node[right]{$y_j^*$};
\draw[blue] (-0.65,-1.3)--(0.65,1.3);

\draw[] (-1.3,-1.1)--(-1.3,1.1)--(1.3,1.1)--(1.3,-1.1)--(-1.3,-1.1);

\filldraw[] (-1,0.2)circle(\radius)(-0.9,-0.2)circle(\radius) (-0.8,0.4)circle(\radius) (-0.7,0.7)circle(\radius) (-0.7,-0.4)circle(\radius) (-0.5,0.1)circle(\radius)  (-0.3,0.6)circle(\radius) (-0,0.85)circle(\radius);

\filldraw[black!15!white,dashed](-1,0.2)--(-0.7,0.7)--(-0,0.85)--(ya)--(yj)--(-0.9,-0.2)--(-1,0.2) ;

\filldraw[] (-1,0.2)circle(\radius)(-0.9,-0.2)circle(\radius) (-0.8,0.4)circle(\radius) (-0.7,0.7)circle(\radius) (-0.7,-0.4)circle(\radius) (-0.5,0.1)circle(\radius)  (-0.3,0.6)circle(\radius) (-0,0.85)circle(\radius);

\draw[cyan, <-] (-0.23,0.05)node[below,yshift=-3]{$y$}--(0,0);

\end{tikzpicture}

    \end{minipage}%
    \begin{minipage}{0.5\textwidth}
    
    \centering
    
    \begin{tikzpicture}[domain=0:3, scale=3,thick]

\let\radius\undefined
\newlength{\radius}
\setlength{\radius}{0.5pt}

\let\shorten\undefined
\newlength{\shorten}
\setlength{\shorten}{4pt}

\coordinate (x) at (0,0);
\coordinate (xa) at (-0.8,0.8);
\coordinate (xi) at (1.2,-0.2);
\coordinate (xj) at (-1.2, 0.2);
\coordinate (yi) at (0.5,0.5);
\coordinate (yj) at (-0.5,-0.5);
\coordinate (S) at (.3,0.75);
\coordinate (y) at (-.25,-.1);
\coordinate (xk) at (-.5,-.2);
\coordinate (xl) at (-.5,-.35);
\coordinate (z) at (-.5,.15);

\filldraw[red](xa)circle(\radius);
\draw[red] (xa)node[right,yshift=5]{$x_A^*$};
\draw (S)node{$S$};

\tkzMarkRightAngle[size=.05](xa,x,yj);

\draw[red,dashed,<-,shorten >=\shorten, shorten <=\shorten] (xa)--(x);

\draw (-1,-0.7)--(0.4,0.7)--(1,0.7)--(-0.4,-0.7)--(-1,-0.7);
\filldraw[blue](x)circle(\radius);

\draw[red] (x)node[right,yshift=-5]{$x$};
\filldraw[red](x)circle(\radius);



\draw[cyan,dashed,shorten <=0.5pt,shorten >=\shorten] (-.25,-.1)node[below]{$y$}--(x);

\draw[green!20!black,<-, shorten <=\shorten,shorten >=0.5pt] (-.5,0.15)node[left,xshift=6]{$z$}--(-.25,-.1);
\tkzMarkRightAngle[size=.05](z,y,x); 

\end{tikzpicture}

    \end{minipage}
    \caption{Perturbations $x \to y \in S$ (left) and $y \to z \notin S$ (right).}
      \label{fig:Euclid-2}
\end{figure}

We now use these observations to deduce \eqref{Equation-Existenceofy}. \cref{fig:Euclid-2} (left panel) shows a head-on view of the plane $S$. By (i), for any line $L \subset S$ through $y^*_A$, there are at most two voters $i\neq j$ whose constrained ideal points lie on $L$. Therefore, at least $n-2$ voters' constrained ideal points lie off of $L$. The pigeonhole principle implies that of these $n-2$ constrained ideal points, a strict majority must lie ``above'' or ``below'' the line $L$; our figure shows $(n-1)/2$ points above the line.   
We can then shift $y_A^*$ to some new policy $y$ slightly ``above'' $L$, so that all of those $(n-1)/2$ voters strictly prefer $y_A^*$ to $y$. Moreover, by fact (ii) above, we can pick the direction $y-y_A^*$ so that at least one of voters $i$ and $j$ also strictly prefers $y$ to $y_A^*$; in \cref{fig:Euclid-2} (left panel), this is voter $j$. Thus, a majority strictly prefers $y$ to $y_A^*=x$. Furthermore, we can chose $y_A^*$ arbitrarily close to $y_A^*$, so the agenda setter only incurs a second-order loss (because $y_A^*$ is the agenda setter's ideal policy in $S$).\medskip

\noindent\textit{Step 2: Constructing $z$.} We construct $z$ by perturbing $y$ off of the plane $S$ in the direction $\nabla u_A(x) = x^*_A -x$; see the right panel of \cref{fig:Euclid-2}. Moving from $y$ to $z$ generates a first-order gain for the agenda setter, ensuring that $z \succ_A x$ (because the original move from $x$ to $y$ generates only a second-order loss). Moreover, we can choose the point $z$ close enough to $y$ to ensure that those who strictly prefers $y$ to $x$ (a majority) also strictly prefers $z$ to $x$. \hfill \qedsymbol

\subsection{Distributive Politics}\label{Subsection-Distributive}

This section shows that collective choice problems involving ``distributive politics'' are generally manipulable, and consequently that all problems become manipulable with the addition of pork or transfers. Moreover, settings with distributive politics satisfy a strong version of manipulability that encompasses all voting rules for which no voter has veto power. 

We begin with a definition:
\begin{definition}\label{defi:DP}
\hypertarget{DistributionProblem}A collective choice problem $\CCP = (X, \{\prefto_i\}_{i=1,\dots,n,A})$ is a \hyperlink{DistributionProblem}{Distribution Problem} if it satisfies the following two properties for every policy $x \in X$ and player $i \in N \cup \{A\}$ (where we let $u_i$ represent $\prefto_i$):
\begin{enumerate}[noitemsep,label={\normalfont (\alph*)}]
     \item \textbf{Scarcity:} \hypertarget{Bullet:Scarcity} If $u_i(x) < \max_{z \in X} u_i(z)$, then there exists either some player $j \neq i$ such that $u_j(x) > \min_{z \in X} u_j(z)$, or some policy $y$ such that $u_k(y) > u_k(x)$ for all players $k$.%
     \item \textbf{Transferability:} \hypertarget{Bullet:Transferability} If $u_i(x) > \min_{z \in X} u_i(z)$, then there exists some policy $y$ such that $u_j (y) > u_j(x)$ for all players $j \neq i$.
\end{enumerate}

\end{definition}

Scarcity captures the notion that utility flows from a limited resource: if the resource is not being used to maximize player $i$'s payoff, then either it is being used to give some other player more than her minimal utility, or there is waste, in which case some other allocation could make all players strictly better off. Transferability captures the notion that the underlying resource is at least somewhat fungible: if player $i$ enjoys surplus, we can redistribute some of that surplus to everyone else.\footnote{\cite{banks2006general} call this notion ``limited transferability.''} Notably, this definition does not require utility to be fully transferable.

As we show next, the class of Distribution Problems encompasses a wide range of possibilities, including two canonical cases: any collective choice problem augmented with transfers (including divide-the-dollar), and settings with pork-barrel politics. 
\begin{example}[Divide-the-Dollar / Collective Choices with Transfers]\label{Example-CCPTransfers}
    Consider any collective choice problem with policy space $X$ and utility profile $(u_i)_{i =1, \dots, n,A}$. We augment this problem with monetary transfers. Assuming utility is quasi-linear in money and that each player has an outside option yielding a payoff of zero, the policy space for the resulting transferable-utility collective choice problem is
    \begin{align*}
        \mathcal{Y} = \left\{ y \in \mathbb{R}_+^{n+1} : \exists x \in X \text{ such that }  \sum_{i=1,\dots,n,A} y_i = \sum_{i=1,\dots,n,A} u_i(x)   \right\}
    \end{align*}
    and the utility functions are $v_i(y) = y_i$. 
    This formulation encompasses both the standard divide-the-dollar problem \citep[e.g.][]{baron1989bargaining}, as well as settings involving both production decisions and transfers.\footnote{For simplicity, in this example, utility is fully transferable. As noted above, one can weaken this assumption.}
\end{example}

\begin{example}[Pork Barrel Politics]\label{Example-PorkBarrelPolitics}
Suppose there are finitely many public projects $k \in \mathcal{K}$, each of which generates an aggregate benefit $B^k > 0$ and aggregate cost $C^k > 0$. Some projects may be inefficient ($C^k>B^k$). A policy $x$ specifies (i) the projects the group will implement (a subset $\mathcal{I} \subseteq \mathcal{K}$), and (ii) for each of those projects, the distribution of benefits and costs among the players (i.e., $b^k, c^k \in \mathbb{R}^{n+1}_+$ such that $\sum_{i = 1, \dots, n, A} b^k_i = B^k$ and $\sum_{i = 1, \dots, n, A} c^k_i = C^k$). Player $i$'s preferences correspond to $u_i(x) = \sum_{k \in \mathcal{I}} \left( b^k_i - c^k_i \right)$. Thus, costs and benefits are both perfectly transferable.\footnote{This example blends the model of \citet{baron1991majoritarian}, who considers a single project with perfectly transferable benefits but a fixed distribution of costs, with that of \citet{ffm-porkbarrel-AJPS1987}, \citet{bernheim2006power}, and others, who consider multiple projects with fixed distributions of both benefits and costs.} 
\end{example}

In addition to being ubiquitous, Distribution Problems are manipulable. The proof is simple: if a policy $x$ is not one of the agenda setter's favorites (i.e., not in $X^*_A$), \hyperlink{Bullet:Scarcity}{Scarcity} implies that either (i) some other policy $y$ strongly Pareto dominates  $x$, or (ii) some voter $i$ obtains more than her minimal utility from policy $x$, in which case  \hyperlink{Bullet:Transferability}{Transferability} implies that there is a policy $y$ such that both the agenda setter and all voters other than $i$ strictly prefer $y$ to $x$. In either case, $x$ is obviously improvable; indeed, \emph{all} players (with the possible exception of $i$) strictly prefer $y$ to $x$.

As the preceding argument makes clear, Distribution Problems satisfy a strong version of manipulability that encompasses any voting rule for which no voter has veto power (rather than just majority rule). We formalize this point as follows. A general voting rule is a collection $\mathcal{D} \subseteq 2^N$ of \emph{winning coalitions} $D \subseteq N$, by which we mean that a proposal passes if and only if there exists some $D \in \mathcal{D}$ for which all voters $i \in D$ support the proposal. This class of voting rules includes \emph{quota rules}, for which there is a quota $q$ such that $\mathcal{D} = \{D \subseteq N : |D|\geq q\}$, as well as rules that treat voters asymmetrically. A voting rule $\mathcal{D}$ is \emph{veto-proof} if for every voter $i \in N$, there exists a winning coalition $D \in \mathcal{D}$ such that $D \subseteq N \backslash \{i\}$; in other words, voter $i$'s support is not necessary for a proposal to pass. We say that a collective choice problem is \emph{$\mathcal{D}$-Manipulable} if for every $x\notin X_A^*$, there exists a policy $y$ and a coalition $D\in \mathcal D$ such that $y\succ_A x$ and $y\succ_i x$ for every $i\in D$. 

The following result formalizes our observation that Distribution Problems are manipulable for a broad class of voting rules.
\begin{theorem}\label{Theorem-DP-Manip}
If $\mathcal{C}$ is a \hyperlink{DistributionProblem}{Distribution Problem}, then it is $\mathcal{D}$-\hyperlink{Definition:Manipulable}{Manipulable} for every veto-proof voting rule $\mathcal D$.
\end{theorem}

We omit a formal proof, as the argument is identical to the one given above for simple-majority rule. \cref{Theorem-DP-Manip}, coupled with our prior results, highlights the broad power of agenda control: for any Distribution Problem and veto-proof voting rule, the agenda setter obtains a near-favorite policy in every (non-capricious) equilibrium, regardless of the initial default, provided there are sufficiently many rounds.\footnote{\label{GenFn}As we asserted in \Cref{subsection:finite}, \Cref{Theorem-GenericFiniteAlternatives,Theorem-epsilongrid,Theorem-Capricious} extend, with obvious (minor) adjustments, to general voting rules $\mathcal{D}$ and collective choice choice problems that are $\mathcal{D}$-manipulable.} 

In fact, for Distribution Problems, we obtain an even starker result: the agenda setter can obtain her favorite policy (not merely an approximation) even if the game is short, where the required number of rounds depends on the voting rule. Recalling that $u^*_A \equiv \max_{x \in  X} u_A(x)$, we have:

\begin{theorem}\label{Theorem-ThreeRounds}
Suppose $\CCP$ is a \hyperlink{DistributionProblem}{Distribution Problem} satisfying \hyperlink{TI}{Thin Individual Indifference}. Then:
\begin{enumerate}[label={\normalfont (\alph*)}]
    \item If the voting rule is a quota rule with $q<n$, the agenda setter obtains payoff $u^*_A$ in every \hyperlink{NonCapricious}{Non-Capricious} equilibrium regardless of the initial default for any game with at least $\lceil n/(n-q)\rceil$ rounds.
    \item If the voting rule is veto-proof, the agenda setter obtains payoff $u^*_A$ in every \hyperlink{NonCapricious}{Non-Capricious} equilibrium regardless of the initial default for any game with at least $n$ rounds.
\end{enumerate}
\end{theorem}
\Cref{Theorem-ThreeRounds}(a) implies that three rounds suffice for the agenda setter to obtain her favorite policy under any quota rule requiring no more than two-thirds majority. For more demanding quotas, more rounds are required. \Cref{Theorem-ThreeRounds}(b) tells us that $n$ rounds suffice for all veto-proof voting rules (because the agenda setter can appropriate the surplus of at least one voter in each round). 

\Cref{Theorem-DP-Manip,Theorem-ThreeRounds} have two broad implications for legislative bargaining.
First, because any collective choice problem becomes a Distribution Problem when bundled with transfers (as in \Cref{Example-CCPTransfers}), our analysis highlights how the introduction of distributional policy instruments can augment an agenda setter's power. Second, equilibrium outcomes need not maximize total surplus. In \Cref{Example-PorkBarrelPolitics}, the agenda setter secures a policy that includes all projects, which maximizes total benefits, along with transfers that offload all costs onto the voters. Plainly, such policies typically involve excess spending relative to the utilitarian optimum. 

\section{Commitment, Procedures, and Deadlines}\label{Section-Robustness}
This section presents results that either clarify or extend our main findings. \Cref{Subsection-Precommitment} explains how commitments to dynamic or fixed agendas change the attainable outcomes. \Cref{Section-OtherProcedures} considers other legislative protocols and voting rules. \Cref{Section-Deadlines} clarifies the role of a finite horizon and highlights the resulting deadline effect. 
\subsection{The Commitment Benchmark}\label{Subsection-Precommitment}
We have studied agenda control without commitment. If the agenda setter could commit to a dynamic strategy, she would do weakly better. \Cref{Theorem-GenericFiniteAlternatives,Theorem-epsilongrid,Theorem-Capricious} imply that if the collective choice problem is manipulable, the agenda setter gains little or nothing from commitment. In this section, we make the same comparison without imposing manipulability. This exercise shows how sequential rationality constrains the agenda setter in non-manipulable problems, and also connects our work to prior research that assumes she can make commitments.  

For simplicity, we restrict attention to \hyperlink{GFA}{Generic Finite Alternatives}, and assume that $T\geq |X|-1$. For any positive integer $K$ and policies $x$ and $y$, we say that \emph{$y$ is $K$-reachable from $x$} if there is a sequence of policies $\left\{a^k \right\}_{k =0}^K$ such that (i) $y = a^K$ and $x = a^0$ and (ii) $a^{k} \prefto_M a^{k-1}$ for all $k \in \left\{1, \dots , K \right\}$. We say that \emph{$y$ is reachable from $x$} if it is $K$-reachable from $x$ for some $K$. For settings with commitments to dynamic agendas, we have:
\begin{fact}\label{Theorem-Commitment}
If the agenda setter can commit to any strategy, then she obtains her favorite policy among those reachable from $x^0$.
\end{fact}
We omit a formal proof but sketch the logic. Suppose that from an initial default option $x^0$, the agenda setter's favorite reachable policy is $y$, and that $\{a^k\}_{k=0}^K$ is the proposal sequence that reaches it. The agenda setter can obtain $y$ by committing to any $K$-round strategy with the following property: if the default option in some round is $a^k$, she proposes $a^{\min\{K,k+1\}}$. Clearly, voters necessarily approve the final proposal; approval of every prior proposal follows recursively.\footnote{To verify that the agenda setter can attain \emph{only} policies that are reachable from $x^0$, fix an arbitrary (pure) strategy for the agenda setter and let $f^t$ denote the equilibrium continuation outcome in round $t$ \emph{if the initial default $x^0$ has not yet been amended prior to round $t$}. By construction, $f^1$ is the equilibrium outcome under this strategy and, because voters are sequentially rational, we must have $f^1 \prefto_M f^2 \prefto_M \cdots \prefto_M f^T \prefto_M x^0$. It follows that the outcome $y:= f^1$ is reachable from $x^0$.}

\Cref{Theorem-Commitment} is familiar from the literature on \emph{binary voting trees}---in other words, multi-stage voting games in which (majorities of) voters decide to move ``left'' or ``right'' in each round, where the resulting path determines the final policy.\footnote{See \citet[Ch. 4]{austen2005positive} for a formal definition, as well as a survey of the classic literature \citep[e.g.][]{black1958theory,farquharsion-book,miller1977graph}.} This connection is not coincidental: an agenda-setter strategy with default $x^0$ induces a binary voting tree for which $x^0$ is feasible, and conversely, every binary voting tree for which $x^0$ is feasible is outcome-equivalent to some agenda-setter strategy with $x^0$ as the initial default.

Next, we consider commitments to fixed agendas as in \cite{shepsle1984uncovered}. Such agendas are equivalent to \emph{default-independent} strategies that prescribe the same proposal for a given round regardless of how the game unfolds. We restate their main result as follows:
\begin{fact}\label{Claim-SW}
If the agenda setter can commit but is restricted to default-independent strategies, she obtains her favorite policy among those that are $2$-reachable from $x^0$.  
\end{fact}
The agenda setter generally does worse with commitments to fixed agendas than with commitments to dynamic agendas because her fixed-agenda options are more limited. The notion of $2$-reachability is equivalent to \citeauthor{shepsle1984uncovered}'s concept of being \emph{uncovered} by the initial default, and consequently the logic of Fact \ref{Claim-SW} is familiar.\footnote{Suppose $y$ is $2$-reachable from $x^0$ via a sequence $\{a^k\}_{k=0}^{2}$. Then the agenda setter can achieve $y$ by committing to proposing $y$ in the first round and $a^1$ in every subsequent round. Now suppose the fixed agenda $(a^1,\ldots,a^T)$ achieves $y$. We claim that $y$ is $2$-reachable. Let $f^t$ be the equilibrium continuation outcome if $a^t$ passes in round $t$. Because $y$ is the eventual outcome, we must have $y=f^\tau$ for some $\tau$, and for all $f^t\neq y$, $y\succ_M f^t$. Were $y$ not $2$-reachable from $x^0$, then it would have to be the case that $x^0\succ_M y$ (otherwise the sequence $\{x_0,y,y\}$ would reach $y$) and for all $f^t\neq x^0$, $x^0\succ_M f^t$ (otherwise the sequence $\{x_0,f^{t'},y\}$ would reach $y$ for some $f^{t'} \succ_M x^0$). But then none of the proposals would pass, a contradiction.}

We now compare real-time agenda setting to these benchmarks. We say that $y$ is \emph{credibly reachable} from $x$ if there is a sequence $\{a^k\}_{k=0}^K$ running from $x$ to $y$ such that $a^k = \phi( a^{k-1})$, where $\phi$ is the agenda setter's favorite improvement (defined on p.  \pageref{Equation-FavoriteImprovement}). In other words, each proposal in the chain that reaches $y$ from $x$ is the agenda setter's favorite among policies that are majority-preferred to proposal's predecessor.  \Cref{lemma:structural}(a) implies:
\begin{fact}\label{Claim-CrediblyReachable}
If the agenda setter cannot commit, then she obtains her favorite policy among those credibly reachable from $x^0$.
\end{fact}

\Cref{Theorem-GenericFiniteAlternatives} establishes that commitment has no value if the collective choice problem is manipulable. In that case, the agenda setter's \emph{favorite} policy is not only reachable, but also credibly reachable, from all $x^0$.\footnote{Manipulability allows the agenda setter to obtain her favorite policy only if the number of rounds is sufficiently large. With a small number of rounds, the agenda setter could potentially benefit from commitment because it allows her to exploit a larger class of majority-preference chains.} However, if the collective choice problem is not manipulable, the agenda setter may do strictly better with commitment, even to a default-independent strategy, as seen in the second example of \Cref{Section-Examples}. Facts 1-3 imply only that commitment to general strategies weakly outperforms both alternative protocols. All weak or strict rankings over these three modes of commitment are feasible as long as they are compatible with this implication.

\subsection{General Legislative Procedures}\label{Section-OtherProcedures}

Legislatures sometimes use alternatives to the amendment procedure studied in previous sections. The best known alternative is the \emph{successive procedure} (also  \emph{closed-rule bargaining}): all proposals include \emph{adjournment provisions} specifying that their acceptance ends deliberation. Another is \emph{open-rule bargaining}: in any round, the agenda setter can ``move'' the prevailing default; if the motion passes, the legislature adjourns.\footnote{The literature on legislative bargaining has focused on the closed- and open-rule procedures since \citet{baron1989bargaining}, while the literature on agenda setting with fixed agendas has largely focused on the amendment (or Anglo-Saxon) and successive (or Euro-Latin) procedures since \citet{black1958theory}, \citet{farquharsion-book}, and \cite{miller1977graph}. While the literature models the closed-rule and successive procedures differently, they are essentially equivalent in that, under both procedures, the first accepted proposal is implemented.}  Legislatures also differ with respect to voting rules (e.g., majority versus supermajority requirements). 

This section analyzes the implications of real-time agenda control for these alternative procedures. We develop a general framework that allows for an arbitrary voting rule and a general \emph{adjournment protocol}, including as special cases our baseline framework and both alternatives mentioned above. We obtain the following result: \emph{for every preference profile and voting rule, (essentially) all adjournment protocols result in the same equilibrium outcome.} In other words, real-time agenda control nullifies the effect of adjournment provisions, rendering the distinction between these various protocols moot. 

We extend the framework of \Cref{Section-Model} as follows. The definition of a collective choice problem $\CCP$ is unchanged, except we allow for an even number of voters, $n$. For simplicity, we focus on settings with \hyperlink{GFA}{Generic Finite Alternatives}. Policy selection takes place over finitely many rounds $t=1, \dots, T$. The agenda setter (resp., voters) has exclusive proposal (resp., approval) power. Here we allow for a wider class of voting rules and adjournment protocols, which we call \emph{generalized amendment procedures}: 
\begin{enumerate}[label={\normalfont (\alph*)}]
    \item The \emph{voting rule} is defined (as in \Cref{Subsection-Distributive}) by a collection $\mathcal{D} \subseteq 2^N$ of winning coalitions. A proposal passes if and only if all voters in some coalition $D \in \mathcal{D}$ approve it. We impose no structure on  $\mathcal{D}$. 
  
    \item The \emph{adjournment protocol} is defined as follows. In round $t$, the agenda setter can propose an alternative $\hat{a}^t=(a^t,i)\in \policyspace\times\{0,1\}$, where $a^t$ denotes the policy to supersede the prevailing default $x^{t-1}$, and $i$ denotes the presence or absence of an \emph{adjournment provision}. If $i=0$, passage makes policy $a^t$ the default in round $t+1$, as in our baseline model. If $i =1$, passage ends deliberation and results in the implementation of $a^t$. In either case, rejection means that $x^{t-1}$ remains the default in round $t+1$. We allow for the possibility that deliberation changes the set of feasible proposals: for a generic history $h$, the agenda setter can propose an element of $X(h)\subseteq X\times \{0,1\}$.
\end{enumerate}

A generalized amendment procedure is \emph{rich} if, at every history $h$, either $X(h) \subseteq X \times \{0\}$ or $X(h) \subseteq X \times \{1\}$ (or both). In other words, richness rules out protocols where some policy $x$ is available only without an adjournment provision, while some other policy $y$ is only available with one. 
Our baseline model and the other procedures mentioned above are rich generalized amendment procedures.\footnote{For the amendment procedure, $X(\cdot) = X \times\{0\}$. For the successive/closed-rule procedure,  $X(\cdot) = X \times\{1\}$. For the open-rule procedure, $X(h) = \left[ X\times \{0\}\right] \cup \left[\{x(h)\} \times \{1\}\right]$, where $x(h)$ denotes the prevailing default at history $h$. Note that the open-rule procedure involves history-dependent feasible sets.}  

We show that, for any fixed preference profile, all rich generalized amendment procedures with the same voting rule yield equivalent equilibrium outcomes. To state the formal result, we extend our notion of improvability and favorite improvements to arbitrary voting rules. First, given any policy $x$, we define the set of policies that some winning coalition prefers to $x$: 
\begin{align*}
    M_{\mathcal D}(x)\equiv\{y\in X:y=x\text{ or }\exists D\in \mathcal D\text{ such that for every } i\in D,\, y\succ_i x\}.
\end{align*}
A policy $x$ is $\mathcal{D}$-improvable if there exists a policy $y \in M_{\mathcal D}(x)$ such that $y\succ_A x$; otherwise, policy $x$ is $\mathcal D$-unimprovable. 
Let $\phi_\mathcal{D} : X \to X$ denote the agenda setter's \emph{favorite $\mathcal{D}$-improvement}: 
\begin{align*}
    \{\phi_\mathcal{D}(x)\}\equiv \argmax_{y\in M_{\mathcal D}(x)} u_A(y).
\end{align*}
The set of $\mathcal{D}$-unimprovable policies is $E_\mathcal{D} := \left\{ x \in X: x = \phi_\mathcal{D}(x) \right\}$. Using this notation, we state our \emph{protocol-equivalence} result:

\begin{theorem}\label{Theorem-proceduralequivalence}
Suppose the collective choice problem $\mathcal C$ satisfies \hyperlink{GFA}{Generic Finite Alternatives} and the generalized amendment procedure is rich. For any game with $T$ rounds and initial default policy $x^0$, the unique equilibrium outcome is $\phi_\mathcal{D}^T(x^0)$. Consequently, for $T\geq|X|-1$, a policy is an equilibrium outcome if and only if it is an element of $E_\mathcal{D}$.
\end{theorem}

Thus, with real-time agenda setting, equilibrium outcomes do not depend on the adjournment protocol. This result has two noteworthy implications. 
First, as long as the collective choice problem is $\mathcal{D}$-manipulable (as defined in \Cref{Subsection-Distributive}), the agenda setter is effectively a dictator regardless of the adjournment protocol. Formally, \Cref{Theorem-proceduralequivalence} implies:
\begin{corollary}\label{Corollary-GeneralizedAmendment}
Suppose the collective choice problem $\mathcal C$ satisfies \hyperlink{GFA}{Generic Finite Alternatives} and the generalized amendment procedure is rich. For any game with at least $|\policyspace|-1$ rounds, the agenda setter obtains her favorite policy in every equilibrium regardless of the initial default if and only if $\CCP$ is $\mathcal{D}$-\hyperlink{Definition:Manipulable}{Manipulable}.  
\end{corollary}

Second, \Cref{Theorem-proceduralequivalence} contrasts with known results on fixed agendas. In that context, the agenda setter's power depends on the adjournment protocol.\footnote{This theme emerges in \cite{farquharsion-book}, \cite{miller1977graph}, and \cite{mckelvey1978multistage}; see Chapter $4$ of \cite{austen2005positive} for a survey. More recent work includes \cite{apesteguia2014foundation} and \cite{barbera2017sequential}.} Specifically, commitment to a fixed successive (or closed-rule) agenda allows her to obtain her favorite policy among those reachable from the initial default,\footnote{This observation essentially restates \Cref{Theorem-Commitment} for binary voting trees. The logic is as follows: if policy $y$ is reachable from default $x$ through the sequence $\{a^k\}_{k=0}^K$, the agenda setter can obtain $y$ in $T=K$ rounds through the fixed agenda where the first proposal is $a^K$, the second is $a^{K-1}$, and so on, and each proposal includes an adjournment provision.} whereas commitment to a fixed amendment agenda only allows her to obtain her favorite $2$-reachable policy \citep{shepsle1984uncovered}. \Cref{Theorem-proceduralequivalence} shows that this distinction  disappears when the agenda setter selects proposals in real time without commitment. 

To prove \Cref{Theorem-proceduralequivalence}, we adjust the argument for \Cref{lemma:structural} to account for adjournment provisions; we omit a formal proof but describe the adjustment. Richness of the generalized amendment procedure guarantees that, at every history, if the default option is $x$, the agenda setter can propose at least one of the following: (i) her favorite $\mathcal{D}$-improvement $\phi_\mathcal{D}(x)$ without an adjournment provision, or (ii) the ``eventual outcome'' $\phi^T_\mathcal{D}(x)$ with an adjournment provision. These options yield the same outcome in any one-round subgame ($T=1$). Therefore, by the backward-induction logic of \Cref{lemma:structural}, both proposals lead to the outcome $\phi^T_\mathcal{D}(x)$ regardless of how many rounds, $T \in \mathbb{N}$, remain. \Cref{Theorem-proceduralequivalence} then follows from our observation that $\bigcup_{x^0 \in X}\phi^T_\mathcal{D}(x^0)  = E_\mathcal{D}$ for $T\geq |X|-1$.\footnote{As the following example illustrates, richness plays an additional role in the proof of \Cref{Theorem-proceduralequivalence}. Suppose the options are $\{w,x,y,z\}$, and that preferences are as depicted in \cref{Figure-IntroFiniteExample} of \cref{Section-Examples}. With our baseline amendment procedure, the agenda setter obtains her favorite policy $w$ in every equilibrium provided there are at least three rounds, regardless of the initial default. Now consider the non-rich generalized amendment procedure with simple majority rule for which, at every history, the agenda setter can propose policies $w,y,z$ only with adjournment provisions and policy $x$ only without one. As the reader may verify, starting from initial default $x^0 = z$, regardless of the number of rounds, $(y,1)$ is the outcome (the agenda setter proposes it and it passes), contrary to \Cref{Theorem-proceduralequivalence}. The logic of the proof fails because $y = \phi(z)$ is not available without adjournment and $x = \phi^2(z)$ is not available with adjournment.} 

\subsection{The Role of Deadlines}\label{Section-Deadlines}
We think of our framework as representing negotiations, starting at date 0, over the policy that will prevail at date $\tau$.\footnote{The legislature presumably undertakes many such negotiations (one for each future date) in parallel.} Negotiations obviously cannot continue past the implementation date. Given the inherent frictions arising either from institutional constraints or simply from speed-of-light latency considerations, we treat each round of bargaining as requiring at least $\Delta>0$ units of time. Consequently, there can be at most $T = \lfloor \tau/\Delta \rfloor$ rounds of deliberation. Hence we follow the prior literature on agenda setting by modeling finite-round processes.\footnote{The classical literature on agenda setting discussed in \Cref{Subsection-Precommitment} studies fixed agendas with a finite sequence of proposals or binary voting trees of finite depth.} The finite deadline effectively provides the agenda setter with a bit of commitment power: the process ends with a take-it-or-leave-it offer. In this section, we investigate the role of this deadline.

In contrast to our analysis, \cite{diermeier2011legislative} and \cite{anesi2014bargaining} model agenda control with an infinite horizon. For concreteness, we focus on the latter analysis, which differs from ours in one key respect: there is no exogenous terminal round $T$. Instead, bargaining endogenously terminates only when the agenda setter either (a) proposes the prevailing default option or (b) makes a proposal that is rejected. Payoffs are undiscounted and determined by the policy implemented at termination; a non-terminating path is the worst outcome for all players.\footnote{Instead of adopting this termination rule, \cite{diermeier2011legislative} study the patient limit of an infinite-horizon model in which players maximize the discounted payoff from the infinite sequence of equilibrium default policies. We group them and \cite{anesi2014bargaining} together as their models have identical equilibrium outcomes in settings with \hyperlink{GFA}{Generic Finite Alternatives}.\label{fn:AS-DF-models}} The solution concept is pure strategy Markov perfect equilibrium with as-if pivotal voting (henceforth MPE), meaning in this context that strategies are stationary and condition only on the prevailing default.

To illustrate the implications of this protocol, we revisit our introductory example from \Cref{Section-Examples} (see \Cref{Figure-IntroFiniteExample}). We showed in that section that, assuming a finite horizon, the agenda setter obtains her favorite policy ($w$) starting from any initial default if there are three or more rounds. In contrast, with the infinite-horizon protocol (and its termination rule), the agenda setter can do no better than $y$ when starting from an initial default of $z$ or $y$. We sketch the logic by considering each default option.

\begin{itemize}
    \item \textbf{Default option of $w$}: \emph{In every MPE, the agenda setter proposes $w$ and the game ends.} Any other outcome (a different policy or a non-terminating cycle) is worse for the agenda setter.
    \item \textbf{Default option of $x$}: \emph{In every MPE, the agenda setter proposes $w$ and it passes.} Voters predict that passing $w$ enacts it (by the preceding logic). Because $w\succ_M x$, a majority of voters approve the proposal.
    \item \textbf{Default option of $y$}: \emph{In every MPE, $y$ is implemented.} A majority of voters will not approve either $w$ or $x$: by the preceding logic, they predict that passage of either will result in the implementation of $w$, and $y\succ_M w$. 
    \item \textbf{Default option of $z$}: \emph{In every MPE, the agenda setter proposes $y$ and it passes.} Voters predict that passing $y$ leads to its enactment, and $y\succ_M z$, so a majority of voters approve the proposal. As with a default of $y$, a majority of voters will not approve either $w$ or $x$. 
\end{itemize}
When $T$ is finite, the agenda setter can propose $x$ in the terminal round \emph{while also credibly committing not to amend it further in the future}, even though $x$ is improvable. In contrast, with an infinite horizon, the preceding discussion reveals that the agenda setter can never make a similar commitment. The freedom to reconsider policies indefinitely generates additional sequential rationality constraints, significantly weakening the agenda setter's power. 

\cite{diermeier2011legislative} and \cite{anesi2014bargaining} show generally that, in infinite-horizon settings with \hyperlink{GFA}{Generic Finite Alternatives}: (a) the set of MPE outcomes corresponds to the \emph{von Neumann-Morgenstern stable set}, $V$, and (b) in every MPE with initial default $x^0$, the agenda setter obtains her favorite policy $y$ among those in $V$ satisfying $y \prefto_M x^0$.\footnote{A set $V \subseteq X$ is \emph{stable} if no $x \in V$ is improvable by another $y  \in V$ (``internal stability''), while every $x \notin V$ is improvable by some $y \in V$ (``external stability''). \citet{diermeier2012characterization} show that there exists a unique stable set in the present context.}  The stable set necessarily includes all unimprovable policies (i.e., $E\subseteq V$), but is typically larger. Consequently, (a) more policies can arise in equilibrium with an infinite horizon than with a long finite horizon, and (b) the agenda setter is weakly worse off with an infinite horizon than with a single proposal round.

We view the finite- and infinite-horizon models as having different domains of applicability. A leading interpretation of the infinite-horizon model is that it represents legislative decisionmaking with an uncertain deadline: deliberations end in round $t$ with probability $(1-\beta)$ with the realization occurring after that round's proposal and votes; analysis focuses on the limiting case of $\beta\rightarrow 1$.\footnote{An alternative interpretation is that the infinite-horizon model with discounting and no termination, as in \cite{diermeier2011legislative}, captures settings in which the legislature chooses policies for a sequence of calendar dates $t \in \{1,2,\ldots\}$. Specifically, the winning option for round $t$, $x^t$, becomes the policy for that period and serves as the default for $t+1$. Accordingly, policies do not vary over time (i.e., $x^\tau = x^t$ for all $\tau \geq t+1$) unless there are further amendments. Under this interpretation, the ``legislative session'' at each calendar date $t$ consists of a single proposal round. In contrast, our view is that legislatures can negotiate over \emph{policy trajectories} specifying continuation paths $(x^t, x^{t+1}, \dots)$ of time-indexed policies for all future dates. Examples include phase-in and sunset provisions. In other words, default trajectories are not necessarily constant as the preceding perspective assumes. We also take the view that each time-indexed session should consist of multiple proposal rounds rather than one. Modeling dynamic collective choice in this manner effectively makes the problem separable across periods, in which case our separate solutions for all of the time-indexed-policy selection  problems collectively provide a solution to the full dynamic collective choice problem. See Section 6 of \cite{bernheim2006power} for an elaboration of this perspective. \label{fn:dynamic}}  For such settings, the deadline is both uncertain and unbounded. As articulated earlier, our perspective is that there often is a known upper bound on the number of rounds, particularly for negotiations over time-indexed policies.
Even when the deadline is uncertain ex ante, our results apply if it becomes known during deliberations. \Cref{Theorem-Capricious} implies that revelation of the deadline $T_\delta$ rounds in advance allows the agenda setter to obtain a payoff within $\delta$ of her maximum. More starkly, \Cref{Theorem-ThreeRounds} implies that in (essentially) any \hyperlink{DistributionProblem}{Distribution Problem}, three rounds of advance notice concerning the deadline allows the agenda setter to obtain her favorite policy. Both of these conditions strike us as modest, particularly when negotiations are relatively frictionless: if players learn the deadline $\varepsilon>0$ units of time in advance, then with sufficiently short proposal rounds, there will be at least three rounds left, and potentially many more.

Setting aside the question of applicability, our analysis allows us to characterize the effect of the number of proposal rounds on the agenda setter's power. To that end, let $U_T(x^0)$ denote the agenda setter's equilibrium payoff in the finite-horizon game with $T$ rounds and initial default $x^0$, and let $U_\infty(x^0)$ denote that payoff in the infinite-horizon game.\footnote{Formally, \Cref{lemma:structural} implies that $U_T (x^0) = u_A (\phi^T(x^0))$  and \cite{diermeier2011legislative,diermeier2012characterization} and \cite{anesi2014bargaining} imply that $U_\infty (x^0) = u_A ( \psi(x^0;V) )$, where $\psi(x^0;V)$ is the agenda setter's favorite improvement on $x^0$ among policies in the stable set, $V$ (see \Cref{eqn:psi-def} for a formal definition).} We obtain the following characterization by analyzing properties of the unimprovable and stable sets.

\begin{theorem}\label{Theorem-Horizon}
Suppose the collective choice problem $\CCP$ satisfies \hyperlink{GFA}{Generic Finite Alternatives}.
For every $x^0 \in X$ and $T' > T \geq 1$, we have
\[
U_{T'}(x^0)\geq U_T(x^0)\geq U_\infty (x^0).
\]
Moreover, exactly one of the following two statements holds:
\begin{enumerate}[noitemsep,label={\normalfont (\alph*)}]
\item There exists some $x^0\in X$ such that $U_T(x^0)> U_1(x^0)> U_\infty (x^0)$ for all $T \geq 2$.
\item For all $x^0\in X$ and $T \geq 2$, $U_T(x^0)= U_1(x^0)= U_\infty (x^0)$.
\end{enumerate}
\end{theorem}
\Cref{Theorem-Horizon} shows that the agenda setter either (a) benefits most from having multiple (but finite) rounds and least from having infinite rounds, or (b) is indifferent about the number of rounds.
Thus, her payoff is non-monotone in the number of rounds, except when she cannot benefit from an ability to revisit any one-round proposal. This non-monotonicity suggests that an agenda setter may benefit from creating a deadline even if one does not arise naturally. She might accomplish this objective by creating a ``crisis'' to instill urgency or by bundling the policy issue of interest with a separate time-indexed matter (e.g., a deadline for raising the debt ceiling).

\section{Conclusion}\label{Section-Conclusion}

We have shown that agenda setters have dictatorial power in collective choice problems with two features. The first is that the agenda setter proposes policies in real-time without commitment, tailoring her current proposal to the prevailing default option. The second is a widely satisfied manipulability condition that ensures the existence of one-step improvements. 

Our analysis contributes to a literature that seeks to understand why legislative institutions concentrate political power in the hands of agenda setters, and why majority rule may not be an effective safeguard. To this end, our results also highlight how the agenda setter benefits from bundling policies with transfers and pork, or by linking unrelated policy issues. Finally, we have shown that when the agenda setter makes proposals in real time, many bargaining protocols have equivalent implications for equilibrium outcomes.

\appendix
\renewcommand{\thesection}{\Alph{section}}

\begin{singlespace}
	\addcontentsline{toc}{section}{References}
	\bibliographystyle{ecta}
	\bibliography{ABBC.bib}
\end{singlespace} 
\addtocontents{toc}{\protect\setcounter{tocdepth}{1}} 

\section{Appendix}
\subsection{Proof of \Cref{lemma:structural} on p. \pageref{lemma:structural}}
We prove that $f_T(x^0)=\{\phi^T(x^0)\}$ by induction.\smallskip

\noindent\textbf{Base Step:} \emph{If $T =1$, then $f_T(x^0) = \left\{ \phi(x^0) \right\}$ for all $x^0 \in X$.} 

\noindent Rejection of any proposal results in $x^0$ being chosen. Therefore, in every equilibrium, any proposal $y \succ_M x^0$ passes with probability 1. Thus, if $x^0$ is improvable, proposing $\phi(x^0)$ with probability 1 is uniquely optimal for the agenda setter, and hence must occur in every equilibrium. If $x^0$ is unimprovable, then in equilibrium, any proposal that the agenda setter makes results in $x^0$ being chosen. In both cases, the result follows. \smallskip

\noindent\textbf{Inductive Step:} \emph{Given any $T \in \mathbb{N}$ and $x^0 \in X$, if $f_{T-1}(x) = \left\{ \phi^{T-1}(x) \right\}$ for all $x \in X$, then $f_T(x^0) = \left\{ \phi^T(x^0) \right\}$.} 

\noindent Consider the procedure with $T$ rounds and initial default $x^0$. By subgame perfection, $f_{T-1}(x)$ is the set of outcomes arising with positive probability in any equilibrium in any subgame with $T-1$ rounds in which $x$ is the prevailing default after the first round.
Therefore, by the inductive hypothesis, passage of proposal $y$ in the first round results in outcome $\phi^{T-1}(y)$ and rejection of that proposal results in outcome $\phi^{T-1}(x^0)$. Thus, in every equilibrium, any proposal $y$ where $\phi^{T-1}(y) \succ_M \phi^{T-1}(x^0)$ will pass with probability 1. Therefore: 
\begin{itemize}
    \item if $\phi^{T-1} (x^0)$ is improvable, then in every equilibrium, the agenda setter proposes some $y$ whose continuation outcome $\phi^{T-1}(y)$ coincides with $\phi\left( \phi^{T-1} (x^0) \right)$, which is $\phi^T (x^0)$. Note that $\phi(x^0)$ is one such proposal because $\phi$ and $\phi^{T-1}$ commute. If multiple such proposals exist, then she may randomize among them.
    \item if $\phi^{T-1} (x^0)$ is unimprovable, then it is an element of $E$ and therefore, $\phi^{T-1} (x^0) = \phi^{T} (x^0)$). Thus, any proposal that the agenda setter makes in equilibrium results in $\phi^{T-1} (x^0)$ being chosen; again, there may be multiple such proposals.
\end{itemize}
In either case, $f_T (x^0) = \phi^T (x^0)$. \medskip

\noindent\Cref{lemma:structural}(a) follows immediately from above; 
(b) follows immediately from above and the definition of $E$. For (c), observe that the inclusion $\bigcup_{x^0 \in X} f_T (x^0) \supseteq E$ follows immediately from (b), while the opposite inclusion $\bigcup_{x^0 \in X} f_T (x^0) \subseteq E$ follows from the argument above, together with the fact that $\phi^T(x^0) \in E$ for $T\geq |X|-1$.

\subsection{Details and Proofs for \cref{Theorem-epsilongrid}  on p. \pageref{Theorem-epsilongrid}}\label{appendix-theorem2-full}

Our argument proceeds in three steps. First, we present a \emph{uniform improvement} lemma that establishes that in any manipulable collective choice problems, whenever the agenda setter's payoff is more than $\delta>0$ away from that of her favorite policy, there is another policy that improves her payoff and those of a majority of voters by at least $\eta_\delta>0$; this step plays a critical role in showing that the agenda setter can obtain a payoff within $\delta$ of her highest payoff, $u_A^*$, with a uniform bound on the number of rounds. The second step formalizes the assertion that \hyperlink{TI}{Thin Individual Indifference} characterizes collective choice problems that admit arbitrarily fine generic $\epsilon$-grids. With these steps in place, we then prove \Cref{Theorem-epsilongrid}.

\subsubsection{A Uniform Improvement Lemma}
For each $\delta >0$, define $\Gamma_{\delta} := \left\{ x \in \policyspace \mid u_A^* \geq u_A(x) + \delta \right\}$. We say that policies in $\Gamma_\delta$ are \emph{$\delta$-suboptimal} for the agenda setter, while policies in $X \backslash \Gamma_\delta$ are \emph{$\delta$-optimal} for her. For each $x \in \policyspace$ and $\eta >0$, define
\begin{align*}
Q(x, \eta):= \big\{ y \in X \mid & u_A(y) \geq u_A(x) + \eta \text{ and}\\
&\text{$\exists$ majority $S \subseteq N$ such that } u_i(y) \geq u_i(x) + \eta \ \ \forall  i \in S \big\} \notag
\end{align*}
to be the set of policies that lead to a utility improvement of at least $\eta$ for some winning coalition. If $Q(x,\eta)\neq \emptyset$, then we say that $x$ is $\eta$-improvable.
\begin{lemma}\label{lemma:uniform-improvement}
Suppose the collective choice problem $\CCP$ is \hyperlink{Definition:Manipulable}{Manipulable}. Then for every $\delta >0$, there exists $\eta_\delta >0$ such that $Q(x, \eta_\delta) \neq \emptyset$ for all $x \in \Gamma_\delta$.
\end{lemma}
Manipulability implies that any policy that is $\delta$-suboptimal for the agenda setter must be improvable, but does not specify how much the agenda setter and a winning coalition of voters gain from that improvement.
\Cref{lemma:uniform-improvement} asserts that for each $\delta>0$, there is a uniform threshold $\eta_\delta$ such that any policy that is $\delta$-suboptimal for the agenda setter must also be $\eta_\delta$-improvable. This uniformity will be important for establishing the uniform bounds on the number of rounds needed for the agenda setter to achieve $\delta$-optimality in \cref{Theorem-epsilongrid,Theorem-Capricious}.

\begin{proof}[Proof of Lemma \ref{lemma:uniform-improvement}]
Let $\delta >0$ be given. Suppose that $\Gamma_\delta \neq \emptyset$, for otherwise the lemma is vacuously true. Define the map $\eta^* : \Gamma_\delta \to \mathbb{R}_+$ by
\begin{align}\label{eqn:eta}
\eta^*(x) &:= \sup_{\eta \in \mathbb{R}_{+}} \eta 
\quad \text{s.t.} \quad Q(x,\eta) \neq \emptyset. 
\end{align}

Note that the supremum in \eqref{eqn:eta} is attained because the correspondence $(x,\eta) \to Q(x, \eta)$ is compact-valued and $x \in Q(x,0)$. Moreover, because $\CCP$ is \hyperlink{Definition:Manipulable}{Manipulable}, for each $x \in \Gamma_\delta$ there exists some $\eta_x >0$ such that $Q(x,\eta_x) \neq \emptyset$. It follows that $\eta^* \left( \Gamma_\delta \right) \subseteq (0,\infty)$. 

We claim that $\eta^*$ is lower semi-continuous.
Let $x^* \in \Gamma_\delta$ be given and take any sequence $\left\{x_n \right\} \subset \Gamma_\delta$ satisfying $x_n \to x^*$. Because preferences are continuous, for every $\epsilon>0$ there exists some $N_\epsilon >0$ such that $n \geq N_\epsilon$ implies $|u_i (x_n) - u_i (x^*)| < \epsilon$ for all players $i$. Letting $\epsilon \in \left(0, \eta^*(x^*) \right)$, which is possible because $\eta^*(x^*) >0$, we therefore have $Q\left(x_n, \eta^*(x^*) - \epsilon \right) \neq \emptyset$ for $n \geq N_\epsilon$. Hence, $\eta^*(x_n) \geq \eta^*(x^*) - \epsilon$ for $n \geq N_\epsilon$. Sending $\epsilon \to 0$, we obtain $\liminf_{n \to \infty} \eta^*(x_n) \geq \eta^*(x^*)$, which establishes the claim.

To conclude the proof, note that $\eta_\delta := \min_{x \in \Gamma_\delta} \eta^*(x)$ is well-defined because $\Gamma_\delta$ is compact and $\eta^*$ is lower semi-continuous, is strictly positive because $\eta^* \left( \Gamma_\delta \right) \subseteq (0,\infty)$, and satisfies $Q(x, \eta_\delta) \neq \emptyset$ for all $x \in \Gamma_\delta$ by construction (recall that the supremum in \eqref{eqn:eta} is attained).
\end{proof}

\subsubsection{Generic $\epsilon$-Grids and Thin Individual Indifference}\label{appendix:grids}

Here we formalize the assertion that \hyperlink{TI}{Thin Individual Indifference} (\Cref{defi:TI}) characterizes collective choice problems that admit arbitrarily fine generic $\epsilon$-grids. The following formalizes what it means to admit arbitrarily fine grids:

\begin{definition}\label{defi:FA}\hypertarget{FA}
A collective choice problem $\CCP = (\policyspace, \{\prefto\}_{i=1,\dots,n,A} \} )$ is \hyperlink{FA}{Finitely Approximable} if, for every $x \in X$ and $\epsilon>0$, there exists a generic $\epsilon$-grid $X_\epsilon$ such that $x \in X_\epsilon$.
\end{definition}

\Cref{defi:FA} requires not only that a generic $\epsilon$-grid $X_\epsilon \subseteq X$ exists for every $\epsilon>0$, but also that such a grid can be constructed so as to contain any pre-specified point $x$ in the ambient policy space $X$. \Cref{Lemma-FA-TI-equiv} shows that finite approximability is characterized by \hyperlink{TI}{Thin Individual Indifference}.

\begin{lemma}\label{Lemma-FA-TI-equiv}
A collective choice problem $\CCP$ is \hyperlink{FA}{Finitely Approximable} if and only if it satisfies \hyperlink{TI}{Thin Individual Indifference}.
\end{lemma}

As the proof of this result is technical and somewhat involved, we relegate it to the Supplementary Appendix. 

\subsubsection{Proof of \cref{Theorem-epsilongrid}}\label{appendix:epsilongrid-proof}

For any $\epsilon >0$ and generic $\epsilon$-grid $\policyspace_\epsilon$, we denote the corresponding discretized collective choice problem by $\CCP_\epsilon := (\policyspace_\epsilon, \left\{\prefto\}_{i=1,\dots,n,A} \} \right)$. 
We define two maps analogous to the definitions in \Cref{subsection:finite}. The agenda setter's favorite improvement within grid $\policyspace_\epsilon$, denoted by $\phi(\cdot ; \policyspace_\epsilon) : \policyspace_\epsilon \to \policyspace_\epsilon$, is
\begin{align}\label{Equation-DefinitionPhiGrid}
\left\{\phi(x ; \policyspace_\epsilon)\right\} := \argmax_{y \in M(x) \bigcap X_\epsilon} u_A(y)
\end{align}
where, as in \Cref{subsection:finite}, $M(x) := \{y \in X : y \succ_M \text{ or } y = x\}$. The second map is $f_T(\cdot; \policyspace_\epsilon) : \policyspace_\epsilon \rightrightarrows \policyspace_\epsilon$, which denotes the equilibrium outcome correspondence (as defined in \Cref{subsection:finite}) for $\CCP_\epsilon$. With these definitions in hand, we prove each direction of \cref{Theorem-epsilongrid} in turn.

\medskip\noindent\textbf{Sufficiency of Manipulability for Approximate Dictatorial Power.} Suppose that $\CCP$ is \hyperlink{Definition:Manipulable}{Manipulable}. As $\CCP$ satisfies \hyperlink{TI}{Thin Individual Indifference},  \cref{Lemma-FA-TI-equiv} assures that   for each $\epsilon>0$, there exists a generic $\epsilon$-grid $\policyspace_\epsilon$. 

Let $\delta>0$ be given and let $\eta_\delta$ be as defined in \cref{lemma:uniform-improvement}, which applies because the collective choice problem $\CCP$ is \hyperlink{Definition:Manipulable}{Manipulable}. Let $\epsilon_\delta >0$ be such that 
\begin{equation}
\max_{i \in N \cup \left\{A \right\} } \max_{x \in \policyspace} \max_{y \in B_{\epsilon{_\delta}}(x)} |u_i (x) - u_i(y)| \leq \frac{\eta_{\delta}}{2}, \label{eqn:fa-pf1}
\end{equation}
where $B_{\epsilon{_\delta}}(x) := \left\{ y \in \policyspace : d(y, x) < \epsilon_\delta\right\}$, noting that such an $\epsilon_\delta >0$ exists because each $u_i$ is uniformly continuous (being that $\policyspace$ is compact). For each $\epsilon < \epsilon_\delta$,  $\policyspace_\epsilon \cap B_{\epsilon_\delta}(x) \neq \emptyset$ for all $x \in \policyspace$ by construction. Therefore, \eqref{eqn:fa-pf1} implies that 
\begin{equation}
\max_{i \in N \cup \left\{A \right\} } \max_{x \in \policyspace} \min_{y \in \policyspace_\epsilon} |u_i (x) - u_i(y)| \leq \frac{\eta_{\delta}}{2}. \label{eqn:fa-pf2}
\end{equation}
 Henceforth, we consider $\epsilon<\epsilon_\delta$. 

We first claim, building on \Cref{lemma:uniform-improvement}, that for every $\delta$, there exists an $\eta_\delta$ such that every policy in $X_\epsilon$ that is $\delta$-suboptimal for the agenda setter is $\eta_\delta/2$-improvable in $\policyspace_\epsilon$, viz. there exists an alternative in $X_\epsilon$ that leads to a utility increase of at least $\eta_\delta/2$ for herself and some majority of voters. Formally:
\begin{align}\label{Inequality-DiscreteGrid}
    \text { For every }x\in \policyspace_\epsilon\cap \Gamma_\delta, \quad Q(x,\eta/2)\cap\policyspace_\epsilon \neq \emptyset.
\end{align}
To see why \eqref{Inequality-DiscreteGrid} is true, first observe that by \Cref{lemma:uniform-improvement}, there exist $\eta_\delta>0$ and $y'\in \policyspace$ such that $u_i(y')\geq u_i(x)+\eta_\delta$ for every $i\in S\cup \{A\}$, where $S\subseteq N$ contains some majority of voters. Second, \eqref{eqn:fa-pf2} assures that there exists some $y \in X_\epsilon$ such that $|u_i (y') - u_i(y)| \leq \eta_\delta/2$ for every $i$. Combining these observations, we conclude that $u_i (y) \geq u_i (x) + \eta_\delta/2$ for all $i\in S \cup \left\{A\right\}$. 

Because $y \sprefto_M x$ above, an important implication of \eqref{Inequality-DiscreteGrid} is that
\begin{align}\label{Inequality-DiscreteAgendaSetter}
 \text { for every }x\in \policyspace_\epsilon\cap \Gamma_\delta, \quad   u_A(\phi(x;\policyspace_\epsilon))\geq u_A(x)+\frac{\eta_\delta}{2},
\end{align}
where the map $\phi(\cdot ; \policyspace_\epsilon)$ is the agenda setter's favorite improvement in grid $X_\epsilon$, as defined in \Cref{Equation-DefinitionPhiGrid}. 

We use this fact to prove the theorem: there exists some (uniform) $T_\delta \in \mathbb{N}$ such that, if there are $T \geq T_\delta$ rounds, then the agenda setter's payoff is no lower than $u^*_A - \delta$ in every equilibrium for any generic $\epsilon$-grid $\policyspace_\epsilon$ with  $\epsilon<\epsilon_\delta$.\footnote{We note that this statement does not follow from \Cref{lemma:structural}. \Cref{lemma:structural} implies that if there are $T \geq |\policyspace_\epsilon| -1$ rounds, then $\bigcup_{x^0 \in  \policyspace_\epsilon} f_T(x^0 ; \policyspace_\epsilon) = E\left( \policyspace_\epsilon \right)$, where $E\left( \policyspace_\epsilon \right) := \left\{ x \in \policyspace_\epsilon : x = \phi(x; \policyspace_\epsilon) \right\}$ denotes the set of unimprovable policies in $\mathcal{C}_\epsilon$. We know from \eqref{Inequality-DiscreteGrid} that $E\left( \policyspace_{\epsilon} \right) \subseteq X_\epsilon \backslash \Gamma_\delta$, i.e., any policy that is unimprovable in $\CCP_\epsilon$ must be $\delta$-optimal for the agenda setter. It would then follow that the agenda setter's payoff is at least $u^*_A - \delta$ when there are $T \geq |\policyspace_\epsilon| -1$ rounds; as $\epsilon\rightarrow 0$, this argument would then require the number of rounds to grow without bound. \label{fn:unimprovable-grid}} To put it formally, there exists $T_\delta$ such that for every $T\geq T_\delta$ and $\epsilon<\epsilon_\delta$, 
\begin{align*}
    \bigcup_{x^0 \in  \policyspace_\epsilon} f_T(x^0 ; \policyspace_\epsilon) \subseteq  \policyspace_\epsilon\backslash\Gamma_\delta.
\end{align*}

If $x^0 \in \policyspace_\epsilon\backslash\Gamma_\delta$, the conclusion follows from applying \Cref{lemma:structural}(a) to $\CCP_\epsilon$, noting that $\phi(x;\policyspace_\epsilon)\prefto_A x$ for every $x$. Thus we consider $x^0 \in \policyspace_\epsilon\cap \Gamma_\delta$. We denote the payoff difference between the agenda setter's favorite and least favorite policies by $\Delta := u_A^* - \min_{y \in \policyspace} u_A(y)$, which is well-defined and finite because $u_A$ is continuous and $\policyspace$ is compact. Correspondingly, define $T_\delta := \lceil 2 \Delta / \eta_\delta \rceil \in \mathbb{N}$. Suppose, towards a contradiction, that $y:= f_{T_\delta}(x^0; \policyspace_\epsilon) \in  \policyspace_\epsilon\cap \Gamma_\delta$. Then it follows that
\begin{align*}
u_A \left( \phi(y; \policyspace_\epsilon) \right) - u_A(x^0) &\geq 
u_A \left( y \right) - u_A(x^0) + \frac{\eta_\delta}{2} \\
&= u_A \left( \phi^{T_\delta}(x^0; \policyspace_\epsilon) \right) - u_A(x^0) + \frac{\eta_\delta}{2} \\
&= \sum_{t=1}^{T_\delta} \left[ u_A \left( \phi^{t}(x^0; \policyspace_\epsilon) \right) - u_A \left( \phi^{t-1}(x^0; \policyspace_\epsilon)\right) \right] + \frac{\eta_\delta}{2}\\
&\geq T_\delta \cdot \frac{\eta_\delta}{2} + \frac{\eta_\delta}{2}\\
& \geq \Delta + \frac{\eta_\delta}{2},
\end{align*}
where the first line is by \eqref{Inequality-DiscreteAgendaSetter}, the second line is by \Cref{lemma:structural} applied to $\CCP_\epsilon$, the third line is an identity, the fourth line is by another application of \eqref{Inequality-DiscreteAgendaSetter} to each term in the sum (noting that $\phi^{T_\delta}(x^0; \policyspace_\epsilon) \in \policyspace_\epsilon\cap \Gamma_\delta$ implies that $\phi^t(x^0;\policyspace_\epsilon) \in \policyspace_\epsilon\cap \Gamma_\delta$ for all $t < T_\delta$), 
and the final line is by definition of $T_\delta$. However, given that $\eta_\delta>0$, this inequality contradicts the definition of $\Delta$. We conclude that $y \in \policyspace_\epsilon \backslash \Gamma_\delta$, as desired.

\medskip\noindent\textbf{Necessity of Manipulability for Approximate Dictatorial Power.} Suppose that $\CCP$ is not \hyperlink{Definition:Manipulable}{Manipulable}. Then there exists an unimprovable policy $x$ and $\delta>0$ such that $u_A(x)<u_A^*-\delta$. 
As $\CCP$ satisfies \hyperlink{TI}{Thin Individual Indifference}, \cref{Lemma-FA-TI-equiv} assures that there exists an $\overline{\epsilon}>0$ such that, for all $\epsilon \in (0, \overline{\epsilon})$, there is a generic $\epsilon$-grid $X_\epsilon$ for which $x \in X_\epsilon$. Observe that $x$ must also be unimprovable in the corresponding discretized collective choice problem $\CCP_\epsilon$. Applying \cref{lemma:structural} to this discretized problem reveals that for every number of rounds, the equilibrium outcome starting from initial default $x^0 = x$ is $x$ itself: $f_T(x; X_\epsilon) = \{x\}$ for every $T \in \mathbb{N}$. The agenda setter then attains a payoff of $u_A(x)<u_A^*-\delta$, failing to achieve approximate dictatorial power regardless of the number of rounds.

\subsection{The Divide-the-Dollar Problem}\label{Section-DivideDollar}

Herein, we describe the implications of real-time agenda control in the standard ``Divide-the-Dollar'' problem, in which the policy space is $X=\Delta^{n+1}$ and a policy $x\equiv(x_1,\ldots,x_n,x_{n+1})$ reflects a division of the dollar; the first $n$ indices are the shares of the $n$ voters and $x_{n+1}$ is that of the agenda setter. Each player has selfish risk-neutral preferences, so $u_i(x)=x_i$. The legislature begins with an initial default option $x^0$, and as in our baseline analysis, uses simple majority rule in each of finitely many rounds. 

In this context, we elucidate two features of our general analysis. First, we construct a non-capricious equilibrium in which the agenda setter appropriates the entire dollar whenever there are $3$ or more rounds. Second, we highlight how our dictatorial power result (\Cref{Theorem-Capricious}) does not apply to equilibria with capricious tiebreaking: regardless of the number of rounds, there exists an equilibrium with capricious tiebreaking in which the agenda setter fails to appropriate the entire dollar. 

\medskip\noindent\textbf{A Non-Capricious Equilibrium.}
To describe a non-capricious equilibrium, we adapt the agenda setter's favorite improvement operator $\phi$ from \Cref{subsection:finite} to this setting. For default policy $x$, let $\beta(x)$ denote the policy that sets the $(n-1)/2$ largest elements (among the first $n$ elements) to $0$ and reallocates that portion of the dollar to the agenda setter; in the event of ties, $\beta(x)$ selects a group of voters with this size with the lowest player indices. More precisely, let $G^0(x)\equiv\emptyset$, and define $G^k(x)$ inductively for $k \in \{1,\dots,n\}$ as follows:
\begin{align*}
    G^{k}(x):=G^{k-1}(x)\cup\left\{j\in N:j=\min \left(\argmax_{i\in N\backslash G^{k-1}(x)} x_i\right) \right \}.
\end{align*}
Observe that $G^{k}(x)$ identifies the $k$ voters who have the highest shares in default policy $x$ (and breaks ties in favor of those with lower player labels). We define the policy $\beta(x)$ as 
\begin{align*}
    (\beta(x))_i:=\begin{cases}
			0, & \text{if }i\in G^{(n-1)/2}(x),\\
			x_i, & \text{if } i\in N\backslash G^{(n-1)/2}(x),\\
            x_{n+1}+\sum_{j\in G^{(n-1)/2}(x)} x_j, & \text{if }i=n+1.
		 \end{cases}
\end{align*}
This operator adapts the favorite improvement operator $\phi$ to this setting: among policies that a majority of voters weakly prefer to $x$, $\beta(x)$ is one of the agenda setter's favorites. Observe that for any policy $x$, $\beta^2(x)$ extracts bargaining shares from all but one voter---the one who has the lowest share in policy $x$---and $\beta^3(x)$ has the agenda setter obtaining the entire dollar.

We now construct a non-capricious equilibrium in which the agenda setter obtains the entire dollar if there are  $T\geq 3$ rounds. Consider the following strategy profile: in each round $t \in \{1,\dots,T\}$, if the prevailing default is $x$, then (i) the agenda setter proposes $\beta(x)$ whenever $x$ is the prevailing default and (ii) each voter $i \in N$ votes in favor of a proposal $y$ if and only if $\beta^{T-t}(y) \prefto_i \beta^{T-t}(x)$, viz., she weakly prefers the continuation outcome from acceptance to the continuation outcome from rejection. As no player has a strictly profitable deviation and the strategy profile is pure and Markovian, this defines a \hyperlink{NonCapricious}{Non-Capricious} equilibrium by construction. 

We illustrate the path of play in this equilibrium using the following example:
\begin{example}\label{Example-3VoterBargaining}
Suppose that there are three voters and the default option $x^0$ is such that $x^0_1>x^0_2>x^0_3>0$. Then $\beta(x^0)=(0,x^0_2,x^0_3,1-x^0_2-x^0_3)$, $\beta^2(x^0)=(0,0,x^0_3,1-x^0_3)$, and $\beta^3(x^0)=(0,0,0,1)$. Observe that voters $1$ and $2$ approve each equilibrium-path proposal: they are indifferent between $\beta^2(x^0)$ and $\beta^3(x^0)$, and anticipate that rejecting either the first or second on-path offer nevertheless results in both of them obtaining zero surplus.
\end{example}
The above construction demonstrates a particular non-capricious equilibrium in which the agenda setter obtains the entire dollar within $3$ rounds. \Cref{Theorem-ThreeRounds} further implies that \emph{all} non-capricious equilibria of this \hyperlink{DistributionProblem}{Distribution Problem} share this property. 

\medskip\noindent\textbf{A Capricious Equilibrium.} We now show, using a setting with $3$ voters, that the dictatorial power conclusion of \Cref{Theorem-Capricious} does not apply to equilibria with capricious tiebreaking. 

Consider a strategy profile that differs from that described above only with respect to voters' tiebreaking rule: voters now resolve indifference in favor of the agenda setter's proposal if and only if it is in the final or penultimate round (i.e., $t \in\{T-1,T\}$), and otherwise break ties in favor of the prevailing default option. Observe that this strategy profile satisfies \Cref{defi:NC}(a), as it is pure and Markovian. However, it violates \Cref{defi:NC}(b) because, for any pair of continuation outcomes, the tiebreaking decision conditions on the current round. 

We claim that, for \emph{any} initial default $x^0$ and number of rounds $T\geq 2$, this strategy profile (i) results in the outcome $\beta^2(x^0)$ and (ii) is an equilibrium. Consequently, given any default $x^0$ in which all voters have positive shares (as in \Cref{Example-3VoterBargaining}) and any number of rounds, the agenda setter fails to appropriate the entire dollar in this (capricious) equilibrium.

To see that the outcome is $\beta^2(x^0)$, observe that---as in \Cref{Example-3VoterBargaining}--- the outcome is $\beta(x^0)$ if $T=1$ and $\beta^2(x^0)$ if $T=2$. Suppose now that $T=3$. Voters anticipate that accepting the initial on-path proposal $\beta(x^0)$ results in outcome $\beta^3(x^0)$, whereas rejecting it leads to $\beta^2(x^0)$. As the two voters with the largest shares (and lowest indices) in $x^0$ both receive zero shares under both outcomes,\footnote{This pair consists of voters 1 and 2 for the initial default in \Cref{Example-3VoterBargaining}, but may be different for other initial defaults.} our capricious tiebreaking rule stipulates that they both vote against the initial proposal---unlike in \Cref{Example-3VoterBargaining}---resulting in an on-path outcome of $\beta^2(x^0)$. It is then easy to see by induction that the same outcome arises for all $T\geq 3$.

We now argue that this strategy profile is an equilibrium. It suffices to consider the agenda setter's incentives in rounds $t \leq T-2$.\footnote{From the non-capricious equilibrium construction above, it is easy to see that the continuation strategies in any round-$T$ or round-$(T-1)$ subgame constitute an equilibrium therein. Moreover, voters have no profitable deviations in any round by construction.} Consider round $T-2$ and let the prevailing default $x$ be given. Suppose, towards a contradiction, that the agenda setter has a strictly profitable deviation by proposing some $y \neq \beta(x)$. By the argument in the preceding paragraph, (i) the continuation outcome is $\beta^2(x)$ if the agenda setter follows her strategy of proposing $\beta(x)$, and (ii) players anticipate that acceptance of $y$ results in outcome $\beta^2(y)$ whereas its rejection leads to $\beta^2(x)$. By the supposition and property (i), it must be that some voter---call her $i$---is strictly worse off under $\beta^2(y)$ than under $\beta^2(x)$. By property (ii), voter $i$ must vote against this proposal. Moreover, by the construction of $\beta$, both continuation outcomes result in at least two voters obtaining zero utility; although the identities of these voters may differ across these outcomes, as there are three voters in total, the pigeonhole principle implies that at least one voter---call her $j$---obtains a zero share in \emph{both} $\beta^2(y)$ and $\beta^2(x)$. Our capricious tiebreaking rule and property (ii) then stipulate that voter $j$ rejects the proposal $y$. Hence, we have found two distinct voters, $i$ and $j$, who both vote against proposal $y$, implying that it does not pass and thereby contradicting that it is a strictly profitable deviation for the agenda setter. By induction, this argument also applies for all rounds $t\leq T-2$.\footnote{The reader may wonder as to why the agenda setter cannot induce voters to break ties in favor of her proposal in rounds $t\leq T-2$ by offering $\epsilon>0$ more to each voter, as in models of legislative bargaining based on the closed-rule procedure. The issue is that here, unlike closed-rule bargaining, such policies are left open to further reconsideration; voters (correctly) anticipate that bargaining in rounds $T-1$ and $T$ leave at most one voter with a non-zero surplus.}
\newpage

\section{Supplementary Appendix (For Online Publication)}\label{Section-SupplementaryAppendix}
\subsection{Proof of \Cref{Lemma-FA-TI-equiv} on p. \pageref{Lemma-FA-TI-equiv}}
We consider each direction in turn. For the ``if'' direction, suppose $\CCP$ satisfies \hyperlink{TI}{Thin Individual Indifference}. Let $x \in X$ and $\epsilon>0$ be given. Note that because $X$ is compact, the open covering $\left\{ B_{\epsilon/2}(y) \right\}_{y \in X}$ has a finite subcovering; enumerate it by $\{B_k\}_{k=1}^K$ and suppose, without loss of generality, that $x \in B_1$. For any $y \in X$, let $\mathcal{I}(y) := \bigcup_{i \in N \cup\{A\}} I_i(y)$. Recursively construct the sequences $\{D_k\}_{k=1}^K \subseteq 2^X$ and $\{x_k\}_{k=1}^K \subseteq X$ as follows:
\begin{itemize}
    \item Let $D_0 := \emptyset$ and $D_k := D_{k-1} \bigcup \left[\mathcal{I}(x_k) \backslash \{x_k\}\right]$ for $k \geq 1$;
    \item Let $x_1 := x \in B_1$ and pick $x_k \in B_k \backslash D_{k-1} $ arbitrarily for $k \geq 2$. 
\end{itemize}
We claim that $X_\epsilon := \{x_k\}_{k=1}^K$ is a generic $\epsilon$-grid; since $x \in X_\epsilon$ by construction, this suffices to prove the lemma. We establish the claim in three steps. 

\noindent \emph{Step 1: The sequences $\{D_k\}_{k=1}^K$ and $\{x_k\}_{k=1}^K$ are well-defined, viz., $B_k \backslash D_{k-1} \neq \emptyset$ for all $k$.} The argument is by induction. For the base step, note that $D_1$ has empty interior by \hyperlink{TI}{Thin Individual Indifference}; because $B_2$ is nonempty and open, this implies that $B_2 \backslash D_1 \neq \emptyset$ and thus that $x_2$ is well-defined. By construction, $x_2 \in \mathcal{I}(x_1)$ only if $x_2 = x_1$. For the inductive step, let $2\leq k \leq K-1$ be given and suppose that, for all $\ell \leq k$, $D_{\ell}$ and $x_{\ell}$ are well-defined and satisfy the following: 
\begin{equation}
\text{$x_\ell \in \bigcup_{j=1}^{\ell-1} \mathcal{I}(x_j)$ $\implies$ $x_\ell \in \{x_j\}_{j=1}^{\ell-1}$.} \label{indiff-recur}
\end{equation}
To complete the induction, we must show that (a) $B_{k+1} \backslash D_k \neq \emptyset$ and (b) every $x_{k+1} \in B_{k+1} \backslash D_k$ satisfies \eqref{indiff-recur} (with $\ell = k+1$). Towards (a), note that \eqref{indiff-recur} holding for all $\ell \leq k$ immediately implies that: 
\begin{enumerate}
    \item[(i)] For all $j, \ell \leq k$, either $\mathcal{I}(x_j) = \mathcal{I}(x_\ell)$ or $\mathcal{I}(x_j) \cap \mathcal{I}(x_\ell) = \emptyset$;
    \item[(ii)] $D_k = \left[\bigcup_{j=1}^{k} \mathcal{I}(x_j) \right] \backslash \{x_1, \dots, x_k\}$.
\end{enumerate}
By (i), the union in (ii) is essentially disjoint (modulo repetition of the same set), and by continuity of preferences, each $\mathcal{I}(x_j)$ is closed. It follows that $\text{int} \left( \bigcup_{j=1}^{k} \mathcal{I}(x_j)\right) =\bigcup_{j=1}^{k} \text{int} \left(\mathcal{I}(x_j) \right)$. Since $\text{int} \left(\mathcal{I}(x_j)\right) \subseteq \left\{x_j\right\}$ for all $j\leq k$ by \hyperlink{TI}{Thin Individual Indifference}, we conclude from (ii) that $\text{int}\left( D_k \right) = \emptyset$. This establishes that $B_{k+1} \backslash D_k \neq \emptyset$ and thus completes the proof of (a). The proof of (b) is then immediate from (ii), concluding the proof of Step 1.

\noindent \emph{Step 2: All players have strict preferences on $X_\epsilon$.} This is immediate from the fact that \eqref{indiff-recur} holds for all $1\leq k \leq K$, as established in the proof of Step 1 above.

\noindent \emph{Step 3: $X_\epsilon$ is an $\epsilon$-grid, viz., $\max_{y \in X} d(y,X_\epsilon) < \epsilon$.} Recall that $\left\{B_k\right\}_{k=1}^K$ is a covering of $X$ by open balls of radius $\epsilon/2$, while $X_\epsilon$ is constructed from a selection $k \mapsto x_k \in B_k$. Hence, $\sup_{y \in B_k} d(y,x_k) < \epsilon$ for every $k \leq K$. Observe that 
\begin{align*}
    \max_{y \in X} d(y,X_\epsilon) 
\leq \max_{k \leq K} \left[ \sup_{y \in B_k} d(y,x_k) \right] < \epsilon.
\end{align*}

This completes our proof of the sufficiency of \hyperlink{TI}{Thin Individual Indifference}. 
To establish its necessity---the ``only if'' direction---suppose that $\CCP$ violates \hyperlink{TI}{Thin Individual Indifference}. Then there exists a policy $x \in X$, player $i \in N \cup \{A\}$, and nonempty open set $O \subset X$ such that $O \subset I_i(x)\backslash\{x\}$. Pick $y \in O$ and $\delta>0$ so that $B_\delta (y):= \{z \in X : d(y,z) <\delta \} \subseteq O$. Given any $\epsilon  \in (0, \delta]$, let $X_\epsilon \subset X$ be a (not necessarily generic) $\epsilon$-grid for which $x \in X_\epsilon$. There exists some $z \in X_\epsilon \cap B_\delta(y)$ by definition of $X_\epsilon$,\footnote{If not, then $\min_{z \in X_\epsilon} d(y,z) \geq \delta$ and so $X_\epsilon$ would not be an $\epsilon$-grid with $\epsilon \leq \delta$.} and hence $z \in X_\epsilon \cap I_i(x) \backslash \{x\}$ by definition of $B_\delta(y)$, implying that player $i$'s preferences are not strict on $X_\epsilon$. It follows that $\CCP$ does not admit any generic $\epsilon$-grid containing $x$ with $\epsilon \in (0, \delta]$, and therefore is not \hyperlink{FA}{Finitely Approximable}.\qed

\subsection{Details and Proofs for \Cref{Theorem-Capricious} on p. \pageref{Theorem-Capricious}}\label{appendix-theorem3-full}

We provide the main proof of \Cref{Theorem-Capricious} in \Cref{appendix:thm3-proof}. The key lemmas presented therein, \Cref{Lemma-NC-full2,Lemma-NC-full}, are proved separately in \Cref{Section-ProofLemmaNCfull2,Section-ProofLemmaNCfull}.

\subsubsection{Proof of \Cref{Theorem-Capricious}}\label{appendix:thm3-proof}

\paragraph{Step 1: Preliminaries.} We define the \emph{weak} majority acceptance correspondence by $M^\text{w}(x) := \left\{ y \in X : y \prefto_M x \right\}$, the \emph{strict} majority acceptance correspondence by $M^\text{s}(x) := \left\{ y \in X : y \succ_M x \right\}$, and the \emph{almost-strict} majority acceptance correspondence by $M^{\text{as}}(x) := \text{cl}\left[M^{\text{s}}(x)\right] \cup \{x\}$. Define the agenda setter's \emph{favorite almost-strict improvement} value function $V^\text{as}_A : X \to \mathbb{R}$ by $V_A^{\text{as}}(x) \equiv \max_{y \in M^{\text{as}}(x)} u_A(y)$, and her \emph{one-round improvement} correspondence $\Phi^\text{or} : X \rightrightarrows X$ by
\begin{equation}
\Phi^\text{or}(x) \equiv \left\{ y \in X : y \in M^\text{w}(x) \text{ and } u_A(y) \geq V^{\text{as}}_A(x) \right\}. \label{Equation-PhiOR}
\end{equation}
Notice that, because the policy space is compact and all players' preferences are continuous, each of the correspondences described above is nonempty- and compact-valued. We denote the set of unimprovable policies---as in \Cref{defi:Improv}---by $\mathcal{E}$.

In settings with \hyperlink{GFA}{Generic Finite Alternatives}, $\Phi^\text{or}(x) = \{\phi(x)\}$, where $\phi$ is the favorite improvement function defined in \Cref{Equation-FavoriteImprovement}. \Cref{Lemma-General1} shows that $\Phi^\text{or}$ is the appropriate generalization of $\phi$ to general collective choice problems.

\begin{lemma}\label{Lemma-General1}
For any collective choice problem $\CCP$, the following hold:
\begin{enumerate}[noitemsep,label={\normalfont (\alph*)}]
    \item The set of unimprovable policies satisfies $\mathcal{E} = \{ x \in X : x \in \Phi^{\text{or}}(x) \}$.
\item For any pair of policies $x$ and $y$, we have $y \in \Phi^\text{or}(x)$ if and only if $y$ is the outcome of some \hyperlink{NonCapricious}{Non-Capricious} equilibrium of the one-round game with initial default $x$.
\end{enumerate}
\end{lemma}

\begin{proof}
To prove part (a), let $x \in X$ be given. We show that $x \in \mathcal{E}$ if and only if $x \in \Phi^{\text{or}}(x)$. For the ``if'' direction, note that $x \in \Phi^{\text{or}}(x)$ implies that $u_A(x) \geq V^{\text{as}}_A(x)$, and hence there do not exist any $y \in M^\text{as}(x)$ such that $y \succ_A x$; as $M^\text{s}(x) \subseteq M^\text{as}(x)$, it follows that $x \in \mathcal{E}$. For the ``only if'' direction, suppose that $x \notin \Phi^{\text{or}}(x)$. Then, because $x \in M^\text{w}(x)$, it must be that $u_A(x) < V_A^\text{as}(x)$. Thus, there exists some $y \in M^{\text{as}}(x) \backslash\{x\}$ such that $y \succ_A x$. Because $M^{\text{as}}(x) \backslash\{x\} = \text{cl}\left[ M^{\text{s}}(x) \right]$ by definition, there exists a sequence $\{y^n\} \subseteq M^{\text{s}}(x)$ such that $y^n \to y$ and, being that $\prefto_A$ is continuous, there exists an $N \in \mathbb{N}$ such that $y^n \succ_A x$ for all $n \geq N$. Thus, $x$ is improvable by any $y^n$ with $n \geq N$, implying that $x \notin \mathcal{E}$.

To prove part (b), let $x,y \in X$ be given. For the ``if'' direction, suppose that $y$ is the outcome under a \hyperlink{NonCapricious}{Non-Capricious} equilibrium $\sigma$ in the one-round game with default $x$. Suppose, towards a contradiction, that $y \notin \Phi^\text{or}(x)$. As this implies that $y \neq x$, it must be that $y$ is proposed and accepted with probability one under $\sigma$. There are two cases. First, if $y \notin M^\text{w}(x)$, then there exists some voter $i \in N$ such that $i$ votes to approve proposal $y$ under $\sigma$, and yet $x \succ_i y$;  voter $i$ then has a strictly profitable deviation by voting to reject $y$. 
Second, suppose that $u_A(y) < V^{\text{as}}_A (x)$. Then there exists some $z \in M^{\text{s}}(x)$ such that $z\succ_A y$. The agenda setter then has a strictly profitable deviation by proposing $z$ instead of $y$, as every policy in $M^{\text{s}}(x)$ must be accepted by a majority of voters in every equilibrium. In either case, we obtain the contradiction that $\sigma$ is  not an equilibrium. 

For the ``only if'' direction, suppose that $y \in \Phi^\text{or}(x)$. Consider the following pure strategy profile in the one-round game with initial default $x$: the agenda setter proposes $y$ and each voter $i \in N$ votes to accept a proposal $z$ if and only if either (i) $z \succ_i x$ or (ii) $z \sim_i x$ and $z = y$. By construction, no voter has a profitable deviation; the agenda setter's payoff from proposing $y$ is $u_A(y) \geq V^{\text{as}}_A(x)$, while her payoff from any other proposal is bounded above by $\max\{u_A(x), \sup_{z \in M^{\text{s}}(x)} u_A(z) \} \leq  V^{\text{as}}_A(x)$, so that she also has no profitable deviations. Therefore, this strategy profile is an equilibrium. As every pure-strategy equilibrium of any one-round game is \hyperlink{NonCapricious}{Non-Capricious}, this strategy profile is a \hyperlink{NonCapricious}{Non-Capricious} equilibrium inducing outcome $y$.
\end{proof}

\medskip\noindent\textbf{Step 2: Non-Capricious Equilibrium Outcomes.} We now characterize non-capricious equilibrium outcomes in general collective choice problems. Let $\Sigma^\text{NC}(x^0,T)$ denote the set of non-capricious equilibria of the game with $T$ rounds and initial default $x^0$. Let $g_T^\sigma(x^0) \in X$ denote the outcome induced by equilibrium $\sigma \in \Sigma^\text{NC}(x^0,T)$, and $G_T(x^0) \equiv \bigcup_{\sigma \in \Sigma(x^0,T)} \left\{g^\sigma_{T}(x^0)\right\}$ denote those across all non-capricious equilibria. 

We characterize outcomes for all equilibria using the $\Phi^{\text{or}}$ operator. We say that $\hat\phi:X\rightarrow X$ is a \emph{selection} of $\Phi^\text{or}$ if $\hat\phi(x)\in \Phi^\text{or}(x)$ for every $x\in X$; we denote selections by $\hat\phi(\cdot)\in\Phi^\text{or}(\cdot)$.

\begin{lemma}\label{Lemma-NC-full2}
For any collective choice problem $\CCP$, $x^0 \in X$, and $T \in \mathbb{N}$, the following hold:
\begin{enumerate}[label={\normalfont (\alph*)}]
    \item For any selection $\hat{\phi}(\cdot) \in \Phi^\text{or} (\cdot)$, there exists a \hyperlink{NonCapricious}{Non-Capricious} equilibrium $\sigma \in \Sigma(x^0,T)$ inducing the outcome $g_{T}^\sigma(x^0) = \hat{\phi}^T(x^0)$.
  %
%
    \item For any \hyperlink{NonCapricious}{Non-Capricious} equilibrium $\sigma \in \Sigma(x^0,T)$, there exists a collection $\{\hat{\phi}_t (\cdot) \}_{\tau=1}^T$ of selections $\hat{\phi}_t (\cdot) \in \Phi^{\text{or}}(\cdot)$ such that the equilibrium outcome is given by 
    \[
    g_T^\sigma(x^0) = \left[\hat{\phi}_1 \circ \hat{\phi}_{2} \circ \cdots \circ \hat{\phi}_T \right](x^0)
    \]
    and, for every $x \in X$, we have $\hat{\phi}_t (x) \sim_A \hat{\phi}_{T} (x)$ for all $1 \leq t \leq T$.
    \end{enumerate}
\end{lemma}

To interpret \Cref{Lemma-NC-full2}, observe that in the special case where \hyperlink{GFA}{Generic Finite Alternatives} holds, it reduces to the characterization from \Cref{lemma:structural}, viz., the unique equilibrium outcome of the $T$-round game with initial default $x^0$ is $\phi^T(x^0)$. The more complicated statement here reflects the fact that, in settings with indifference, (i) both voters and the agenda setter can break ties differently across non-capricious equilibria and (ii) the agenda setter can break ties differently across rounds in a given non-capricious equilibrium. We prove part (a) by construction and part (b) using a backward induction argument that leverages non-capriciousness; the proof is in \Cref{Section-ProofLemmaNCfull2}.

The primary import of \Cref{Lemma-NC-full2} is that it implies bounds on the sets of outcomes and agenda setter payoffs across all non-capricious equilibria and initial defaults as the number of rounds becomes large.

 \begin{lemma}\label{Lemma-NC-full}
For any collective choice problem $\CCP$, the following hold: 
    \begin{itemize}
    \item[(a)] For every $x^0 \in X$ and $T \in \mathbb{N}$, we have
    \[
    \mathcal{E} \subseteq \bigcup_{x^0 \in X} G_{T+1}(x^0) \subseteq \bigcup_{x^0 \in X} G_T(x^0).
    \]
\item[(b)] For every $\delta>0$, there exists some $T_\delta \in \mathbb{N}$ such that:
\[
\text{If $T \geq T_\delta$, then $u_A(x) \geq \min_{y \in \mathcal{E}}u_A(y) - \delta$ for all $x \in \bigcup_{x^0 \in X} G_T(x^0)$.}\footnote{Note that $\min_{y \in \mathcal{E}}u_A(y)$ is well-defined because $u_A$ is continuous by assumption and $\mathcal{E}$ is closed, as the definition of improvability (\Cref{defi:Improv}) and continuity of players' preferences imply that the set $X \backslash \mathcal{E}$ of improvable policies is open. \label{footnote-E-closed}}
\]
\end{itemize}
\end{lemma}

\Cref{Lemma-NC-full}(a) establishes that the set of \hyperlink{NonCapricious}{Non-Capricious} equilibrium outcomes---across all such equilibria and all initial defaults---converges monotonically downward to some set $\mathcal{G}^\text{NC}_\infty \supseteq \mathcal{E}$. \Cref{Lemma-NC-full}(b) further shows that $\mathcal{G}^\text{NC}_\infty \subseteq \{x \in X : u_A(x) \geq \min_{y \in \mathcal{E}}u_A(y) \}$. Together, these facts imply that the agenda setter's minimal payoff across all non-capricious equilibria is precisely $\min_{y \in \mathcal{E}} u_A(y)$ in the $T \to \infty$ limit. We prove \Cref{Lemma-NC-full} in \Cref{Section-ProofLemmaNCfull} below.

\medskip\noindent\textbf{Step 3: Main Argument for \Cref{Theorem-Capricious}.} \Cref{Theorem-Capricious}(a) follows immediately from the existence claim in \Cref{Lemma-NC-full2}(a).  We now use \Cref{Lemma-NC-full} to prove \Cref{Theorem-Capricious}(b). 

First, we show that manipulability is sufficient for approximate dictatorial power. Let $\CCP$ be a \hyperlink{Definition:Manipulable}{Manipulable} collective choice problem. Then $\mathcal{E} = X^*_A$ and $\min_{y\in \mathcal{E}} u_A(y) = u^*_A$. Let $\delta>0$ be given. \Cref{Lemma-NC-full}(b) implies that there exists some $T_\delta \in \mathbb{N}$ such that the agenda setter's payoff is at least $u^*_A - \delta$ in every \hyperlink{NonCapricious}{Non-Capricious} equilibrium of any game with $T \geq T_\delta$ rounds, regardless of the initial default.

Next, we show that manipulability is necessary for approximate dictatorial power. Let $\CCP$ be a collective choice problem that is not \hyperlink{Definition:Manipulable}{Manipulable}. Then there exists some $x \in \mathcal{E} \backslash X^*_A$ and $\delta>0$ such that $u_A(x) < u^*_A - \delta$. Because $x \in \mathcal{E}$, \Cref{Lemma-NC-full}(a) implies that, given any $T \in \mathbb{N}$, there exists an initial default $x^0 \in X$ and \hyperlink{NonCapricious}{Non-Capricious} equilibrium $\sigma \in \Sigma^\text{NC}(x^0,T)$ such that the outcome is $g_{T}^{\sigma} (x^0) = x$; in fact, \Cref{Lemma-General1}(a) and \Cref{Lemma-NC-full2} together imply we can always pick the initial default to be $x^0 = x$. Thus, the agenda setter does not have approximate dictatorial power. \qed

\subsubsection{Proof of \Cref{Lemma-NC-full2} on p. \pageref{Lemma-NC-full2}}\label{Section-ProofLemmaNCfull2}

Throughout this section, we take $x^0 \in X$ and $T \in \mathbb{N}$ as given and consider the $T$-round game with initial default $x^0$. For any round $t\in \{1,\ldots,T\}$ and prevailing default $x^{t-1} \in X$, let $\mathcal{H}^t(x^{t-1})$ denote the set of round-$t$ histories consistent with this default. For any strategy profile $\sigma$ that satisfies \Cref{defi:NC}(a), recall that $g_T^\sigma(x^0) \in X$ is the induced outcome starting from the initial history; correspondingly, for each $t \in \{2, \dots , T\}$ and $x^{t-1}\in X$, let $g^\sigma_{T,x^0}(x^{t-1} \mid t) \in X$ denote the induced continuation outcome if $x^{t-1}$ is the prevailing default in round $t$. To simplify some statements, we also extend this notation to the final round by letting $g^\sigma_{T,x^0}(x^{T} \mid T+1) \equiv x^T$. Finally, let $G_{T,x^0}^\sigma (t) \equiv \bigcup_{x^{t-1} \in X} \left\{g^\sigma_{T,x^0}(x^{t-1} \mid t)\right\}$ denote the set of continuation outcomes arising across all round-$t$ subgames.

\medskip\noindent\textbf{Proof of Part (a).} Let a selection $\hat{\phi}(\cdot) \in \Phi^\text{or} (\cdot)$ be given. We construct a pure Markovian strategy profile $\sigma$ in the $T$-round game with initial default $x^0$ as follows:
\begin{itemize}[noitemsep]
    \item The agenda setter always proposes $\hat{\phi}(x)$ when the prevailing default is $x$.
    \item Each voter $i \in N$ votes to approve a proposal $y$ in round $t$ when the prevailing default is $x^{t-1}$ if and only if either
    \begin{itemize}[nolistsep]
        \item[(i)] $\hat{\phi}^{T-t}(y) \succ_i \hat{\phi}^{T-t}(x^{t-1})$, or
        \item[(ii)] $\hat{\phi}^{T-t}(y) \sim_i \hat{\phi}^{T-t}(x^{t-1})$ and $\hat{\phi}^{T-t}(y) = \hat{\phi}^{T-t+1}(x^{t-1})$.
    \end{itemize}
\end{itemize}
%
We claim that $\sigma$ is a non-capricious equilibrium. 

First, observe that $\sigma$ satisfies \Cref{defi:NC}(a) by construction; it induces the desired outcome $g^\sigma_T(x^0) = \hat{\phi}^T(x^0)$ and the continuation outcomes $g^\sigma_{T,x^0}(x^{t-1} \mid t) = \hat{\phi}^{T-t+1}(x^{t-1})$. Second, observe that $\sigma$ satisfies \Cref{defi:NC}(b) also by construction: at any round-$t$ history $h^t \in \mathcal{H}^t(x^{t-1})$, each voter $i \in N$ votes to approve a proposal $y$ if and only if either
\begin{itemize}[noitemsep]
    \item[(i$^*$)] voter $i$ strictly prefers $g^\sigma_{T,x^0}(y \mid t+1)$, the continuation outcome from approval of $y$, over $g^\sigma_{T,x^0}(x^{t-1} \mid t+1)$, the continuation outcome from rejection of $y$; or
    \item[(ii$^*$)] voter $i$ is indifferent between these continuation outcomes and $g^\sigma_{T,x^0}(y \mid t+1) = \hat{\phi}( g^\sigma_{T,x^0}(x^{t-1} \mid t+1) )$.
\end{itemize}
\Cref{defi:NC}(b) is satisfied because the tie-breaking rule in (ii$^*$) depends only on the continuation outcomes conditional on approval and rejection. 

Finally, we claim that $\sigma$ is an equilibrium, and hence satisfies \Cref{defi:NC}. Clearly, no voter has a strictly profitable deviation, so it suffices to consider the agenda setter's incentives. Let $x^{t-1} \in X$ and a round-$t$ history $h^t \in \mathcal{H}^t(x^{t-1})$ be given and let $\omega \equiv g_{T,x^0}^\sigma(x^{t-1}|t+1)$.  By construction, a proposal $y$ passes if and only if $g^\sigma_{T,x^0}(y \mid t+1) \in M^\text{s}(\omega) \cup \{\hat{\phi}(\omega)\}$. Thus, the agenda setter can induce all and only continuation outcomes $z \in M^\text{s}(\omega) \cup \{\hat{\phi}(\omega), \omega\}$, where $\omega$ is achieved by proposing any $y$ that does not pass. Because $\hat{\phi}(\omega) \in \Phi^\text{or}(\omega)$ implies that $\hat{\phi}(\omega)$ is optimal for the agenda setter within $M^\text{as}(\omega) \supseteq M^\text{s}(\omega) \cup \{\hat{\phi}(\omega), \omega\}$, it follows that any proposal $y$ for which $g^\sigma_{T,x^0}(y \mid t+1) = \hat{\phi}(\omega)$ is a best response for the agenda setter. Therefore, observing that $g^\sigma_{T,x^0}(\hat{\phi}(x^{t-1}) \mid t+1) = \hat{\phi}^{T-t+1}(x^{t-1}) = \hat{\phi}(\omega)$ completes the proof.

\medskip\noindent\textbf{Proof of Part (b).} Let a non-capricious equilibrium $\sigma \in \Sigma^\text{NC}(x^0,T)$ be given. We establish the existence of the desired collection $\{\hat{\phi}_t\}_{t=1}^T$ of selections $\hat{\phi}_t (\cdot) \in \Phi^\text{or}(\cdot)$ through a series of claims. The first claim records useful properties of continuation play and outcomes in the final round, $t=T$.

\begin{claim}\label{Obs-NC1}
There exists a selection $\hat{\phi}_T(\cdot) \in \Phi^{\text{or}}(\cdot)$ and an acceptance correspondence $M^\sigma : X \rightrightarrows X$ with the following properties:
\begin{enumerate}[label={\normalfont (\alph*)}]
    \item For every $x \in X$, $M^\text{s}(x) \cup\{\hat{\phi}_T(x), x \}\subseteq M^\sigma (x) \subseteq M^\text{w}(x)$.
    \item  For every $x^{T-1} \in X$ and round-$T$ history $h^T \in \mathcal{H}^T(x^{T-1})$, a proposal $y$ such that $y\neq x$ is accepted if and only if $y \in M^\sigma (x^{T-1})$.
    \item For every $x^{T-1} \in X$, $\hat{\phi}_T(x^{T-1}) \in \argmax_{z \in M^\sigma(x^{T-1})} u_A(z)$.
    \item For every $x^{T-1} \in X$,  $g_{T,x^0}^\sigma (x^{T-1} \mid T) = \hat{\phi}_T(x^{T-1})$.
\end{enumerate}
\end{claim}
\begin{proof}
In the final round $T$, for any proposal $y$ and prevailing default $x^{T-1}$, acceptance of the proposal leads to continuation outcome $y$ and rejection leads to $x^{T-1}$. Because $\sigma$ satisfies \Cref{defi:NC}(b), there exists a correspondence $M^\sigma : X \rightrightarrows X$ such that, for every default $x^{T-1} \in X$ and history $h^T \in \mathcal{H}^T(x^{T-1})$, a proposal $y$ is accepted if and only if $y \in M^\sigma(x^{T-1})$. This establishes part (b).

Let $x^{T-1} \in X$ and $h^T \in \mathcal{H}^T(x^{T-1})$ be given. The fact that the continuation of $\sigma$ in this subgame is an equilibrium thereof implies that $M^\sigma$ satisfies $M^\text{s}(\cdot) \subseteq M^\sigma(\cdot) \subseteq M^\text{w}(\cdot)$ and the continuation outcome, call it $y(h^T)$, satisfies $y(h^T) \in \argmax_{z \in M^\sigma(x^{T-1})} u_A(z)$. Moreover, because this continuation equilibrium is \hyperlink{NonCapricious}{Non-Capricious}, \Cref{Lemma-General1}(b) implies that $y(h^T) \in \Phi^\text{or}(x^{T-1}) \cap M^\sigma(x^{T-1})$. 

Finally, because $\sigma$ satisfies \cref{defi:NC}(a), there exists some $\hat{\phi}(x^{T-1}) \in \Phi^\text{or}(x^{T-1}) \cap M^\sigma(x^{T-1})$ such that $y(h^T) = \hat{\phi}(x^{T-1}) $ for all $h^T \in \mathcal{H}^T(x^{T-1})$ and $x^{T-1} \in X$. This establishes parts (c) and (d). To complete the proof of part (a) and hence the claim, simply note that we may include $x^{T-1} \in M^\sigma(x^{T-1})$ for all $x^{T-1} \in X$ without loss of generality, as both passage and rejection of a proposal $y =  x^{T-1}$ leads to continuation outcome $x^{T-1}$ at every history in $\mathcal{H}^T(x^{T-1})$, and part (b) of the claim only concerns proposals $y \neq x^{T-1}$.
\end{proof}

The next claim uses the \hyperlink{NonCapricious}{Non-Capricious} refinement to show that the majority acceptance correspondence from \Cref{Obs-NC1} also characterizes voter behavior in all rounds, and records a useful implication of this fact.

\begin{claim}\label{Obs-NC2}
For every $1 \leq t \leq T$, default $x^{t-1} \in X$, and round-$t$ history $h^t \in \mathcal{H}^t(x^{t-1})$, the following hold:
\begin{enumerate}[label={\normalfont (\alph*)}]
    \item A proposal $y$ such that $g^\sigma_{T,x^0}(y \mid t+1) \neq g^\sigma_{T,x^0}(x^{t-1} \mid t+1)$ is accepted at $h^t$ if and only if $g^\sigma_{T,x^0}(y \mid t+1) \in M^\sigma \left( g^\sigma_{T,x^0}(x^{t-1} \mid t+1) \right)$.
\item The continuation outcome at $h^t$ satisfies
\[
    g^\sigma_{T,x^0}(x^{t-1} \mid t) \in \argmax\left\{ u_A(z) : z \in M^\sigma \left( g^\sigma_{T,x^0}(x^{t-1} \mid t+1) \right) \bigcap G^\sigma_{T,x^0}(t+1)\right\}. 
    \]
\end{enumerate}
\end{claim}
\begin{proof}
Part (a) follows directly from \Cref{Obs-NC1}(b) and the fact that $\sigma$ satisfies \Cref{defi:NC}. Part (b) follows directly from part (a) and the fact that $\sigma$ satisfies \Cref{defi:NC}.
\end{proof}

The next claim records the elementary observation that any continuation outcome of a round-$t$ subgame must also be the continuation outcome of some round-$(t+1)$ subgame. 

\begin{claim}\label{Obs-NC3}
For every $1 \leq t \leq T-1$, we have $G^\sigma_{T,x^0}(t) \subseteq G^\sigma_{T,x^0}(t+1)$.
\end{claim}
\begin{proof}
Let the round $t \in \{1, \dots, T-1\}$, default $x^{t-1}\in X$, and history $h^t \in \mathcal{H}^t(x^{t-1})$ be given. By definition, the continuation outcome at $h^t$ is $g^\sigma_{T,x^0}(x^{t-1}\mid t)$. If $\sigma$ specifies that some proposal $a^\sigma(h^t) \in X$ is made and passed with positive probability at $h^t$, then by construction we have $g^\sigma_{T,x^0}(x^{t-1}\mid t) = g^\sigma_{T,x^0}(a^\sigma(h^t)\mid t+1)$. If $\sigma$ specifies that no proposals pass with positive probability at $h^t$, then by construction we have $g^\sigma_{T,x^0}(x^{t-1}\mid t) = g^\sigma_{T,x^0}(x^{t-1} \mid t+1)$. The claim now follows immediately from the definition of $G^\sigma_{T,x^0}(t)$ and $G^\sigma_{T,x^0}(t+1)$.
\end{proof}

The final claim uses \Cref{Obs-NC1,Obs-NC2,Obs-NC3} to characterize outcomes at every history. \Cref{Lemma-NC-full2}(b) is directly implied by this claim. 

\begin{claim}\label{Obs-NC4}
There exists a collection $\{\hat{\phi}_{t}\}_{t=1}^{T}$ of selections $\hat{\phi}_t(\cdot) \in \Phi^{\text{or}}(\cdot) \cap M^\sigma(\cdot)$ such that the following hold:
\begin{enumerate}[noitemsep,label={\normalfont (\alph*)}]
 %
    \item For every $2 \leq t \leq T$, the continuation outcomes satisfy the following: 
    \begin{equation}
   \text{For all $x^{t-1} \in X$, } \  g^\sigma_{T,x^0}(x^{t-1} \mid t) = \left[ \hat{\phi}_{t} \circ \hat{\phi}_{t+1} \circ \cdots \circ \hat{\phi}_T \right] (x^{t-1}).\label{Equation-NC-CO}
    \end{equation}
    Analogously, the equilibrium outcome of the game is $g^\sigma_{T}(x^{0}) = \left[ \hat{\phi}_{1} \circ \hat{\phi}_{2} \circ \cdots \circ \hat{\phi}_T \right] (x^{0})$.
    \item For every $x \in X$ and $1 \leq t \leq T$, we have $\hat{\phi}_t(x) \sim_A \hat{\phi}_T(x)$.
\end{enumerate}
\end{claim}

\begin{proof}
We prove (a) by backward induction. \Cref{Obs-NC1}(d) establishes the base ($t=T$) case of \eqref{Equation-NC-CO}. By letting $z \equiv \hat{\phi}(x^{T-1})$ and $z' \equiv \hat{\phi}(z)$ for any given $x^{T-1} \in X$, \Cref{Obs-NC1}(d) also establishes the base ($t=T$) case of the following property:
\begin{equation}
    \text{If $z \in G^\sigma_{T,x^0}(t)$, then  $\exists \ z' \in G^\sigma_{T,x^0}(t) \cap M^\sigma(z) \cap \Phi^{\text{or}}(z)$ such that $z'\prefto_A \hat{\phi}_T(z)$.} \label{Equation-NC-Feas}
    \end{equation}
For the inductive step, suppose for a given $\tau \in \{2,\dots, T-1\}$ that (i) there exist selections $\hat{\phi}_t(\cdot) \in \Phi^\text{or}(\cdot)$ satisfying \eqref{Equation-NC-CO} for $t = \tau+1$ and (ii) \eqref{Equation-NC-Feas} holds for $t = \tau +1$. 

We first assert that there exists a selection $\hat{\phi}_\tau(\cdot) \in \Phi^\text{or}(\cdot)$ satisfying \eqref{Equation-NC-CO} for $t=\tau$. Let $x^{\tau-1} \in X$ and $h^\tau \in \mathcal{H}^\tau (x^{\tau-1})$ be given. Let $z \equiv g^\sigma_{T,x^0}(x^{\tau-1}\mid \tau+1) \in G^\sigma_{T,x^0} (\tau+1)$ denote the continuation outcome if $x^{\tau-1}$ remains the default in the next round, $t = \tau+1$. By the inductive hypothesis that \eqref{Equation-NC-Feas} holds for $t = \tau +1$ and \Cref{Obs-NC2}(b), the continuation outcome at $h^\tau$, which is $g^\sigma_{T,x^0}(x^{\tau-1}\mid \tau)$, must satisfy $g^\sigma_{T,x^0}(x^{\tau-1}\mid \tau) = \hat{\phi}_\tau(z)$ for some $\hat{\phi}_\tau(z) \in \Phi^\text{or}(z)\cap M^\sigma(z)$ such that $\hat{\phi}_\tau(z) \prefto_A \hat{\phi}_T(z)$. Now, repeating this logic across all round-$\tau$ histories delivers, for all $x \in  G^\sigma_{T,x^0} (\tau+1)$, the existence of some $\hat{\phi}_\tau(x) \in \Phi^\text{or}(x)\cap M^\sigma(x)$ such that \eqref{Equation-NC-CO} holds for $t=\tau$ and $\hat{\phi}_\tau(x) \prefto_A \hat{\phi}_T(x)$. Since no policy in $X \backslash  G^\sigma_{T,x^0} (\tau+1)$ can be induced as a continuation outcome by any proposal at any round-$\tau$ history, we may arbitrarily assign $\hat{\phi}_\tau (x) \equiv \hat{\phi}_T(x) \in \Phi^{\text{or}}(x) \cap M^\sigma (x)$ for each $x \in X \backslash  G^\sigma_{T,x^0} (\tau+1)$. This results in the desired selection $\hat{\phi}_\tau(\cdot) \in \Phi^\text{or}(\cdot)$, completing the proof of the assertion.

Next, we assert that \eqref{Equation-NC-Feas} holds for $t = \tau$. Let $z \in G_{T,x^0}^\sigma (\tau)$ be given. \Cref{Obs-NC3} implies that $z \in G_{T,x^0}^\sigma (\tau+1)$, so that $z = g^\sigma_{T,x^0}(x\mid \tau+1)$ for some round-$(\tau+1)$ default $x \in X$. Let $z' \equiv g^\sigma_{T,x^0}(x \mid \tau)$ denote the continuation outcome if $x$ is the round-$\tau$ default. By the argument in the preceding paragraph, we have $z' = \hat{\phi}_\tau(z) \in M^\sigma(z) \cap \Phi^\text{or}(z)$ and thus $z' \prefto_A \hat{\phi}_T(z)$. As $z ' \in G_{T,x^0}^\sigma (\tau)$ by construction, the assertion is proved. 

This completes the inductive proof of part (a) for all rounds $t\in \{2,\ldots,T\}$. Repeating the first inductive step above once more establishes it for round $t=1$. 

To prove part (b), note that $\hat{\phi}_t(x) \prefto_A \hat{\phi}_T(x)$ for all $1\leq t \leq T$ and $x \in X$ by construction. Suppose, towards a contradiction, that there exists some $1\leq t\leq T$ and $x \in X$ such that $\hat{\phi}_t(x) \succ_A \hat{\phi}_T(x)$. Then consider any round-$T$ history $h^T(x)$ in which the default is $x^{T-1}=x$. Because $\hat{\phi}_t (x) \in M^\sigma (x)$ by construction, \Cref{Obs-NC1}(b) implies that the agenda setter has a strictly profitable deviation at $h^T$ by proposing $\hat{\phi}_t (x)$ instead of $\hat{\phi}_T(x)$, contradicting that $\sigma$ is an equilibrium.
\end{proof}

\subsubsection{Proof of \Cref{Lemma-NC-full} on p. \pageref{Lemma-NC-full}}\label{Section-ProofLemmaNCfull}

In this section, we use the same notation introduced at the beginning of \Cref{Section-ProofLemmaNCfull}. We prove each part of \Cref{Lemma-NC-full} in turn. 

\medskip\noindent\textbf{Proof of Part (a).} We first show that every unimprovable policy is a \hyperlink{NonCapricious}{Non-Capricious} equilibrium outcome. Let $x \in \mathcal{E}$ be given. \Cref{Lemma-General1}(a) implies that $x \in \Phi^\text{or}(x)$; hence, there exists a selection $\hat{\phi}(\cdot) \in \Phi^\text{or}(\cdot)$ such that $\hat{\phi}(x) = x$. \Cref{Lemma-NC-full2}(a) then implies that, for every $T \in \mathbb{N}$, there exists a \hyperlink{NonCapricious}{Non-Capricious} equilibrium $\sigma \in \Sigma^\text{NC}(x,T)$ with outcome $g^\sigma_T(x) = x$. Thus, $\mathcal{E} \subseteq \bigcup_{x^0 \in X} G_T(x^0)$ for every $T \in \mathbb{N}$. 

Next, we show that the equilibrium outcome sets are decreasing in the number of rounds. Let $T \in \mathbb{N}$, $x^0 \in X$, and $\sigma \in \Sigma^{\text{NC}}(x^0,T)$ be given. By \Cref{Obs-NC3} in \Cref{Section-ProofLemmaNCfull2}, we have $\{g^\sigma_T(x^0)\} = G^\sigma_{T,x^0}(1) \subseteq G^\sigma_{T,x^0}(2)$. As the continuation of $\sigma$ at any round-$2$ history (of the $T$-round game with initial default $x^0$) is a \hyperlink{NonCapricious}{Non-Capricious} in the corresponding $(T-1)$-round subgame, it follows that $ G^\sigma_{T,x^0}(2) \subseteq \bigcup_{y^0 \in X} G_{T-1}(y^0)$. It follows that
\[
\bigcup_{x^0 \in X} G_T(x^0) \subseteq \bigcup_{x^0 \in X} \bigcup_{\sigma \in \Sigma^{\text{NC}}(x^0,T)} G^\sigma_{T,x^0}(2) \subseteq \bigcup_{y^0 \in X} G_{T-1}(y^0),
\]
which completes the proof. 

\medskip\noindent\textbf{Proof of Part (b).} We begin by stating a useful variant of the \emph{uniform improvement lemma} (\Cref{lemma:uniform-improvement}) used in \Cref{appendix-theorem2-full} to prove \Cref{Theorem-epsilongrid}. For each $\delta >0$, let
\[
\Upsilon_\delta \equiv \left\{ x \in X : \min_{y \in \mathcal{E}}u_A(y) \geq u_A(x) + \delta  \right\}
\]
denote the set of policies that are $\delta$-dominated for the agenda setter by \emph{all} unimprovable policies $y \in \mathcal{E}$. Obviously, $\Upsilon_\delta \subseteq X \backslash \mathcal{E}$ for all $\delta>0$. As in \Cref{appendix-theorem2-full}, for each $x \in \policyspace$ and $\eta >0$, define
\begin{align*}
Q(x, \eta):= \big\{ y \in X \mid & u_A(y) \geq u_A(x) + \eta \text{ and}\\
&\text{$\exists$ majority $S \subseteq N$ such that } u_i(y) \geq u_i(x) + \eta \ \ \forall  i \in S \big\} \notag
\end{align*}
to be the set of policies that lead to a utility improvement of at least $\eta$ for some winning coalition. If $Q(x,\eta)\neq \emptyset$, then we say that $x$ is $\eta$-improvable.
\begin{lemma}\label{lemma:uniform-improvement-2}
For any collective choice problem $\CCP$ and every $\delta >0$, there exists $\eta_\delta >0$ such that $Q(x, \eta_\delta) \neq \emptyset$ for all $x \in \Upsilon_\delta$.
\end{lemma}
\begin{proof}
Let $\delta >0$ be given. Suppose that $\Upsilon_\delta \neq \emptyset$, for otherwise the lemma is vacuously true. First, observe that $\Upsilon_\delta$ is compact because $u_A$ is continuous and $X$ is compact. Second, observe that for each $x \in \Upsilon_\delta$ there exists some $\eta_x >0$ such that $Q(x,\eta_x)\neq \emptyset$; this follows from the definition of improvability and the inclusion $\Upsilon_\delta\subseteq X \backslash \mathcal{E}$. Given these observations, the remainder of the proof is identical to the proof of \Cref{lemma:uniform-improvement} in \Cref{appendix-theorem2-full}.\footnote{The only difference is that here we appeal to the definition of improvability and the inclusion $\Upsilon_\delta\subseteq X \backslash \mathcal{E}$---rather than the manipulability of $\CCP$---to establish the second observation above.} 
\end{proof}

Now, let $\delta>0$ be given and let $\eta_\delta>0$ be as described in \cref{lemma:uniform-improvement-2}. By the definitions of  $V^\text{as}_A(\cdot)$ and $\Phi^\text{or}(\cdot)$, we have 
\begin{equation}
    u_A(z) - u_A(y) \geq V^\text{as}_A(y) - u_A(y) \geq \eta_\delta \ \ \text{ for all $y \in \Upsilon_\delta$ and $z \in \Phi^\text{or}(y)$}. \label{Equation-UnifImpFinal}
\end{equation}

Let $\Delta \equiv u^*_A - \min_{x \in X} u_A(x)$. We prove \Cref{Lemma-NC-full}(b) by showing that
\begin{equation}
    T \geq T_\delta \equiv \left\lceil \Delta/ \eta_\delta \right\rceil \ \ \implies \ \ \bigcup_{x^0 \in X} G_T(x^0) \subseteq X \backslash \Upsilon_\delta. \label{Equation-T-Ups}
\end{equation}
Towards a contradiction, suppose that there exists a default $x^0 \in X$, number of rounds $T \geq T_\delta$, and \hyperlink{NonCapricious}{Non-Capricious} equilibrium $\sigma \in \Sigma^{\text{NC}}(x^0, T)$ such that $g^\sigma_T(x^0) \in \Upsilon_\delta$. By \Cref{Obs-NC4}(a), $g^\sigma_T(x^0) = \left[ \hat{\phi}_1 \circ \cdots \circ \hat{\phi}_T\right] (x^0)$ and $g^\sigma_{T,x^0}(x^0 \mid t) = \left[ \hat{\phi}_t \circ \cdots \circ \hat{\phi}_T\right] (x^0)$ for some selections $\hat{\phi}_t(\cdot) \in \Phi^\text{or}(\cdot)$. Let $g^\sigma_{T,x^0}(x^0 \mid 1) \equiv g^\sigma_T (x^0)$.  Note that $g^\sigma_{T,x^0}(x^0 \mid t) \in \Phi^\text{or} ( g^\sigma_{T,x^0}(x^0 \mid t+1))$ for all $1 \leq t \leq T$. Then it follows that
\begin{align*}
V^\text{as}_A \left( g^\sigma_T(x^0) \right) - u_A(x^0) &\geq 
u_A \left( g^\sigma_T(x^0) \right) - u_A(x^0) + \eta_\delta \\
&= \sum_{t=1}^{T} \left[ u_A \left( g^\sigma_{T,x^0}(x^0 \mid t) \right) - u_A \left( g^\sigma_{T,x^0}(x^0 \mid t+1) \right) \right] + \eta_\delta\\
&\geq T_\delta \cdot \eta_\delta + \eta_\delta\\
& \geq \Delta + \eta_\delta,
\end{align*}
where the first line is by \eqref{Equation-UnifImpFinal}, the second line is an identity, the third line is by the hypothesis that $T \geq T_\delta$ and another application of \eqref{Equation-UnifImpFinal} to each term in the sum (noting that $g^\sigma_T (x^0) \in \Upsilon_\delta$ implies that $g^\sigma_{T,x^0}(x^0 \mid t) \in \Upsilon_\delta$ for all $1 \leq t \leq T+1$), 
and the final line is by definition of $T_\delta$. However, given that $\eta_\delta>0$, this inequality contradicts the definition of $\Delta$. We conclude that \eqref{Equation-T-Ups} holds, as desired.

\subsection{Proof of \cref{Theorem-Spatial-Manip} on p. \pageref{Theorem-Spatial-Manip}}

Our argument proceeds in several steps. First, in $3$ dimensions, we establish a connection between non-coplanarity of utility gradients and improvability; as the reader will see, this argument applies for general payoff functions. We use this step to prove \cref{Theorem-Spatial-Manip}(a) for $d\geq 3$, by doing the appropriate reduction to $3$ dimensions and noting that non-coplanarity of ideal points with Euclidean preferences (\hyperlink{DefinitionNoCoplanarity}{Non-Coplanarity}) implies that of utility gradients. Finally, we prove \cref{Theorem-Spatial-Manip}(b) directly.

\medskip\noindent\textbf{Step 1: A General Improvability Lemma for $3$ Dimensions.} We establish here, for general payoff functions, that if utility gradients are non-coplanar at policy $x\neq x_A^*$, policy $x$ must be improvable.

\begin{lemma}\label{Lemma-Spatial-Gen}
Suppose $X = \mathbb{R}^3$ and each player $i \in N\cup\{A\}$ has a strictly quasi-concave and continuously differentiable utility $v_i : X \to \mathbb{R}$ with unique maximizer $x^*_i$. Then any policy $x \neq x^*_A$ is improvable if
\begin{equation}
\text{no $4$ vectors among $\{\nabla v_i (x)\}_{i=1,\dots, n,A}$ are coplanar}. \label{Equation-Grad-NCP}
\end{equation}
\end{lemma}

We prove \Cref{Lemma-Spatial-Gen} here. 
Let the profiles $\{v_i\}_{i=1,\dots,n,A}$ and $\{x^*_i\}_{i=1,\dots,n,A}$ be as described above, and consider a policy $x \in \mathbb{R}^3\backslash\{x_A^*\}$ that satisfies \eqref{Equation-Grad-NCP} . Denote the plane tangent to the agenda setter's indifference surface at $x$ by $S \equiv \{y \in \mathbb{R}^3 : (y -x) \cdot \nabla v_A(x) =0 \}$. The tangent space of $S$ is denoted by $\mathcal{T}(S) \equiv \{ z \in \mathbb{R}^3 : z \cdot \nabla v_A(x) =0 \}$ and the orthogonal complement of $S$ by $S^\perp \equiv \{y \in \mathbb{R}^3 : y \cdot z = 0 \ \forall z \in \mathcal{T}(S) \}$. For each voter $i \in N$, denote the orthogonal projection of $\nabla v_i (x)$ onto $S$ by $\nabla_S v_i(x)\equiv \nabla v_i(x) - \left(\frac{\nabla v_i (x) \cdot \nabla v_A (x)}{\|\nabla v_A (x)\|^2}\right)\nabla v_A (x)$. By construction, $\nabla_S v_i(x) \in \mathcal{T}(S)$ and $\nabla_S v_i(x) \cdot y = \nabla v_i(x) \cdot y$ for all $y \in \mathcal{T}(S)$.

We now establish \Cref{Lemma-Spatial-Gen} through a sequence of claims, which parallel the geometric sketch for $3$-dimensional Euclidean preferences given in \Cref{subsection:spatial}. The first claim records useful implications of \eqref{Equation-Grad-NCP} for the voters' projected gradients.

\begin{claim}\label{Claim-Spatial1}
The following hold:
\begin{enumerate}[label={\normalfont (\alph*)},nolistsep]
\item There exists some voter $i \in N$ for whom $ \nabla_S u_i(x) \neq \mathbf{0}$.
\item Given any voter $i \in N$ for whom $ \nabla_S v_i(x) \neq \mathbf{0}$, there exists at most one other voter $j \in N \backslash\{i\}$ with a collinear projected gradient, viz., such that $\nabla_S v_j(x) = \alpha \cdot \nabla_S v_i(x)$ for some $\alpha \in \mathbb{R}$.
\end{enumerate}
\end{claim}
For intuition, observe that if preferences are Euclidean, \Cref{Claim-Spatial1} reduces to the assertion from Step 1 and \Cref{fig:Euclid-1} in \Cref{subsection:spatial} that at most two voter constrained ideal points, $y^*_i$ and $y^*_j$, are collinear with $x$. For Euclidean preferences, the gradient at $x$ is $\nabla v^\text{E}_i (x) = x^*_i - x$ and the $S$-projected gradient at $x$ is $\nabla_S v^\text{E}_i (x) = y^*_i - x$. Thus, the collinearity of $\nabla_S v^\text{E}_i (x)$ and $\nabla_S v^\text{E}_j (x)$ is equivalent to the collinearity of the vectors $\{x,y^*_i, y^*_j\}$.

\begin{proof}[Proof of \Cref{Claim-Spatial1}]
For part (a), suppose that $ \nabla_S u_i(x) = \mathbf{0}$ for all $i \in N$. By definition, $ \nabla_S v_i(x) = \mathbf{0}$ if and only if $\nabla v_i(x) \perp S$. Thus, the supposition implies that all players' gradients are collinear, as the orthogonal complement $S^\perp$ has dimension $1$; \eqref{Equation-Grad-NCP} is then violated. By contraposition, \eqref{Equation-Grad-NCP} implies that (a) holds. 

For part (b), consider a voter $i \in N$ for whom $ \nabla_S u_i(x) \neq \mathbf{0}$. Suppose that there exist two distinct voters $j,k \in N \backslash\{i\}$ such that the set $\{\nabla_S v_i(x) , \nabla_S v_j(x) , \nabla_S v_k(x) \}$ is collinear. Observe that this set of vectors is also trivially collinear with $\nabla_S v_A(x) = \mathbf{0}$. Hence, the set $\text{span}\left( \{\nabla_S v_i(x) , \nabla_S v_j(x) , \nabla_S v_k(x) ,  \nabla_S v_A(x) \}\right)$ has dimension $1$. Note that $\nabla v_\nu(x) = \nabla_S v_\nu(x) +\nabla_{S^\perp} v_\nu(x) $ for all $\nu \in \{i,j,k,A\}$ by definition of orthogonal projection. Because each $\nabla_{S^\perp} v_\nu(x) \in S^\perp$ and $S^\perp$ has dimension $1$ by construction, it follows that the set $\text{span}\left( \{\nabla v_i(x) , \nabla v_j(x) , \nabla v_k(x) ,  \nabla v_A(x) \}\right)$ has dimension $2$, implying that \eqref{Equation-Grad-NCP} is violated.
\end{proof}

The next claim uses \Cref{Claim-Spatial1} to establish the existence of an alternative policy in $S$ that a majority of voters strictly prefer to $x$.

\begin{claim}\label{Claim-Spatial3}
There exists some $y \in S$ such that $y\succ_M x$. 
\end{claim}
The argument mirrors that from Step 1 and the left-hand panel of \Cref{fig:Euclid-2} in \Cref{subsection:spatial}: we (i) fix some voter $i \in N$ whose projected gradient $\nabla_S v_i(x)$ is nonzero and therefore defines a line in $S$ that contains $x$, (ii) partition the other voters into sets according to whether their projected gradients point ``above'' or ``below'' that line, and (iii) construct a new policy $y \in S$ that strictly benefits voter $i$ and all voters on one side of the line. For Euclidean preferences, $\nabla_S v^\text{E}_j (x) = y^*_j - x$ for all voters $j$, and so part (ii) here is equivalent to partitioning voters based on their constrained ideal points $y^*_j$ lying ``above'' or ``below'' the line.
\begin{proof}[Proof of \Cref{Claim-Spatial3}]
Consider a voter $i$ for whom $\nabla_S v_i(x) \neq \mathbf{0}$; such a voter exists by \Cref{Claim-Spatial1}(a). We denote the set of other voters whose $S$-projected gradients at policy $x$ are collinear to $i$'s by $C_i \equiv \{j \in N \backslash \{i\} :   \nabla_S v_j(x) = \alpha \cdot  \nabla_S v_i(x) \ \exists \alpha \in \mathbb{R} \}$. We then define $N' \equiv N \backslash C_i$. Observe that $i \in N'$ by construction and that $|N'| \geq n-1$ by \Cref{Claim-Spatial1}(b).

Now, let $\omega \in \mathcal{T}(S)$ satisfy $\omega \cdot \nabla_S v_i(x) = 0$ be given. Define the following sets of voters:
\begin{align*}
    N'_+ \equiv \{j \in N' : \nabla_S v_j(x) \cdot \omega >0\} \ \  \text{ and } \ \ N'_- &\equiv \{j \in N' : \nabla_S v_j(x) \cdot \omega <0\}.
\end{align*}
By construction, $N'_+ \cap \{i\} = N'_- \cap \{i\} = \emptyset$ and $N' = N'_+ \cup N'_- \cup \{i\}$, viz., $\left\{ N'_+, N'_-, \{i\}\right\}$ forms a partition of $N'$. It follows that $|N'_+| + |N'_-| \geq n-2$, which in turn implies that $\max\{|N'_+|, |N'_-|\} \geq \frac{n-1}{2}$ because $n-2$ is an odd number. We suppose, without loss of generality, that $|N'_+| \geq \frac{n-1}{2}$. Therefore, we have
\begin{equation}
    |N'_+ \cup\{i\}| \geq \frac{n+1}{2}. \label{Equation-Nplus}
\end{equation} 
We assert that there exists some $\rho \in \mathbb{R}^3$ such that
\begin{equation}
    \rho \in \mathcal{T}(S) \  \text{ and } \  \nabla_S v_j (x) \cdot \rho >0 \ \text{ for all $j \in N'_+ \cup \{i\}$}. \label{Equation-rho-perturb}
\end{equation}
To this end, define the sequence $\{\rho^k\} \subset \mathbb{R}^3$ by $\rho^k \equiv \frac{1}{k} \nabla_S v_i(x) + \frac{k-1}{k} \omega$. It is clear that $\rho^k \in \mathcal{T}(S)$ for all $k \in \mathbb{N}$, as $\mathcal{T}(S)$ is a convex set containing both $ \nabla_S v_i(x)$ and $\omega$ by construction. It is also clear that $\nabla_S v_i (x) \cdot \rho^k = \|\nabla_S v_i (x)\|^2 / 2 >0$ for all $k \in \mathbb{N}$ by construction. So, let $j \in N'_+$ be given. As $ \nabla_S v_j(x) \cdot \omega >0$ by construction, there exists some $K_j \in \mathbb{N}$ such that $\nabla_S v_j (x) \cdot \rho^k>0$ for all $k \geq K_j$. Defining $K \equiv \max_{j \in N'_+} K_j$ and letting $\rho \equiv \rho^k$ for any $k \geq K$ then establishes \eqref{Equation-rho-perturb}, as desired.

We now use \eqref{Equation-Nplus} and \eqref{Equation-rho-perturb} to prove the claim. Let $\rho \in \mathbb{R}^3$ satisfy \eqref{Equation-rho-perturb}. For each voter $j \in N'_+ \cup \{i\}$, we have that
\begin{align*}
    v_j (x + \epsilon \rho) &= v_j (x) + \epsilon \nabla v_j(x) \cdot \rho + \smallO(\epsilon) = v_j (x) + \epsilon \nabla_S v_j(x) \cdot \rho + \smallO(\epsilon),
\end{align*}
where the first equality is by Taylor's Theorem and the second holds because, by the definition of the $S$-projected gradient $\nabla_S v_j(x)$, we have $\nabla v_j(x) \cdot \rho' = \nabla_S v_j(x) \cdot \rho'$ for all $\rho' \in \mathcal{T}(S)$. \eqref{Equation-rho-perturb} then implies that there exists some $\epsilon>0$ such that $v_j(x + \epsilon \rho) > v_j (x)$ for all $j \in N'_+ \cup \{i\}$. Letting $y \equiv x + \epsilon \rho$, \eqref{Equation-Nplus} then implies that $y \succ_M x$. As $\rho \in \mathcal{T}(S)$, it follows that $y \in S$.
\end{proof}

The final claim establishes that any policy $y \in S$ for which $y \succ_M x$ can be perturbed to some $z \notin S$ such that both $z \succ_M x$ and $z\succ_A x$; this claim formalizes the argument sketched in Step 2 and the right-hand panel of \Cref{fig:Euclid-2}. 

\begin{claim}\label{Claim-Spatial4}
For any $y \in S$ such that $y \succ_M x$, there exists $z \notin S$ such that $z \succ_M x$ and $z \succ_A x$.
\end{claim}
\begin{proof}
As $y \succ_M x$ and voters' preferences are continuous, there exists an $\epsilon >0$ such that the policy $\zeta \equiv y + \epsilon \nabla v_A(x)$ satisfies $\zeta \succ_M x$. For each $\beta \in (0,1)$, define $z(\beta) \equiv \beta \zeta + (1-\beta) x$. The strict convexity of voters' preferences then implies that $z(\beta) \succ_M x$ for all $\beta \in (0,1)$. 

We assert that there exists some $\overline{\beta} \in (0,1)$ such that $z(\beta) \succ_A x$ for all $\beta \in (0, \overline{\beta})$. Let $\rho \equiv \zeta - x$. As $v_A$ is continuously differentiable, its directional derivative at policy $x$ in direction $\rho$ is given by $\nabla v_A(x) \cdot \rho$. We have the following:
\begin{align*}
    \nabla v_A(x) \cdot \rho &= \nabla v_A(x) \cdot \left( y + \epsilon \nabla v_A(x) -x\right)  
    = \epsilon \|\nabla v_A(x)\|^2 >0 
\end{align*}
where the first equality is an identity, and the second follows from rearranging terms and noting that $\nabla v_A(x) \neq \mathbf{0}$ and $\nabla v_A(x) \perp \left( y-x\right)$. Thus, we have $\nabla v_A(x) \cdot \rho>0$. Now, because $z(\beta) - x = \beta \rho $ by construction, Taylor's Theorem implies that $v_A\left( z(\beta)\right) = v_A(x) + \beta \nabla v_A(x) \cdot \rho + \smallO(\beta)$. It follows that $v_A(z(\beta)) > v_A(x)$ for all sufficiently small $\beta >0$, as desired. 

Now let $z \equiv z(\beta)$ for any $\beta \in (0,\overline\beta )$. It follows that  $z \succ_A x$ and $z\succ_M x$ by construction and $z \notin S$ because $\nabla v_A(x) \neq \mathbf{0}$ is normal to $S$.
\end{proof}
 \Cref{Claim-Spatial3,Claim-Spatial4} together complete the proof of \Cref{Lemma-Spatial-Gen}. \qedsymbol

\paragraph{Step 2: Proof of \cref{Theorem-Spatial-Manip}(a).}
We now consider the case of Euclidean preferences: $d\geq 3$, $X = \mathbb{R}^d$, player $i$ has utility function $u_i (x) = - \frac{1}{2}\|x-x^*_i\|^2$ for each $i \in N \cup\{A\}$. Suppose that the ideal point profile $(x^*_i)_{i=1,\dots, n,A} \in\mathbb{R}^{d(n+1)}$ satisfies \hyperlink{DefinitionNoCoplanarity}{Non-Coplanarity}. Let an arbitrary $x \neq x^*_A$ be given; we show below that $x$ is improvable. 

If $d=3$, the result follows immediately from \Cref{Lemma-Spatial-Gen} by observing that $\nabla u_i (x) = x^*_i - x$, so that \hyperlink{DefinitionNoCoplanarity}{Non-Coplanarity} directly implies \eqref{Equation-Grad-NCP}. So suppose that $d>3$. In this case, we may indirectly apply \Cref{Lemma-Spatial-Gen} by restricting attention to a suitable $3$-dimensional subspace of $\mathbb{R}^d$. Let $a,b,c \in \{1,\dots, d\}$ denote $3$ distinct policy dimensions for which the projections $[x]_{abc}$ and $[x^*_A]_{abc}$ satisfy $[x]_{abc} \neq [x^*_A]_{abc}$. Let $[x]_{-abc} \in \mathbb{R}^{d-3}$ denote the $(d-3)$-dimensional projection of $x$ corresponding to deletion of the indices $a,b,c$ (so that $x$ is given by the concatenation of  $[x]_{abc}$ and $[x]_{-abc}$). Define $X \left( [x]_{-abc}\right) \equiv \{y \in \mathbb{R}^d : [y]_{-abc} = [x]_{-abc}\}$ to be the set of policies $y \in \mathbb{R}^d$ that differ from $x \in \mathbb{R}^d$ only in dimensions $a,b,c$. Observe that $X \left( [x]_{-abc}\right)$ is a $3$-dimensional affine subspace of $\mathbb{R}^d$ by construction; with a slight abuse of notation, we identify it with $\mathbb{R}^3$ and identify a generic element $y$ with its projection $[y]_{abc} \in \mathbb{R}^3$. Finally, for each player $i \in N \cup\{A\}$, we define the utility function $v_i : \mathbb{R}^3 \to \mathbb{R}$ by $v_i(\cdot) \equiv u_i(\cdot,[\hat{x}]_{-abc})$, viz., $v_i$ is the restriction of $u_i$ to $X \left( [x]_{-abc}\right)$. 

We now apply \Cref{Lemma-Spatial-Gen} to the utility profile $(v_i)_{i=1,\dots,n,A}$, which represents $3$-dimensional Euclidean preferences with the ideal point profile $([x^*_i]_{abc})_{i=1,\dots,n,A}\in\mathbb{R}^{3(n+1)}$. Observe that $\nabla v_i([x]_{abc}) = [x^*_i]_{abc} - [x]_{abc} \in \mathbb{R}^3$, so that \hyperlink{DefinitionNoCoplanarity}{Non-Coplanarity} of the $d$-dimensional ideal points implies that these $3$-dimensional gradients satisfy \eqref{Equation-Grad-NCP}. We thus conclude from \Cref{Lemma-Spatial-Gen} that there exists some $[y]_{abc} \in \mathbb{R}^3$ such that 
\begin{equation}
    v_A\left([y]_{abc}  \right) > v_A \left( [x]_{abc} \right) \  \text{ and } \  |\left\{i \in N : v_i\left([y]_{abc}  \right) > v_i \left( [x]_{abc} \right) \right\}| \geq \frac{n+1}{2}. \label{Equation-Improve-3}
\end{equation}

To conclude the proof, we let $y \in X \left( [x]_{-abc}\right) \subseteq \mathbb{R}^d$ denote the concatenation of $[y]_{abc}$ and $[x]_{-abc}$, viz., $y \equiv \left( [y]_{abc}, [x]_{-abc}\right)$. By definition of the $v_i$ functions, \eqref{Equation-Improve-3} implies that
\[
u_A\left(y \right) > u_A \left(x \right) \  \text{ and } \  |\left\{i \in N : u_i\left(y \right) > u_i \left( y \right) \right\}| \geq \frac{n+1}{2},
\]
which is equivalent to $y \succ_A x$ and $y \succ_M x$. It follows that $x$ is improvable, as desired. \qedsymbol

\paragraph{Step 3: Proof of \cref{Theorem-Spatial-Manip}(b).} Let $d \geq 3$ be given. For any $3$ distinct policy dimensions $a,b,c \in \{1,\dots, d\}$ and any $4$ distinct players $i,j,k,\ell \in \{1,\dots,n, A\}$, we define the set $C^{(ijk\ell)}_{[abc]} \subseteq \mathbb{R}^{d(n+1)}$ by 
\[
C^{(ijk\ell)}_{[abc]} \equiv \left\{ \left(x^*_\nu \right)_{\nu=1,\dots, n,A} \in \mathbb{R}^{d(n+1)} : [x^*_i]_{abc}, [x^*_j]_{abc}, [x^*_k]_{abc}, \text{ and } [x^*_\ell]_{abc} \text{ are coplanar in $\mathbb{R}^3$} \right\}.
\]
In words, $C^{(ijk\ell)}_{[abc]}$ collects all profiles of ideal points for which \hyperlink{DefinitionNoCoplanarity}{Non-Coplanarity} is violated (at least) for players $i,j,k,\ell$ in the subspace spanned by dimensions $a,b,c$. Observe that, by \Cref{Definition-NoCoplanarity}, taking the union over all such $a,b,c$ and $i,j,k,\ell$ yields exactly the set of ideal point profiles that violate \hyperlink{DefinitionNoCoplanarity}{Non-Coplanarity}. That is, the following holds:
\begin{align}
    C &\equiv \left\{ (x^*_i)_{i=1,\dots,n,A} \in \mathbb{R}^{d(n+1)} : (x^*_i)_{i=1,\dots,n,A} \text{ violates \hyperlink{DefinitionNoCoplanarity}{Non-Coplanarity}} \right\} \label{Equation-Coplanar-Set} \\
    & = \bigcup_{a,b,c \in \{1,\dots,d\}} \bigcup_{i,j,k,\ell \in N\cup\{A\}} C^{(ijk\ell)}_{[abc]}. \notag
\end{align}

We show that $C$ is closed and has zero Lebesgue measure by showing that each $C^{(ijk\ell)}_{[abc]}$ also has these properties (because the union in \eqref{Equation-Coplanar-Set} is finite). To this end, we first claim that 
\begin{equation}
    Z \equiv \left\{ (x_i)_{i=1}^4 \in \mathbb{R}^{3 \times 4} : x_i \in \mathbb{R}^3 \ \forall i =1,2,3,4 \text{ and } x_1,x_2,x_3,x_4 \text{ are  coplanar in $\mathbb{R}^3$} \right\} \label{Equation-Zset}
\end{equation}
is closed and has zero Lebesgue measure. To see why, define $f : \mathbb{R}^{3\times4} \to \mathbb{R}$ by $f(x_1, x_2, x_3, x_4) \equiv \left[ (x_2-x_1) \times (x_3-x_1) \right] \cdot (x_4 - x_1)$,
where $y \times z \in \mathbb{R}^3$ denotes the cross product between vectors $y,z \in \mathbb{R}^3$. By construction,  $\{x_1,x_2,x_3,x_4\} $ are coplanar if and only if $f(x_1, x_2, x_3, x_4) =0 $; therefore, $Z = \left\{ (x_i)_{i=1}^4 \in \mathbb{R}^{3 \times 4} : f(x_1, x_2, x_3, x_4) = 0\right\}$. Now observe, also by construction, that $f$ is a non-constant polynomial function. Therefore, $Z$ is closed and has zero Lebesgue measure, being the set of zeros of a non-constant polynomial function.

We use this claim to establish that $C^{(ijk\ell)}_{[abc]}$ is closed and has zero Lebesgue measure. If $d=3$, this is an immediate consequence of the above. So suppose that $d>3$. Observe that whether a profile $(x^*_i)_{i=1,\dots,n,A} \in \mathbb{R}^{d(n+1)}$ is an element of $C^{(ijk\ell)}_{[abc]}$ is determined exclusively by the collection of projections $\{[x^*_i]_{abc}, [x^*_j]_{abc}, [x^*_k]_{abc}, [x^*_\ell]_{abc}\}$. Hence, 
\begin{align*}
C^{(ijk\ell)}_{[abc]} &=  K \times \mathbb{R}^{d(n+1)-12}, \text{ where $K \subseteq \mathbb{R}^{3\times 4}$ is defined by} \\
K &\equiv\left\{ \left([x^*_\nu]_{abc} \right)_{\nu=i,\ell,j,k} \in \mathbb{R}^{3\times 4} : [x^*_i]_{abc}, [x^*_j]_{abc}, [x^*_k]_{abc}, [x^*_\ell]_{abc} \text{ are coplanar in $\mathbb{R}^3$} \right\}.
\end{align*}
Observe that $K$ is equivalent (modulo relabeling of indices) to $Z$ in \eqref{Equation-Zset}, and therefore is closed and has zero Lebesgue measure in $\mathbb{R}^{12}$. Hence, $C^{(ijk\ell)}_{[abc]}$ is also closed and has zero Lebesgue measure in $\mathbb{R}^{d(n+1)}$.

The above establishes that $C$ is closed and has zero Lebesgue measure. Therefore, its complementary set $NC \equiv \mathbb{R}^{d(n+1)} \backslash C$ is open-dense and has full Lebesgue measure (as any open full-measure set is dense).

\subsection{Failure of Manipulability in Two-Dimensional Spatial Politics}\label{Section-Supplementary2D}
Using a three-voter example, we illustrate the assertion from \Cref{subsection:spatial} that, when there are $d=2$ policy dimensions, manipulability fails whenever the agenda setter's ideal point lies outside the convex hull of voter ideal points; this analysis straightforwardly extends to a general (odd) number of voters, provided that no $3$ of their ideal points are collinear (which is generically satisfied). 

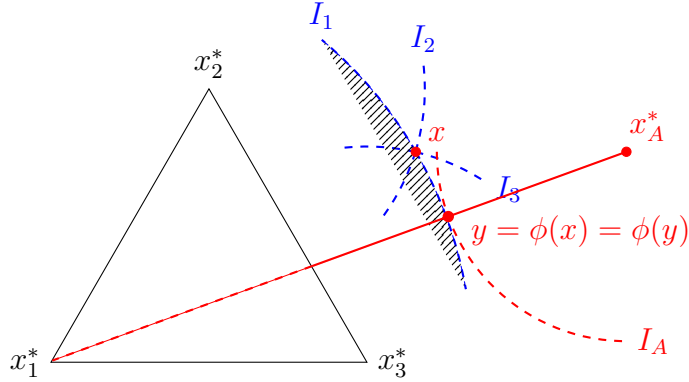
\begin{figure}[h!]
\centering
\begin{tikzpicture}[scale=.7]
\draw (0,0)coordinate (1) -- (0:6) coordinate (3) -- ++ (+120:6) coordinate (2) --cycle;
\draw (1)node[left]{$x_1^*$} (2)node[above]{$x_2^*$} (3)node[right]{$x_3^*$} ;



\draw[blue, thick,dashed] (30:8) arc
[
start angle=30,
end angle=50,
x radius=8,
y radius=8
]

(30:8) arc
[
start angle=30,
end angle=10,
x radius=8,
y radius=8
]
; 

\draw[blue] (50:8) node[above]{$I_1$};

\draw[blue, thick,dashed] (2)++(-17:4.1) arc
[
start angle=-17,
end angle=7,
x radius=4.1,
y radius=4.1
]

(2)++(-17:4.1) arc
[
start angle=-17,
end angle=-37,
x radius=4.1,
y radius=4.1
]
;

\draw[blue] (2)++(7:4.1) node[above]{$I_2$};

\draw[blue, thick,dashed] (3)++(77:4.1) arc
[
start angle=77,
end angle=57,
x radius=4.1,
y radius=4.1
]

(3)++(77:4.1) arc
[
start angle=77,
end angle=97,
x radius=4.1,
y radius=4.1
]
;

\draw[blue] (3)++(57:4.1) node[right,yshift=-3]{$I_3$};


\draw[blue,pattern=north east lines,dashed] (10:8) arc
[
start angle=10,
end angle=50,
x radius=8,
y radius=8
]
;


\filldraw[red] (30:8) circle (2.5pt)node[above,xshift=8]{$x$};

\filldraw[red] (30:8)++(0:4) circle (2.5pt)node[above,xshift=8]{$x^*_A$};

\draw[red,dashed,thick](30:8)++(0:4)++(180:3.6) arc
[
start angle=180,
end angle=270,
x radius=3.6,
y radius=3.6
]
;

\draw[red] (30:8)++(0:4)++(270:3.6) node[right]{$I_A$};

\filldraw[red,thick] (30:8)++(0:4)--++(200:3.6)circle(2.5pt)node[below,xshift=50,yshift=5]{$y=\phi(x)=\phi(y)$} --++(200:2.75);

\filldraw[red,thick,dashed] (30:8)++(0:4)++(200:3.6)++(200:2.75) -- ++(200:5.25);

\end{tikzpicture}
\caption{A failure of manipulability in the two-dimensional spatial model.}
      \label{fig:two-dim-spatial}
\end{figure}
Consider the situation depicted in \Cref{fig:two-dim-spatial} where $(x_i^*)_{i=1,2,3,A}$ depicts the profile of ideal points. We first observe that all policies on the line segment between $x_1^*$ and $x_A^*$ and outside the interior of the convex hull of voter ideal points---this is the solid red line---are unimprovable. To see why, note that for any such policy---such as policy $y$ in the figure---a majority of voters favor another policy, say $z$, to $y$ only if voter $1$ also favors $z$ to $y$.\footnote{To put it differently, voters $2$ and $3$ cannot favor $z$ to $y$ without voter $1$ also doing so.} Because voter $1$'s indifference curve passing through $y$ is tangent to the agenda setter's indifference curve passing through $y$, there is no policy that the agenda setter and voter $1$ prefer to $y$. Thus, $y$ is unimprovable. 

Necessarily, this eliminates any prospect for a dictatorial power result starting from \emph{any} default option: beginning with a default option of $y$ implies that it is the unique \hyperlink{NonCapricious}{Non-Capricious} equilibrium outcome, since the agenda setter's unique favorite improvement at this policy is $y$ itself.\footnote{Formally, the agenda setter's one-round improvement correspondence $\Phi^\text{or}$ (as defined in \Cref{Equation-PhiOR} on p. \pageref{Equation-PhiOR}) satisfies $\Phi^\text{or}(y) =\{y\}$ and the above assertion follows from \Cref{Lemma-NC-full2}(b) on p. \pageref{Lemma-NC-full2}. With a slight abuse of notation, in \Cref{fig:two-dim-spatial} we let $\phi(\cdot)$ denote the unique element of $\Phi^\text{or}(\cdot)$ at points where this correspondence is singleton-valued, and refer to this policy as the agenda setter's unique favorite improvement.}  Interestingly, it also prevents the agenda setter from fully exploiting real-time agenda control even from improvable default options (off this line segment). For example, suppose $x$ is the initial default option. The agenda setter's unique favorite improvement from $x$ is the unimprovable policy $y$, which implies that regardless of the number of rounds, the unique \hyperlink{NonCapricious}{Non-Capricious} equilibrium outcome is $y$. Of course, this logic does not merely apply to the policies $x$ and $y$, but is more general: there exists an open set of policies such that if the initial default option belongs to this open set, the unique \hyperlink{NonCapricious}{Non-Capricious} equilibrium outcome is bounded away from the agenda setter's ideal point. To put it differently, even if the set of unimprovable policies is measure-$0$ in $\Re^2$, strategic forces may compel negotiations to reach the set of unimprovable policies, contravening a dictatorial power result.

\subsection{Proof of \Cref{Theorem-ThreeRounds} on p. \pageref{Theorem-ThreeRounds}}

\paragraph{Step 1: Equilibrium Outcomes for General Voting Rules.} We first introduce notation that extends that from \cref{appendix:thm3-proof} to general voting rules. Given any voting rule $\mathcal{D}$ and $x \in X$, denote the \emph{weak} $\mathcal{D}$-acceptance set $M^\text{w}_\mathcal{D}(x) \equiv \{y \in X : y\prefto_i x \ \forall i \in D, \ \exists D \in \mathcal{D}\}$, the \emph{strict} $\mathcal{D}$-acceptance set $M^\text{s}_\mathcal{D}(x) \equiv \{y \in X : y \succ_i x \ \forall i \in D, \ \exists D \in \mathcal{D}\}$, and the \emph{almost-strict} $\mathcal{D}$-acceptance set  $M^\text{as}_\mathcal{D}(x) \equiv \text{cl} \left[ M^\text{s}_\mathcal{D}(x)\right] \cup\{x\}$. Define the agenda setter's \emph{favorite almost-strict $\mathcal{D}$-improvement} value function $V^\text{as}_A(\cdot \mid \mathcal{D}) : X \to \mathbb{R}$ by $V_A^{\text{as}}(x\mid\mathcal{D}) \equiv \max_{y \in M^{\text{as}}(x)} u_A(y)$, and her \emph{one-round $\mathcal{D}$-improvement} correspondence $\Phi^\text{or}_\mathcal{D} : X \rightrightarrows X$ by
\begin{equation}
\Phi^\text{or}_\mathcal{D}(x) \equiv \left\{ y \in X : y \in M^\text{w}_\mathcal{D}(x) \text{ and } u_A(y) \geq V^{\text{as}}_A(x\mid\mathcal{D}) \right\}. 
\end{equation}
We denote the set of \hyperlink{NonCapricious}{Non-Capricious} equilibria in the $T$-round game with initial default $x^0$ and voting rule $\mathcal{D}$ by $\Sigma_\mathcal{D}(x^0,T)$. Given any $\sigma \in \Sigma_\mathcal{D}(x^0,T)$,  we denote the equilibrium outcome by $g^\sigma_{T}(x^0\mid \mathcal{D})$. 

The following generalizes \Cref{Lemma-NC-full2}(b) in  \cref{appendix:thm3-proof} to arbitrary voting rules:

\begin{lemma}\label{Lemma-NC-D}
For any collective choice problem $\CCP$ and voting rule $\mathcal{D}$, the following holds:
\begin{quote}
    For any $x^0 \in X$, $T \in \mathbb{N}$, and \hyperlink{NonCapricious}{Non-Capricious} equilibrium $\sigma \in \Sigma_\mathcal{D}(x^0,T)$, there exists a collection $\{\hat{\phi}_t (\cdot) \}_{\tau=1}^T$ of selections $\hat{\phi}_t (\cdot) \in \Phi^{\text{or}}_\mathcal{D}(\cdot)$ such that the equilibrium outcome is given by 
    \[
    g_T^\sigma(x^0\mid\mathcal{D}) = \left[\hat{\phi}_1 \circ \hat{\phi}_{2} \circ \cdots \circ \hat{\phi}_T \right](x^0)
    \]
    and, for every $x \in X$, we have $\hat{\phi}_t (x) \sim_A \hat{\phi}_{T} (x)$ for all $1 \leq t \leq T$.
\end{quote}
\end{lemma}

The proof of \Cref{Lemma-NC-D} is identical to that of \Cref{Lemma-NC-full2}(b) modulo the notational adaptation described above, and hence omitted. 

\medskip\noindent\textbf{Step 2: Properties of Distribution Problems.} We now characterize the agenda setter's favorite policies and one-round $\mathcal{D}$-improvement operator in \hyperlink{DistributionProblem}{Distribution Problem}s. Throughout our analysis in this Step, we restrict attention to \hyperlink{DistributionProblem}{Distribution Problem}s, assume that \hyperlink{TI}{Thin Individual Indifference} holds, and consider a veto-proof voting rule $\mathcal D$. We denote the set of weakly Pareto efficient policies by $P\equiv\{x\in X:\nexists y\text{ such that }\forall i\in N\cup\{A\},\, y\succ_i x \}$. For any policy $x $, we define its \emph{support} by $\text{supp}(x) \equiv \{ i \in N : u_i(x) > \underline{u}_i\}$, viz., the set of voters for whom $x$ is not a least-preferred policy. The following claim demonstrates that the agenda setter's favorite policies are precisely those that are weakly Pareto efficient and leave all voters with minimal utility:

\begin{claim}\label{Lemma-DP-Xstar}
$X^*_A = \{x \in P :\text{supp}(x) = \emptyset \}$.
\end{claim}
\begin{proof}
For any $x \notin X^*_A$, \hyperlink{Bullet:Scarcity}{Scarcity} implies that  $x \notin P$ or $\text{supp}(x) \neq \emptyset$. By contraposition, it follows that $ \{x \in P :\text{supp}(x) = \emptyset \}\subseteq X_A^*$. For the opposite inclusion, consider a policy $y \notin \{x \in P :\text{supp}(x) = \emptyset \}$; we establish that $y\notin X_A^*$. If $y \notin P$, then by definition there exists a strongly Pareto dominating $z \in X$, and therefore, $y \notin X^*_A$. If $\text{supp}(y)\neq \emptyset$, then by definition there exists some voter $i \in N$ such that $u_i(y) > \underline{u}_i$. \hyperlink{Bullet:Transferability}{Transferability} then implies that there exists some $z \in X$ such that $z \succ_j x$ for all players $j \neq i$, including $j = A$; hence, $y \notin X^*_A$. By contraposition, it follows that $X^*_A \subseteq \{x \in P :\text{supp}(x) = \emptyset \}$.
\end{proof}

The next claim characterizes $\Phi^\text{or}_\mathcal{D}$. Given any policies $x,y\in X$, we let $L(y \mid x) \equiv \{i \in N: y \prec_i x \}$ denote the set of voters who are \emph{losers} if the implemented policy changes from $x$ to $y$. We say that voter $i \in N$ is \emph{minimized} at $x\in X$ if $i \notin \text{supp}(x)$. 
\begin{claim}\label{Lemma-PhiOR-DP}
For every $x \in X$ and $y \in \Phi^\text{or}(x)$, the following hold:
\begin{enumerate}[label={\normalfont (\alph*)},noitemsep]
\item $y$ is weakly Pareto efficient: $y \in P$.
\item Losers are minimized: $L(y\mid x) = \text{supp}(x) \backslash \text{supp}(y)$.
\item Minimized voters remain minimized: $\text{supp}(y) \subseteq \text{supp}(x)$.
\item Minimal winning coalition: $\nexists D \in \mathcal{D}$ and $i \in \text{supp}(y)$ such that $y \prefto_j x$ $ \forall j \in D$ and $D \backslash\{i\} \in \mathcal{D}$.
\end{enumerate}
\end{claim}
\begin{proof}
Let $x \in X$ and $y \in \Phi^\text{or}_\mathcal{D}(x)$ be given. We prove each point by contradiction. For parts (a)-(c), take as given a winning coalition $D' \in \mathcal{D}$ such that $y \prefto_i x$ for all $i \in D'$ (existence of which is guaranteed because $\Phi^\text{or}_\mathcal{D}(x) \subseteq M^\text{w}_\mathcal{D}(x)$).

For (a), suppose $y\notin P$: then there exists a policy $z$ such that $z\succ_i y$ for every player $i$. Hence, $z \in M^\text{s}_\mathcal{D}(x)$ (as $z \succ_i y\prefto_i x$ for all $i \in D'$) and $z \succ_A y$, contradicting that $y \in \Phi^\text{or}_\mathcal{D}(x)$.

For (b), it suffices to establish the inclusion $L(y\mid x) \subseteq \text{supp}(x) \backslash \text{supp}(y)$, as the opposite inclusion is tautological. Observe that by definition, $L(y\mid x)\subseteq \text{supp}(x)$. Suppose towards a contradiction that there exists some voter $i \in L(y\mid x) \cap \text{supp}(y)$. It then follows, by definition of $\text{supp}(y)$, that  $u_i(y) > \underline{u}_i$. Transferability implies that there exists a policy $z$ such that $z \succ_j y$ for all players $j \neq i$, including $j = A$ and all $j \in D'$. Hence, $z \in M^\text{s}_\mathcal{D}(x)$ and $z \succ_A y$, contradicting that $y\in \Phi^\text{or}_\mathcal{D}(x)$.

For (c), suppose there exists a voter $i \in \text{supp}(y) \backslash \text{supp}(x)$. As $u_i(y) > \underline{u}_i$, Transferability implies that there exists some $z \in X$ such that $z \succ_j y$ for all players $j \neq i$, including $j = A$ and all $j \in D'\backslash\{i\}$. As $u_i(z) \geq u_i(x) = \underline{u}_i$ by definition, there are two cases. First, if $z\succ_i x$, then $z \in M^\text{s}_\mathcal{D}(x)$ and $z \succ_A y$, contradicting that $y\in \Phi^\text{or}_\mathcal{D}(x)$. Second, if $z \sim_i x$, then since $z \neq x$ (as $z \succ_A y \prefto_A x$), we have $z \in I_i(x) \backslash\{x\}$ and $d(z,x)>0$. For any $\epsilon \in (0, d(z,x))$,  we have $x  \notin  B_\epsilon(z)$ and hence $B_\epsilon(z)\backslash \left[ I_i(x) \backslash \{x\} \right] = B_\epsilon(z)\backslash I_i(x)$. \hyperlink{TI}{Thin Individual Indifference} then implies that, for every such $\epsilon>0$, there exists some  $z' \in B_\epsilon(z) \backslash I_i(x)$. As $I_i(x) = \{x'  \in X : u_i(x') = \underline{u}_i\}$, any such $z'$ satisfies $z' \succ_i x$. Moreover, by continuity of players' preferences, there exists a small enough $\epsilon>0$ and corresponding $z'\in B_\epsilon(z)\backslash I_i(x)$ such that $z' \succ_j y$ for all $j \neq i$. Then $z' \in M^\text{s}_\mathcal{D}(x)$ (as $z'\succ_j y \prefto_j x$ for all $j\in D'\backslash\{i\}$ and $z' \succ_i x$) and $z'  \succ_A y$, contradicting that $y\in \Phi^\text{or}_\mathcal{D}(x)$.

For (d), suppose that such $D \in \mathcal{D}$ and $i \in \text{supp}(y)$ exist. Transferability implies that there exists some $z \in X$ such that $z \succ_j y$ for all $j \neq i$, including all $j \in D \backslash\{i\}$ and $j =A$. As $D \backslash\{i\} \in \mathcal{D}$, this implies that $z \in M^\text{s}(x)$ (as $z \succ_j y \prefto_j x$ for all $j \in D \backslash\{i\}$) and $z \succ_A y$, contradicting that $y\in \Phi^\text{or}_\mathcal{D}(x)$.
\end{proof}

\medskip\noindent\textbf{Step 3: Main Argument for \Cref{Theorem-ThreeRounds}.} By virtue of \Cref{Lemma-NC-D}, the following lemma implies  \Cref{Theorem-ThreeRounds}.

\begin{lemma}\label{Lemma-DP-Final}
Suppose $\CCP$ is a \hyperlink{DistributionProblem}{Distribution Problem} satisfying \hyperlink{TI}{Thin Individual Indifference}. 
\begin{enumerate}[label={\normalfont (\alph*)}]
\item If $\mathcal{D}$ is a quota rule with quota $q <n$, then for any collection $\{\hat{\phi}_t\}_{t=1}^T$ of selections $\hat{\phi}_t(\cdot) \in \Phi^\text{or}_\mathcal{D}(\cdot)$, the following holds:
\[
\text{If $T \geq \left\lceil \frac{n}{n-q}\right\rceil$, then $\left[ \hat{\phi}_1 \circ \cdots \circ \hat{\phi}_T\right](x) \in X^*_A$ for all $x \in X$.}
\]
\item If $\mathcal{D}$ is a veto-proof voting rule, then for any collection $\{\hat{\phi}_t\}_{t=1}^T$ of selections $\hat{\phi}_t(\cdot) \in \Phi^\text{or}_\mathcal{D}(\cdot)$, the following holds:
\[
\text{If $T \geq n$, then $\left[ \hat{\phi}_1 \circ \cdots \circ \hat{\phi}_T\right](x) \in X^*_A$ for all $x \in X$.}
\]
\end{enumerate}
\end{lemma}
\begin{proof}
We begin by establishing part (a). Let the quota rule $\mathcal{D}$ with quota $q<n$, number of rounds $T \geq \lceil n/(n-q)\rceil$, and $\Phi^\text{or}_\mathcal{D}$-selections $\{\hat{\phi}\}_{t=1}^T$ be given. Let $x \in X$ be given and define $z_t \equiv \left[ \hat{\phi}_t \circ \cdots \circ \hat{\phi}_T\right](x)$ for all $t \in \{1,\dots,T\}$, with $z_{T+1} \equiv x$. We must show that $z_1 \in X^*_A$. \Cref{Lemma-PhiOR-DP}(a) implies that $z_t \in P$ for all $t$. Thus, by \Cref{Lemma-DP-Xstar}, it suffices to show that $\text{supp}(z_1) = \emptyset$. 

To that  end, we claim that for every $t\in \{1,\ldots,T\}$, 
\begin{align}\label{Equation-ThreeRounds}
   \text{supp}(z_t) \subseteq \text{supp}(z_{t+1}) \ \ \text{ and } \ \ |\text{supp}(z_{t+1}) \backslash \text{supp}(z_{t})| = \min\{n-q, |\text{supp}(z_{t+1})|\}.
\end{align}
Observe that \eqref{Equation-ThreeRounds} implies that for every $t \in \{1,\dots,T\}$, 
\begin{align*}
|\text{supp}(z_{t})| &= |\text{supp}(z_{t+1})| - \min\{n-q, |\text{supp}(z_{t+1})|\} \\
&= \max\{|\text{supp}(z_{t+1})| - (n-q), 0\} \\
&= \max\{|\text{supp}(x)| - (T+1-t)(n-q),0\},
\end{align*}
where the first and second lines are identities and the third line follows from iteratively applying the preceding lines. This implies that $|\text{supp}(z_{1})| = 0$ if and only if $T \geq \text{supp}(x)/(n-q)$. As $n \geq \text{supp}(x)$ and $T \geq \lceil n/(n-q)\rceil$ by assumption, it follows that $\text{supp}(z_1)=\emptyset$. 

Therefore, it suffices to prove \eqref{Equation-ThreeRounds}. We do so by appealing to \Cref{Lemma-PhiOR-DP}(b)-(d), noting that $z_t = \hat{\phi}_t(z_{t+1}) \in \Phi^\text{or}_\mathcal{D}(z_{t+1})$ for all $t$ by construction. First, \Cref{Lemma-PhiOR-DP}(c) directly implies that $\text{supp}(z_t) \subseteq \text{supp}(z_{t+1})$. Next, \Cref{Lemma-PhiOR-DP}(b) implies that $L(z_t\mid z_{t+1}) = \text{supp}(z_{t+1})\backslash \text{supp}(z_{t})$. If $|\text{supp}(z_{t+1})|=0$, this proves the claim. So, assume that $|\text{supp}(z_{t+1})|>0$. We assert that $|L(z_t\mid z_{t+1})| =\min\{n-q, |\text{supp}(z_{t+1})|\}$. That $|L(z_t\mid z_{t+1})| \leq |\text{supp}(z_{t+1})|$ follows from $L(z_t\mid z_{t+1}) = \text{supp}(z_{t+1})\backslash \text{supp}(z_{t})$. That $|L(z_t\mid z_{t+1})| \leq n-q$ follows from: (i) $D \in \mathcal{D}$ if and only if $|D|\geq q$ and (ii) $L(z_t\mid z_{t+1}) \subseteq N\backslash D$ for any $D \in \mathcal{D}$ such that $z_t \prefto_i z_{t+1}$ for all $i \in D$. Hence, $|L(z_t\mid z_{t+1})| \leq \min\{n-q, |\text{supp}(z_{t+1})|\}$. Suppose, towards a contradiction, that $|L(z_t\mid z_{t+1})| < \min\{n-q, |\text{supp}(z_{t+1})|\}$. Because $|L(z_t\mid z_{t+1})| < |\text{supp}(z_{t+1})|$ and $L(z_t\mid z_{t+1}) \subseteq \text{supp}(z_{t+1})$, there exists some voter $i \in \text{supp}(z_{t+1}) \backslash L(z_t \mid z_{t+1}) = \text{supp}(z_{t})$. Because $|L(z_t\mid z_{t+1})| < n-q$, the set of voter $D \equiv N \backslash L(z_t\mid z_{t+1})$ satisfies $|D| > q$, implying $D \in \mathcal{D}$ and $D\backslash\{i\} \in \mathcal{D}$. Moreover, $z_t \prefto_j z_{t+1}$ for all $j \in D$ by construction. Therefore, \Cref{Lemma-PhiOR-DP}(d) implies  the contradiction that $z_t  \notin \Phi^\text{or}_\mathcal{D}(z_{t+1})$, as desired. 

This concludes the proof of the claim and thus part (a). The proof of part (b) is very similar, so we provide only a sketch. For a general veto-proof voting rule $\mathcal{D}$, the key claim is that $\text{supp}(z_t) \subseteq \text{supp}(z_{t+1})$ and $|\text{supp}(z_{t+1}) \backslash \text{supp}(z_{t})| \geq \min\{1,\text{supp}(z_{t+1}) \}$, which implies that $n$ rounds suffices by appeals to \Cref{Lemma-DP-Xstar} and \Cref{Lemma-PhiOR-DP}(a), coupled with calculations similar to those above. That $\text{supp}(z_t) \subseteq \text{supp}(z_{t+1})$ again follows directly from \Cref{Lemma-PhiOR-DP}(c). To  show that  $|\text{supp}(z_{t+1}) \backslash \text{supp}(z_{t})| \geq \min\{1,\text{supp}(z_{t+1}) \}$, it suffices to consider the case in which $\text{supp}(z_{t+1})\neq \emptyset$. \Cref{Lemma-PhiOR-DP}(b) implies that $L(z_t\mid z_{t+1}) = \text{supp}(z_{t+1})\backslash \text{supp}(z_{t})$. Suppose towards a contradiction that $|\text{supp}(z_{t+1}) \backslash \text{supp}(z_{t})| = 0$, which implies that (i) $\text{supp}(z_t) = \text{supp}(z_{t+1}) \neq\emptyset$ and (ii) $z_t \prefto_i z_{t+1}$ for all voters $i \in N \in \mathcal{D}$. By (i), there exists some $k \in \text{supp}(z_t)$. By (ii) and that there exists some $D \in \mathcal{D}$ with $i \notin D$ (as $\mathcal{D}$ is veto-proof), \Cref{Lemma-PhiOR-DP}(d) implies the desired contradiction $z_t  \notin \Phi^\text{or}_\mathcal{D}(z_{t+1})$, proving part (b).
\end{proof}

\subsection{Proof of \Cref{Theorem-Horizon} on p. \pageref{Theorem-Horizon}}

For a subset of policies $Y \subseteq X$, the following definitions are standard: $Y$ is \emph{internally stable} if there do not exist distinct $x,y \in Y$ such that $y \succ_M x$ and $y \succ_A x$. $Y$ is \emph{externally stable} if, for every $x \notin Y$, there exists some $y \in Y$ such that $y \succ_M x$ and $y \succ_A x$. $Y$ is \emph{stable} if it is both internally and externally stable. As shown by \cite{diermeier2012characterization}, there exists a unique stable set in the present setting, which we denote by $V$. Recall that $E = \{x \in X : x=\phi(x) \}$ denotes the set of unimprovable policies. Observe that $E \subseteq V$ since excluding any unimprovable policy would contradict the external stability of $V$.

Recall from \Cref{subsection:finite} that $M(x) := \{y \in X : y \succ_M x \text{ or } y=x\}$. Define the agenda setter's \emph{favorite stable improvement} $\psi(\cdot; V) :  X \to V$ by
\begin{equation}
 \{\psi(x; V)\} := \argmax_{y \in M(x) \bigcap V } u_A(y)\label{eqn:psi-def}
\end{equation}
which is well-defined because $V$ is externally stable. By definition, for every policy $x$, $\phi(x)\prefto_A \psi(x;V)$ and $\phi(x) = \psi(x;V)$ if and only if $\phi(x) \in V$. 

We recall the following characterizations of equilibrium outcomes and payoffs:
\begin{itemize}[noitemsep]
    \item \cref{lemma:structural} shows that in the $T$-round game with $T<\infty$, an initial default $x^0$ leads to the unique equilibrium outcome $\phi^T(x^0)$. The agenda setter's equilibrium payoff is denoted $U_T(x^0) \equiv u_A(\phi^T(x^0))$.
    \item \citet{diermeier2012characterization} and \citet{anesi2014bargaining} show that in the $T$-round game with $T=\infty$, an initial default $x^0$ leads to the unique MPE outcome $\psi(x^0; V)$.\footnote{Specifically, Propositions 1 and 2 in \citet{anesi2014bargaining}, specialized to the present setting with a single proposer and \hyperlink{GFA}{Generic Finite Alternatives}, imply the above characterization for \emph{some} stable set; as noted above, Lemmas 1-3 in \citet{diermeier2012characterization} show that the stable set $V$ exists and is unique in the present setting. Theorem 1 in \citet{diermeier2012characterization} provides an analogous characterization of MPE outcomes in the context of \posscite{diermeier2011legislative} infinite-horizon model (with discounting and no termination rule).} The agenda setter's equilibrium payoff is denoted $U_\infty(x^0) \equiv u_A(\psi(x^0; V))$.
\end{itemize}
Observe that these characterizations, and the fact that $\phi^{t+1}(\cdot) \prefto_A \phi^t(\cdot)$ for all $t$, immediately yields the initial claim in \Cref{Theorem-Horizon}. They also permit us to equivalently rephrase statements (a) and (b) from \Cref{Theorem-Horizon} as:
\begin{itemize}[noitemsep]
    \item[(a)] \emph{There exist $x^0 \in X$ such that $\phi^2 (x^0) \succ_A \phi(x^0) \succ_A \psi(x^0; V)$.}
    \item[(b)] \emph{For all $x^0 \in X$, $\phi^2 (x^0) = \phi (x^0) =   \psi(x^0; V)$.}
\end{itemize}
Thus, it suffices to show that statements (a) and (b) above are mutually exclusive and exhaustive. Mutual exclusivity is obvious; the argument below establishes exhaustiveness. 

Let $R := \{x \in X : \phi(x) \in E\}$ denote the set of \emph{at-most-once-improvable} policies. These are the default policies at which the agenda setter does not benefit from having more than a single round of proposals in the finite-horizon game because she obtains $\phi(x)$ with a single round, and $\phi(x)$ is unimprovable; by contrast, if $x\notin R$, the agenda setter strictly prefers having two or more rounds in the finite horizon to a single round.

We show that the following identity holds:
\begin{align}\label{Equation-DefinitionR}
    R=\{x\in X:\phi(x)=\psi(x;V)\text{ and }\phi^2(x)=\psi(\phi(x);V)\}.
\end{align}
To see why \eqref{Equation-DefinitionR} is true, suppose that $x$ is an element of the set on the RHS. Observe that $\phi(x)=\psi(x;V)$ implies that $\phi(x)\in V$. Because $V$ is internally stable, it then follows that $\psi(\phi(x);V)=\phi(x)$. Therefore, $\phi^2(x)=\psi(\phi(x);V)=\phi(x)$, which implies that $\phi(x)\in E$, and therefore, $x\in R$. Proceeding in the other direction, suppose $x\in R$. By definition of $R$, $\phi(x)\in E$, which implies that $\phi^2(x)=\phi(x)$. Because $E\subseteq V$, it also follows that $\phi(x)\in V$ and therefore, $\phi(x)=\psi(x;V)$. As $V$ is internally stable, $\psi(\phi(x);V)=\phi(x)$. Therefore,  $\phi^2(x)=\psi(\phi(x);V)$.

\Cref{Equation-DefinitionR} is at the core of our argument and we use it to show that Statements (a) and (b) reduce to the two exhaustive cases of $R\subsetneq X$ and $R=X$. Because of its relative simplicity, we begin with the latter. 

If $R=X$, then it follows from the definition of $R$ that for every policy $x$, $\phi^2(x)=\phi(x)$, as $\phi(x)$ is unimprovable; it also follows from \eqref{Equation-DefinitionR} that $\phi(x)=\psi(x;V)$. Therefore, (b) necessarily holds. 

If $R\subsetneq X$, then there is a policy $y\notin R$. It follows from the definition of $R$ that $\phi(y)\notin E$, and hence $\phi^2(y)\succ_A\phi(y)$. Moreover, \eqref{Equation-DefinitionR} establishes that either $\phi(y)\succ_A \psi(y;V)$ or $\phi^2(y)\succ_A \psi(\phi(y);V)$. In the former case, Statement (a) is established for $x^0=y$. In the latter case, \eqref{Equation-DefinitionR} implies that $\phi(y)\notin R$, and therefore, $\phi^2(y)\notin E$. Hence, $\phi^3(y)\succ_A \phi^2(y)\succ_A\psi(\phi(y);V)$. Statement (a) is now established for $x^0=\phi(y)$.
\end{document}